\documentclass[10pt]{article}

\usepackage[leqno]{amsmath}
\usepackage{amssymb}
\usepackage{amsthm}
\usepackage{enumerate}
\usepackage{verbatim}
\usepackage{mathrsfs}
\usepackage{dsfont}
\usepackage{graphicx}
\usepackage{tikz}
\usepackage[pdf]{pstricks}
\usepackage{float}
\usepackage{caption}
\usepackage{subcaption}

\usepackage{soul}

\newcommand{\tu}{\textup}
\newcommand{\ds}{\displaystyle}
\newcommand{\ts}{\textstyle}

\newcommand{\mr}{\mathrm}

\newcommand{\mc}{\mathcal}
\newcommand{\mb}{\mathbf}

\newcommand{\scr}{\mathscr}

\newcommand{\uhr}{\upharpoonright}

\newcommand{\sq}[1]{\ensuremath{\langle#1\rangle}}
\newcommand{\sem}[1]{\ensuremath{[\hspace*{-1.4pt}[#1]\hspace*{-1.4pt}]}}

\newcommand{\eps}{\varepsilon}
\newcommand{\wt}{\widetilde}

\newcommand{\BAR}[1]{\overline{#1}}

\let\Ex\undefined
\let\Pr\undefined
\DeclareMathOperator*{\Ex}{\mathds{E}}
\DeclareMathOperator*{\Pr}{\mathds{P}}

\DeclareMathOperator*{\bigbowtie}{\raisebox{-1.5pt}{\scalebox{1.6}{$\bowtie$}}\vphantom{\big|}}

\newcommand{\Supp}{\mr{Supp}}

\usepackage{mathtools}
\newcommand{\defeq}{\vcentcolon=}

\newcommand{\fieldfont}[1]{\mathbb{#1}}
\newcommand{\N}{\fieldfont{N}}
\newcommand{\R}{\fieldfont{R}}
\newcommand{\Z}{\fieldfont{Z}}

\newtheoremstyle{theorem-style}
  {}
  {}
  {\slshape}
  {}
  {\bf}
  {.}
  {.5em}
  {}

\newtheorem{thm}{Theorem}[section]

\newtheorem{prop}[thm]{Proposition}
\newtheorem{la}[thm]{Lemma}

\newtheorem{main-la}[thm]{Main Lemma}

\newtheorem{cor}[thm]{Corollary}

\newtheorem{obs}[thm]{Observation}

\theoremstyle{definition}
\newtheorem{df}[thm]{Definition}

\newtheorem{proviso}[thm]{Proviso}
\newtheorem{rmk}[thm]{Remark}
\newtheorem{ex}[thm]{Example}
\newtheorem{notn}[thm]{Notation}

\usepackage[lf]{Baskervaldx}

\newcommand{\cc}{\mathit}
\newcommand{\formulafont}[1]{\tu{\texttt{#1}}}
\newcommand{\ff}{\formulafont}

\usepackage[bottom]{footmisc}

\usepackage{thm-restate}
\usepackage{fullpage}
\usepackage{hyperref} 

\hypersetup{
    linktocpage=true,
    colorlinks=true,
    linkcolor=teal,
    citecolor=magenta,
    }

\usepackage{float}

\usepackage{esvect}

\usepackage{stmaryrd}

\usepackage{bm}
\renewcommand{\mb}{\bm}

\renewcommand{\mathsf}[1]{\textup{\normalfont{\textsf{#1}}}}

\usetikzlibrary{shapes.geometric}
\usetikzlibrary{arrows.meta}

\newcommand{\joinop}{\sq{\hspace{.06em}}}
\newcommand{\un}[2]{\sq{#1,#2}}
\newcommand{\unp}[2]{\un{#1}{#2}}

\newcommand{\XI}{{\mb\Xi}}

\newcommand{\depth}{\mr{depth}}
\newcommand{\sz}{\mr{size}}

\newcommand{\strict}[1]{\formulafont{strict}(#1)}
\newcommand{\strictjoin}[1]{\mathit{strict}(#1)}
 
\newcommand{\SuppTree}{S}
 
\def\[#1\]{\begin{align*}#1\end{align*}}

\newcommand{\Binomial}{\mathrm{Binomial}}
\newcommand{\Bernoulli}{\mathrm{Bernoulli}}

\newcommand{\Nbd}{\mathsf{Nbd}}

\newcommand{\Exp}{\mr{e}}

\newcommand{\Path}{\mathsf{Path}}

\newcommand{\Le}{\mr{L}}
\newcommand{\Ri}{\mr{R}}

\newcommand{\CNF}{\formulafont{C}}
\newcommand{\DNF}{\formulafont{D}}

\newcommand{\sqq}[1]{\sq{\!\sq{#1}\!}}

\newcommand{\A}{\scr A}
\newcommand{\B}{\scr B}
\newcommand{\C}{\scr C}

\newcommand{\plussym}{+}
\newcommand{\minussym}{--}

\newcommand{\ACzero}{\cc{AC^0}}

\newcommand{\SACzero}{\cc{SAC^0}}

\renewcommand{\L}{\cc{L}}
\newcommand{\NL}{\cc{NL}}

\newcommand{\NCone}{\cc{NC^1}}
\newcommand{\TCzero}{\cc{TC^0}}
\newcommand{\ACCzero}{\cc{ACC^0}}

\renewcommand{\sqq}[1]{\langle\!\hspace{-.7pt}\langle#1\rangle\!\hspace{-.7pt}\rangle}

\renewcommand{\sem}[1]{\llbracket#1\rrbracket}
\newcommand{\semempty}{\sem{\hspace{.06em}}}

\newcommand{\n}{\tilde{n}}

\newcommand{\IMM}{\tu{\textsc{imm}}}
\newcommand{\BMM}{\tu{\textsc{bmm}}}

\newcommand{\PMM}{\tu{\textsc{pmm}}}
\newcommand{\SPMM}{\tu{\textsc{sub-pmm}}}

\renewcommand{\Supp}{\mathsf{Supp}}

\renewcommand{\ul}{\underline}

\newcommand{\nosmash}[1]{#1}

\newcommand{\f}{\formulafont{f}}
\newcommand{\g}{\formulafont{g}}

\begin{document}

\title{Formula Size-Depth Tradeoffs for\\Iterated Sub-Permutation Matrix Multiplication}

\author{Benjamin Rossman\\Duke University}
\date{\today}

\maketitle{}

\begin{abstract}\normalsize
We study the formula complexity of {\em Iterated Sub-Permutation Matrix Multiplication}, the logspace-complete problem of computing the product of $k$ $n$-by-$n$ Boolean matrices\vspace{1pt} with at most a single $1$ in each row and column.  For all $d \le \log k$, this problem is solvable by $n^{O(dk^{\smash{1/d}})}$ size monotone formulas of two distinct\vspace{1pt} types: (unbounded \mbox{fan-in}) $AC^0$ formulas of depth $d+1$ and (semi-unbounded fan-in) $SAC^0$ formulas of $\bigwedge$-depth $d$ and $\bigwedge$-fan-in $k^{1/d}$.  The results of this paper give\vspace{1pt}
\begin{itemize}
  \item
    matching ${n^{\Omega(dk^{\smash{1/d}})}}$ lower bounds for monotone $\ACzero$ and $\SACzero$ formulas for all $k \le \log\log n$, 
    as well as\vspace{1pt}
  \item
    slightly weaker $\smash{n^{\Omega(dk^{\smash{1/2d}})}}$ lower bounds for non-monotone $\smash{\ACzero}$ and $\smash{\SACzero}$ formulas.\vspace{1pt}
\end{itemize}
These size-depth tradeoffs converge at $d = \log k$ to tight $n^{\Omega(\log k)}$ lower bounds for both {unbounded-depth monotone formulas} \cite{rossman2015correlation} and {bounded-depth non-monotone formulas} \cite{rossman2018formulas}.  Our non-monotone lower bounds extend to the more restricted {\em Iterated Permutation Matrix Multiplication}\vspace{1pt} problem, improving the previous $n^{k^{\smash{1/\exp(O(d))}}}$ tradeoff for this problem \cite{beame1998improved}.
\end{abstract}

\newpage

\tableofcontents{}

\newpage

\section{Introduction}

Motivated by the question of $\NCone$ vs.\ $\cc{L}$ (polynomial-size formulas vs.\ branching programs), we study the formula complexity of {\em Iterated Sub-Permutation Matrix Multiplication} ($\SPMM_{n,k}$), 
the logspace-complete problem of computing the product of $k$ $n$-by-$n$ sub-permutation matrices. 
Our results give asymptotically tight size-depth tradeoffs for monotone (unbounded fan-in) $\ACzero$ and (semi-unbounded fan-in) $\SACzero$ formulas solving this problem, as well as slightly weaker tradeoffs for non-monotone formulas.

Our lower bounds are based on a method of 
reducing $\ACzero$ formulas for $\SPMM_{n,k}$ 
to
simpler formulas called {\em join trees}, which construct the length-$k$ path graph from its single-edge subgraphs via the union operation.
A large part of this paper (\S\ref{sec:tradeoffs-for-join-trees}--\ref{sec:tradeoff2}) 
is devoted to the combinatorial theory of join trees. We establish tradeoffs between a few ``size'' and ``depth'' measures on join trees, which we are able to lift to size-depth tradeoffs for $\ACzero$ and $\SACzero$ formulas by extending the Pathset Framework introduced in papers \cite{rossman2015correlation,rossman2018formulas}.

The rest of this introduction is organized as follows. \S\ref{sec:IMM} introduces iterated matrix multiplication problems $\IMM_{n,k}$, $\BMM_{n,k}$, $\PMM_{n,k}$ and $\SPMM_{n,k}$.
\S\ref{sec:NC1} discusses the formula complexity of these problems in relation to the $\NCone$ vs.\ $\L$ question. 
\S\ref{sec:upper-bounds} reviews some upper bounds on the $\ACzero$ and $\SACzero$ formula size of $\SPMM_{n,k}$. \S\ref{sec:results} states our main results.
\S\ref{sec:related} discusses prior related work.
\S\ref{sec:outline} briefly outlines the rest of the paper.

\subsection{Iterated matrix multiplication problems}\label{sec:IMM}

{\em Iterated Matrix Multiplication} 
($\IMM_{n,k}$) is the task, given $n \times n$ matrices $M^{(1)},\dots,M^{(k)}$ (over some specified field), of computing the $1,1$-entry of
the matrix product 
$M^{(1)}\cdots M^{(k)}$. 
In algebraic complexity theory, 
this problem is 
represented by 
a formal polynomial of $kn^2$ variables (the entries of matrices $M^{(i)}$):
\[
  \IMM_{n,k}(\vec M) 
  =
  \sum_{a_1,\dots,a_{k-1} \,\in\, [n]}\ M^{(1)}_{\smash 1,a_1}  M^{(2)}_{a_1,a_2}  \cdots  M^{(k)}_{a_{k-1},\smash 1}.
\]
{\em Iterated Boolean Matrix Multiplication} ($\BMM_{n,k}$) is 
the corresponding problem for matrices with entries in the Boolean ring $\{0,1\}$ with addition $\vee$ and multiplication $\wedge$. This problem is represented by a monotone Boolean function $\{0,1\}^{kn^2} \to \{0,1\}$ described by the formula
\[
  \BMM_{n,k}(\vec M) 
  =
  \bigvee_{a_1,\dots,a_{k-1} \,\in\, [n]}\ M^{(1)}_{\smash 1,a_1} \wedge M^{(2)}_{a_1,a_2} \wedge \dots \wedge M^{(k)}_{a_{k-1},\smash 1}.
\]

A special case of this problem arises when the input $M^{(1)},\dots,M^{(k)}$ are promised to be {\em permutation matrices}, or more generally {\em sub-permutation matrices} with at most a single $1$ in each row and column. (For such matrices, $\BMM_{n,k}$ coincides with the function computed by $\IMM_{n,k}$ over any field.)
These restricted promise problems are called
{\em Iterated (Sub-)Permutation Matrix Multiplication} and denoted by $\PMM_{n,k}$ and $\SPMM_{n,k}$.
Note that there is a significant difference in computational complexity: whereas $\BMM_{n,k}$ and $\IMM_{n,k}$ are complete for nondeterministic logspace $\NL$ and its algebraic analogue $\cc{VBP}$ (polynomial-size algebraic branching programs), both $\PMM_{n,k}$ and $\SPMM_{n,k}$ are complete for deterministic logspace $\L$ \cite{cook1987problems}.

This paper focuses on the problem $\SPMM_{n,k}$ rather than $\PMM_{n,k}$, since like $\BMM_{n,k}$ it is naturally represented by a {\em monotone} Boolean function (whose domain is the downward-closed subset of $\smash{\{0,1\}^{kn^2}}$ representing $k$-tuples of sub-permutation matrices).
It is natural to study the complexity of $\SPMM_{n,k}$ in both monotone and non-monotone models of computation.
In contrast, $\PMM_{n,k}$ is an inherently non-monotone ``slice function'' whose relevant inputs ($k$-tuples of permutation matrices) have equal Hamming weight $kn$; for such functions, it is well-known that monotone and non-monotone formula complexity coincide up to a polynomial factor \cite{berkowitz1982some}.

\subsection{Formula complexity and the question of $\NCone$ vs.\ $\cc{L}$}\label{sec:NC1}

A {\em formula} is a rooted tree whose leaves (``inputs'') are labeled by constants ($0$ or $1$) or literals ($X_i$ or $\BAR X_i$) and whose non-leaves (``gates'') are labeled by $\bigwedge$ or $\bigvee$. 
Understanding the formula complexity of iterated matrix multiplication problems is a major challenge at the frontier of lower bounds 
in Circuit Complexity.
Recall that $\NCone$ is --- by one of its equivalent definitions --- the complexity class of languages that are decidable by a sequence of polynomial-size formulas. It is known that $\ACzero \subset \ACzero[\oplus] \subset \NCone \subseteq \L \subseteq \NL \subseteq \cc{P} \subseteq \cc{NP}$, but open whether $\NCone$ is a proper subclass of $\L$, or even $\cc{NP}$.
At present, 
no lower bound better than $n^{3-o(1)}$ is known on the minimum formula size of any {\em explicit} 
sequence of Boolean functions (describing a language in $\cc{P}$, or even $\cc{NP}$) \cite{tal2014shrinkage}.

Historically, the first lower bounds against many weak circuit classes have targeted functions of barely higher complexity. Consider the classical results that $\textsc{parity} \in \ACzero[\oplus] \setminus \ACzero$ \cite{Ajtai83,FSS84} and $\textsc{maj} \in \NCone \setminus \ACzero[\oplus]$ \cite{Razborov87,Smolensky87}. Even the best $n^{3-o(1)}$ formula lower bound \cite{hastad1998shrinkage,tal2014shrinkage} targets a function (Andreev's function) with $O(n^3)$ size formulas.
From this perspective, {\em Iterated Sub-Permutation Matrix Multiplication} $\SPMM_{n,k}$ is the ideal target for a super-polynomial lower bound against $\NCone$. 
The smallest known formulas for $\SPMM_{n,k}$ have size $n^{O(\log k)}$, which is widely believed to be optimal.
A matching $n^{\Omega(\log k)}$ lower bound on the formula size of $\SPMM_{n,k}$ (respectively: $\BMM_{n,k}$, $\IMM_{n,k}$) for any jointly super-constant range of $k$ and $n$ would separate $\NCone$ from $\cc{L}$ (respectively: $\NCone$ from $\cc{NL}$, $\cc{VNC^1}$ from $\cc{VBP}$). 

Motivated by this goal, previous work of the author \cite{rossman2015correlation,rossman2018formulas} 
took aim at this lower bound challenge in two restricted classes of formulas:
\begin{enumerate}[\quad]
\item
\tu{(\plussym)}\ \ 
monotone unbounded-depth formulas, and
\item
\tu{(\minussym)}\ \ 
non-monotone bounded-depth formulas.
\end{enumerate}
Papers \cite{rossman2015correlation,rossman2018formulas} established tight $n^{\Omega(k)}$ lower bounds in settings (\plussym) and (\minussym) with respect to an average-case version of $\BMM_{n,k}$ (on random matrices with i.i.d.\ $\Bernoulli(\frac1n)$ entries).
The following theorem restates these results as lower bounds for problems $\SPMM_{n,k}$ and $\PMM_{n,k}$, which follow from randomized $\ACzero$ reductions from average-case $\BMM_{n,k}$.

\begin{thm}[\cite{rossman2015correlation,rossman2018formulas}]\label{thm:lbs}
For all $k \le \log\log n$, size $n^{\Omega(\log k)}$ is required by both
\begin{enumerate}[\quad]
\item
\tu{(\plussym)}\ \  
monotone formulas solving $\SPMM_{n,k}$,
\item
\tu{(\minussym)}\ \  
formulas of depth $\frac{\log n}{(\log\log n)^{O(1)}}$ solving $\PMM_{n,k}$.
\end{enumerate}
\end{thm}
Both of these lower bounds come ``close'' to showing $\NCone \neq \L$.  To separate these complexity classes, it suffices to extend (\plussym) from the monotone function $\SPMM_{n,k}$ to the slice function $\PMM_{n,k}$.  It also suffices to extend (\minussym) from 
depth $\frac{\log n}{(\log\log n)^{O(1)}}$ to depth $\log n$.\vspace{-1pt}

\subsection{Upper bounds}\label{sec:upper-bounds}

Lower bounds (\plussym) and (\minussym) of Theorem \ref{thm:lbs} are simultaneously shown to be tight 
by bounded-depth monotone formulas implementing the standard divide-and-conquer (a.k.a.\ recursive doubling) technique of Savitch \cite{savitch1970relationships}.

\begin{prop}\label{prop:upper0}
$\SPMM_{n,k}$ is solvable by monotone formulas of size $n^{O(\log k)}$ and depth $O(\log k)$ for all $n$ and $k$.
\end{prop}

This raises the interesting question: what about formulas of very small depths $2,3,4,\dots$ below $O(\log k)$?
Here the best known upper bounds exhibit a typical size-depth tradeoff. 
In fact, there are two different tradeoffs given by distinct families of (unbounded fan-in) $\ACzero$ formulas and (semi-unbounded fan-in) $\SACzero$ formulas.\footnote{See \S\ref{sec:formulas} for definitions of classes $\Sigma_d$ and $\Pi_d$ of depth-$d$ $\ACzero$ formulas.
Consistently throughout this paper, ``$\SACzero$'' refers to $\ACzero$ formulas where $\bigwedge$-gates have fan-in $O(k^{1/d})$ in the context of {\em upper bounds} and $O(n^{1/k})$ in {\em lower bounds}. 
Note that $k^{1/d} \ll n^{1/k}$ for the range of parameters we consider.}

\begin{prop}\label{prop:upper}
When $k^{1/d}$ is an integer, 
$\SPMM_{n,k}$ is solvable by size $kn^{dk^{1/d}}$ monotone formulas of the following types:\vspace{-2.5pt}
\begin{enumerate}[\quad\normalfont(I)]
  \item
    $\Sigma_{2d}$ formulas with $\ts\bigwedge$-fan-in $k^{1/d}$ \tu(and $\ts\bigvee$-fan-in $n^{k^{\smash{1/d}}}$\tu),\vspace{1pt}
  \item
    $\Sigma_{d+1}$ as well as $\Pi_{d+1}$ formulas \tu(with fan-in $k^{1/d} n^{k^{\smash{1/d}}}$\tu).\vspace{5pt}
\end{enumerate}
\end{prop} 

To illustrate constructions (I) and (II), Examples \ref{ex:d1} and \ref{ex:d2} below concretely present formulas (I)$_{\Sigma_{2d}}$, (II)$_{\Sigma_{d+1}}$ and (II)$_{\Pi_{d+1}}$ in the cases $(d,k) = (1,5)$ and $(d,k) = (2,5^2)$.

\begin{ex}[The case $d=1$ and $k=5$ of Proposition \ref{prop:upper}]\label{ex:d1}
For sub-permutation matrices $M^{(1)},\dots,M^{(5)} \in \{0,1\}^{n\times n}$ and indices $a_0,a_5 \in [n]$, the $a_0,a_5$-entry of the product $M^{(1)}\cdots M^{(5)}$ is computed by the following monotone $\Sigma_2$ and $\Pi_2$ (a.k.a.\ DNF and CNF) formulas:
\begin{align}
\tag{$\DNF_{a_0,a_5}$}
  \bigvee_{a_1,a_2,a_3,a_4}\:
  &\Big(\ 
  M^{(1)}_{a_0,a_1} \wedge M^{(2)}_{a_1,a_2} \wedge M^{(3)}_{a_2,a_3} \wedge M^{(4)}_{a_3,a_4} \wedge M^{(5)}_{a_4,a_5}
  \ \Big),\\
\tag{$\CNF_{a_0,a_5}$}
  \bigwedge_{a_1,a_2,a_3,a_4}\:
  &\Big(\ 
  \bigvee_{b_1 \ne a_1} M^{(1)}_{a_0,\smash{b_1}}
  \vee
  \bigvee_{b_2 \ne a_2} M^{(2)}_{a_1,\smash{b_2}}
  \vee
  \bigvee_{b_3 \ne a_3} M^{(3)}_{a_2,\smash{b_3}}
  \vee
  \bigvee_{b_4 \ne a_4} M^{(4)}_{a_3,\smash{b_4}}
  \vee
  M^{(5)}_{a_4,a_5}
  \ \Big).
\end{align}
Note that $\DNF_{a_0,a_5}$ computes $(M^{(1)}\cdots M^{(5)})_{a_0,a_5}$ 
for arbitrary Boolean matrices $M^{(1)},\dots,M^{(5)} \in \{0,1\}^{n \times n}$, whereas $\CNF_{a_0,a_5}$ relies on the assumption that $M^{(1)},\dots,M^{(5)}$ that have at most a single $1$ in each row.
Setting $a_0=a_5=1$, monotone $\Sigma_2$ and $\Pi_2$ formulas $\DNF_{1,1}$ and $\CNF_{1,1}$ compute the function $\SPMM_{n,5}$. This provide the case $d=1$ and $k=5$ of Proposition~\ref{prop:upper}: namely, $\DNF_{1,1}$ furnishes upper bounds (I)$_{\Sigma_2}$ and (II)$_{\Sigma_2}$, while $\CNF_{1,1}$ furnishes (II)$_{\Pi_2}$.
\end{ex}

\begin{ex}[The case $d=2$ and $k=25$ of Proposition \ref{prop:upper}]\label{ex:d2}
$\SPMM_{n,25}$ is computed by the following monotone $\Sigma_4$, $\Sigma_3$ and $\Pi_3$ formulas, built from sub-formulas $\DNF_{a_{5(i-1)},a_{5i}}$ and $\CNF_{a_{5(i-1)},a_{5i}}$ over matrices $M^{(5i+1)},\dots,M^{(5i)}$:
\begin{align}
\label{eq:Sigma4}\tag*{$\tu{(I)}_{\Sigma_4}$}
  \bigvee_{a_5,a_{10},a_{15},a_{20}}\:
  &\Big(\ 
  \DNF_{1
  ,a_5} \wedge \DNF_{a_5,a_{10}} \wedge \DNF_{a_{10},a_{15}} \wedge \DNF_{a_{15},a_{20}} \wedge \DNF_{a_{20},1
  }\ 
  \Big),\\
\label{eq:Sigma3}\tag*{$\tu{(II)}_{\Sigma_3}$}
  \bigvee_{a_5,a_{10},a_{15},a_{20}}\:
  &\Big(\ 
  \CNF_{1
  ,a_5} \wedge \CNF_{a_5,a_{10}} \wedge \CNF_{a_{10},a_{15}} \wedge \CNF_{a_{15},a_{20}} \wedge \CNF_{a_{20},1
  }\ 
  \Big),\\
\label{eq:Pi3}\tag*{$\tu{(II)}_{\Pi_3}$}
  \bigwedge_{a_5,a_{10},a_{15},a_{20}}\:
  &\Big(\ 
  \bigvee_{b_5 \ne a_5} \DNF_{1
  ,b_5}\ 
  \vee 
  \bigvee_{b_{10} \ne a_{10}} \DNF_{a_5,b_{10}}
  \vee
  \bigvee_{b_{15} \ne a_{15}} \DNF_{a_{10},b_{15}}
  \vee
  \bigvee_{b_{20} \ne a_{20}} \DNF_{a_{15},b_{20}}
  \vee
  \DNF_{a_{20},1
  }
  \ \Big).
\end{align}
Note that (I)$_{\Sigma_4}$ has $\bigwedge$-fan-in $5$ ($=k^{1/d}$) and all three formulas have size at most $25n^{10}$ ($= kn^{dk^{1/d}}$) as required.
A similar recurrence produces formulas (I)${}_{\Sigma_{2d}}$, (II)${}_{\Sigma_{d+1}}$ and (II)${}_{\Pi_{d+1}}$ for all $d \ge 3$.
\end{ex}

Formulas (I)$_{\Sigma_{2d}}$ (and their algebraic counterparts with $\sum,\prod$ replacing $\bigvee,\bigwedge$) implement a standard divide-and-conquer algorithm for problems $\BMM_{n,k}$ and $\IMM_{n,k}$. This algorithm recursively partitions the matrix product ${M^{(1)}\cdots M^{(k)}}$ as $N^{(1)}\cdots N^{(k^{(d-1)/d})}$ where $N^{(i)}$ is the sub-product $M^{((i-1)k^{1/d}+1)}\cdots M^{(ik^{1/d})}$.
Formulas (II)${}_{\Sigma_{d+1}}$ and (II)${}_{\Pi_{d+1}}$ for $\SPMM_{n,k}$ utilize a similar recurrence, which generalizes the well-known $\Sigma_{d+1}$ and $\Pi_{d+1}$ formulas for $\textsc{parity}_k$ (which are essentially the $n=2$ case of formulas (II)).

Notice that Proposition \ref{prop:upper} implies the $n^{O(\log k)}$ upper bound of Proposition \ref{prop:upper0} when $d = \log k$.  In contrast, we remark taht the \ul{circuit} versions of constructions (I) and (II) have size $O(kn^{\smash{k^{1/d}}})$, which shrinks to merely polynomial $O(kn^2)$ at $d = \log k$.

\subsection{Our tradeoffs for $\ACzero$ and $\SACzero$ formulas}\label{sec:results}

Lower bounds (\plussym) and (\minussym) of Theorem \ref{thm:lbs} were proved using a technique, known as the Pathset Framework, which reduces formulas computing $\SPMM_{n,k}$
to simpler combinatorial objects called {\em join trees} (binary trees of $\cup$-gates which construct the length-$k$ path graph from its single-edge subgraphs). 
This reduction provides a means of lifting lower bounds for join trees to lower bounds for formulas computing $\SPMM_{n,k}$.

A curious feature of lower bound (\minussym) for $\ACzero$ formulas in \cite{rossman2018formulas} is that the reduction to join trees starts by converting $\ACzero$ formulas an equivalent DeMorgan formulas with fan-in $2$.  
The resulting $n^{\Omega(\log k)}$ lower bound applies to $\ACzero$ formulas up to depth $\frac{\log n}{(\log\log n)^{O(1)}}$\vspace{-1pt}.
Significantly, the depth $d$ does not appear as a parameter in the lower bound, in contrast to typical $\ACzero$ lower bounds that use switching lemmas or the polynomial approximation method in bottom-up fashion.
This is a powerful feature of the result, which is tight above depth $O(\log k)$. 
However, the lower bound technique does not (in any obvious way) yield lower bounds stronger than $n^{\Omega(\log k)}$ for formulas of depths $2,3,4,\dots$ below $O(\log k)$.

The results of this paper expand the Pathset Framework of \cite{rossman2015correlation,rossman2018formulas} in order to extend
lower bounds 
(\plussym) and (\minussym) 
of Theorem \ref{thm:lbs} to size-depth tradeoffs for both $\ACzero$ and $\SACzero$ formulas.
We obtain tradeoffs for {monotone} formulas which match the upper bounds of Proposition \ref{prop:upper}.
We are able to extend these tradeoffs to {non-monotone} formulas, at the expense of a slightly weaker lower bound and smaller range of $k$.

\begin{thm}[Size-depth tradeoffs for $\ACzero$ and $\SACzero$ formulas computing $\SPMM_{n,k}$]\label{thm:AC0tradeoffs} 
For all $k \le \log \log n$ and $d \le \log k$, size $n^{\Omega({dk^{1/d}})}$ is required by both
\begin{enumerate}[\quad]
  \item
    \tu{\phantom{I}(I)$^+$}\,\ 
    monotone
    $\SACzero$ formulas of 
    $\smash\bigwedge$-depth $d$
    and 
    $\smash\bigwedge$-fan-in $n^{1/k}$, and
  \item
    \tu{(II)$^+$}\,\  
    monotone $\ACzero$ formulas of 
    depth $d+1$.
\end{enumerate}
In the non-monotone setting, for all $k \le \log^\ast n$ and $d \le \log k$, size $n^{\Omega(\smash{dk^{1/2d}})}$ is required by both
\begin{enumerate}[\quad]
  \item
    \tu{\phantom{I}(I)$^-$}\,\ 
    $\SACzero$ formulas of 
    $\smash\bigwedge$-depth $d$
    and 
    $\smash\bigwedge$-fan-in $n^{1/k}$, and
  \item
    \tu{(II)$^-$}\,\  
    $\ACzero$ formulas of 
    depth $d+1$.
\end{enumerate}
\end{thm}

Note that the maximum $\bigwedge$-fan-in $n^{1/k}$ in $\SACzero$ lower bounds (I)$^+$ and (I)$^-$ exceeds the $\bigwedge$-fan-in $k^{1/d}$ of formulas (I)$_{\Sigma_{2d}}$ of Proposition \ref{prop:upper}.
Monotone lower bounds (I)$^+$ and (II)$^+$ are thus tight up to the 
constant in the exponent of $n$,
in light of the upper bounds given by Proposition \ref{prop:upper}.
Also observe that $\ACzero$ tradeoffs (II)$^+$ and (II)$^-$ converge at $d = \log k$ to the tight $n^{\Omega(\log k)}$ lower bounds (\plussym) and (\minussym) of Theorem \ref{thm:lbs}.

Below we remark on a few corollaries and extensions of Theorem \ref{thm:AC0tradeoffs}.

\paragraph{Tradeoffs for circuits.} \ul{Formula} lower bounds $n^{\Omega(\smash{dk^{1/d}})}$ and $n^{\Omega(dk^{1/2d})}$ of Theorem \ref{thm:AC0tradeoffs} directly imply $n^{\Omega(\smash{k^{1/d}})}$ and $n^{\Omega(\smash{k^{1/2d}})}$ lower bounds for the corresponding $\cc{(S)}\ACzero$ \ul{circuits}. This follows from the observation that every size $s$ $\ACzero$ circuit of depth $d+1$, or $\SACzero$ circuit of $\bigwedge$-depth $d$ and $\bigwedge$-fan-in $s^{O(1/d)}$, expands to an $\cc{(S)}\ACzero$ formula with size at most $s^{d + O(1)}$. 

\paragraph{The range of $k$.}
Examination of the proof in \S\ref{sec:tradeoffs-for-formulas} shows that range $k \le \log^\ast n$ of non-monotone lower bounds (I)$^-$ and (II)$^-$ can be improved to at least $k \le 0.99\log\log\log\log n$ (see Remark~\ref{rmk:range}). The statement of Theorem \ref{thm:AC0tradeoffs} confines the range to $k \le \log^\ast n$ for simplicity sake and to emphasize that we are happy with any super-constant $k$.
With some effort, we believe the range $k \le \log\log n$ of monotone lower bounds (I)$^-$ and (II)$^-$ can also be improved somewhat, though probably not to $k = n^{\Omega(1)}$ using our methods.

\paragraph{Tradeoff for $\PMM_{n,k}$.} 
Lower bound (II)$^-$ on the $\ACzero$ formula size\vspace{1.5pt} of $\SPMM_{n,k}$ implies a slightly weaker $n^{\Omega({dk^{\smash{1/(2d+O(1))}}})}$ lower bound on the $\ACzero$ formula size of the more restricted {\em Iterated Permutation Matrix Multiplication} problem $\PMM_{n,k}$, whose $\ACzero$ complexity was previously studied in 
\cite{ajtai1989first,beame1998improved,bellantoni1992approximation}.
This extension of lower bound (II)$^-$ follows from the existence of a randomized $\ACzero$ reduction from $\SPMM_{n,k}$ to $\PMM_{n,k}$, using 
the method of \cite{matias1991converting,hagerup1991fast} to generate uniform random permutations in $\ACzero$ (see also \cite{viola2012complexity}).

\subsection{Related work}\label{sec:related}

\subsubsection{Previous lower bounds} 
Prior to the results of this paper, the strongest size-depth tradeoff for $\PMM_{n,k}$ (or $\SPMM_{n,k}$)
was $n^{\vphantom|\nosmash{k^{\exp(-O(d))}}}$, proved by Beame, Impagliazzo and Pitassi \cite{beame1998improved} using a special purpose {\em switching lemma}. 
This tradeoff, which is non-trivial to depth $O(\log\log k)$, improved previous tradeoffs of \cite{ajtai1989first,bellantoni1992approximation} which were only non-trivial to depth $O(\log^\ast k)$.

With respect to the significantly more general $\BMM_{n,k}$ problem,
Chen, Oliveira, Servedio and Tan \cite{cost2016} established a nearly tight tradeoff $n^{\Omega(\nosmash{(1/d)k^{1/d}})}$ lower bound for depth-$d$ $\ACzero$ {circuits}.
This lower bound is non-trivial to depth $O({\scriptstyle\frac{\log k}{\log\log k}})$, and for depths $d = o({\scriptstyle\frac{\log k}{\log\log k}})$ it is quantitatively stronger than the $n^{\Omega(\nosmash{k^{1/(2d-1)}})}$ lower bound implied by Theorem \ref{thm:AC0tradeoffs} for depth-$d$ ${\ACzero}$ {circuits}.
However, the result of \cite{cost2016} is fundamentally a size-depth tradeoff for {\em Sipser functions} (read-once formulas with a carefully chosen fan-in sequence), combined with a reduction from the $\NCone$-complete \textsc{sipser} problem to $\BMM_{n,k}$. In contrast to \cite{beame1998improved} the present paper, this result says nothing about the complexity of $\PMM_{n,k}$ or the questions of $\NCone$ vs.\ $\L$ or $\NCone$ vs.\ $\NL$.

\subsubsection{Different regimes of parameters $n$ and $k$}

Our focus in this paper is on the formula complexity of $\SPMM_{n,k}$ in the regime where $k$ is an arbitrarily slow-growing but super-constant function of $n$.
We have not attempted to optimize the ranges of $k \le \log\log n$ and $k \le \log^\ast n$ in Theorem \ref{thm:AC0tradeoffs}.
We remark that the circuit lower bound of Beame {\em et al} \cite{beame1998improved} and Chen {\em et al} \cite{cost2016} 
extend to larger ranges $k \le \log n$ and $k \le n^{1/5}$, respectively.

In the opposite parameter regime where $n=2$ and $k$ is unbounded, all three problems $\BMM_{n,k}$, $\SPMM_{n,k}$ and $\PMM_{n,k}$
are equivalent to $\textsc{parity}_k$ under $\ACzero$ reductions.  Here tight $2^{\Omega(\nosmash{k^{1/d}})}$ and $2^{\Omega(\nosmash{dk^{1/d}})}$ lower bounds for depth $d+1$ $\ACzero$ circuits and formulas computing $\textsc{parity}_k$ were shown by H{\aa}stad \cite{Hastad86} and the author \cite{rossman2018average} respectively. 

The complexity jumps at $n=5$,
where all three problems are complete for $\NCone$ 
by Barrington's Theorem 
\cite{barrington1986bounded}.
The complexity of $\BMM_{n,k}$ (resp.\ $\SPMM_{n,k}$ and $\PMM_{n,k}$) then jumps again at $k = n$ to being complete for $\NL$ (resp.\ $\L$). 
(See \cite{mereghetti2000threshold,wang2012circuit}
concerning $\TCzero$ and $\ACCzero$ circuits for $\BMM_{n,k}$.)

\subsubsection{Algebraic circuit classes} 

The {Iterated Matrix Multiplication} polynomial ($\IMM_{n,k}$) is computable by set-multilinear formulas of product-depth $d$ and size $n^{O(\nosmash{dk^{1/d}})}$ by construction (I) of Proposition \ref{prop:upper} with $\sum,\prod$ replacing $\bigwedge,\bigvee$.
A recent breakthrough of Limaye, Srinivasan and Tavenas \cite{limaye2022superpolynomial} showed that arithmetic formulas of product-depth $d$ computing $\IMM_{n,k}$ require size $n^{k^{\nosmash{1/\exp(O(d))}}}$. 
This tradeoff, quantitatively similar to \cite{beame1998improved}, is the first super-polynomial lower bound for constant-depth arithmetic formulas.
A related paper of Tavenas, Limaye, and Srinivasan \cite{tavenas2022set}
gives a different 
$\nosmash{(\log n)^{\Omega(dk^{1/d})}}$ lower bound for 
set-multilinear formulas.
In the further restricted class of set-multilinear formulas with the {\em few parse trees} property,
Legarde, Limaye and Srinivasan \cite{lagarde2019lower} established a tight $n^{\Omega(\nosmash{dk^{1/d}})}$ tradeoff.

He and Rossman \cite{he_et_al:LIPIcs.ITCS.2023.68} studied the complexity of $\PMM_{n,k}$ in the hybrid Boolean-algebraic model of $\ACzero$ formula that are invariant under a group action of $\mr{Sym}(n)^{k-1}$. This is the class of $\ACzero$ formulas on $kn^2$ variables (encoding permutation matrices $M^{(1)},\dots,M^{(k)}$) which are syntactically invariant under a certain action of $\mr{Sym}(n)^{k-1}$, where the $i$th symmetric group $\mr{Sym}(n)$ acts by permuting both the columns of $M^{(i)}$ and rows of $M^{(i+1)}$.
For instance, construction (II) of Proposition \ref{prop:upper} produces $\mr{Sym}(n)^{k-1}$-invariant $\ACzero$ formulas of depth $d+1$ and size $n^{dk^{\nosmash{1/d}} \,+\, O(1)}$ which solve $\PMM_{k,n}$ for all $k \le n$ and $d \le \log k$.
The paper \cite{he_et_al:LIPIcs.ITCS.2023.68} proves a sharply matching $\nosmash{n^{dk^{\nosmash{1/d}} \,-\, O(1)}}$ lower bound in the $\mr{Sym}(n)^{k-1}$-invariant setting. This lower bound even extends to $\mr{Sym}(n)^{k-1}$-invariant $\TCzero$ formulas.

\subsubsection{Smallest known formulas for {\normalfont$\SPMM_{n,k}$}}

Construction (II) of Proposition \ref{prop:upper} in the limit $d = \log_2 k$ produces $\ACzero$ formulas for $\SPMM_{n,k}$ of depth $O(\log k)$ and size $n^{\log_2 k \,+\, O(1)}$. These formulas are moreover $\mr{Sym}(n)^{k-1}$-invariant, monotone, and uniform. 
However, these formulas are not the smallest known. 
For all $k \le n$,
Rossman \cite{rossman2018formulas} showed that there exist
(non-$\mr{Sym}(n)^{k-1}$-invariant, non-monotone, non-uniform) $\ACzero$ formulas for $\SPMM_{n,k}$ of depth $O(\log k)$ and size $n^{\nosmash{\frac12}\log_2 k \,+\, O(1)}$.
This was subsequently improved to $n^{\nosmash{\frac13}\log_{\varphi} k \,+\, O(1)}$ where $\nosmash{\varphi = \frac{1+\sqrt 5}{2}}$ by Kush and Rossman \cite{kush2023tree}, who have conjectured 
that $\frac13 \log_{\varphi} k \approx 0.49 \log_2 k$ is optimal in the exponent of $n$.

\subsection{Outline of the paper}\label{sec:outline}

The rest of this paper is organized as follows. 
    \S\ref{sec:preliminaries} gives several definitions 
    pertaining to $\ACzero$ and $\SACzero$ 
    formulas, subgraphs of $\Path_k$, and join trees.
\S\ref{sec:tradeoffs-for-join-trees} states our combinatorial main result (Theorem \ref{thm:tradeoff}) consisting of two distinct tradeoffs (I) and (II) for join trees.
Parts (I) and (II) of Theorem \ref{thm:tradeoff} are proved in the next two sections \S\ref{sec:tradeoff1} and \S\ref{sec:tradeoff2}.
\S\ref{sec:pathset-tradeoff} reviews the Pathset Framework (introduced in papers \cite{rossman2015correlation,rossman2018formulas}) and establishes new lower bounds for pathset complexity (Theorem \ref{thm:chi2}), which follow from our tradeoffs for join trees. 
Finally, our size-depth tradeoffs for $\ACzero$ and $\SACzero$ formulas (Theorem \ref{thm:AC0tradeoffs}) are proved in \S\ref{sec:tradeoffs-for-formulas} via reductions to our lower bounds for pathset complexity.

We advise the reader that \S\ref{sec:preliminaries}--\ref{sec:tradeoff2} of this paper may be read a self-contained work dealing exclusively with combinatorics of join trees for the graph $\Path_k$.
On the other hand, readers mainly interested in the application to Circuit Complexity (Theorem \ref{thm:AC0tradeoffs}) might wish to skip \S\ref{sec:tradeoff1}--\ref{sec:tradeoff2} (the proof of Theorem \ref{thm:tradeoff}) on an initial reading, in order to first understand our method of reducing $\cc{(S)}\ACzero$ formulas computing $\SPMM_{n,k}$ to pathset formulas computing dense relations $\mc A \subseteq [n]^k$ and finally to join trees computing the graph $\Path_k$.

\section{Preliminaries}\label{sec:preliminaries}

$\Z$ is the set of integers, and $\N = \{0,1,2,\dots\}$ is the set of nonnegative integers. Throughout this paper, $d,k,m,n$ represent positive integers, and $h,i,j,l,r,s,t$ represent (usually nonnegative) integers. $[i]$ stands for the set $\{1,\dots,i\}$. In particular, $[0]$ is the empty set $\emptyset$.
 
$\ln(\cdot)$ and $\log(\cdot)$ denote the base-$\Exp$ and base-2 logarithms. $\log^\ast(\cdot)$ is the number of times the $\log(\cdot)$ function must be iterated until the result is less than $1$.

\subsection{Graphs and trees} 

{\em Graphs} in this paper are 
simple graphs without isolated vertices. That is, a graph $G$ consists of 
sets $V(G)$ and $E(G) \subseteq \binom{V(G)}{2}$ such that $V(G) = \bigcup_{e \in E(G)} e$.
The empty graph is denoted by the empty set symbol $\emptyset$.
With the exception of the infinite path graph $\Path_\Z$, all other graphs considered in this paper are finite.

{\em Trees} in this paper are finite rooted trees. 
{\em Binary trees} are trees in which each non-leaf node has exactly two children, designed ``left'' and ``right''.
For binary trees $T_1$ and $T_2$, we denote by $\un{T_1}{T_2}$ the binary tree consisting of a root with $T_1$ and $T_2$ as its left and right subtrees.

\subsection{$\ACzero$ and $\SACzero$ formulas}\label{sec:formulas}

An {\em $\ACzero$ formula} on variables $X_1,\dots,X_N$ is a rooted tree whose leaves (called ``inputs'') are labeled by a constant $0$ or $1$ or literals $X_i$ or $\BAR X_i$, and whose non-leaves (called ``gates'') are labeled by $\bigwedge$ or $\bigvee$.
An $\ACzero$ formula is {\em monotone} if no input is labeled by a negative literal $\BAR X_i$.
Every [monotone] $\ACzero$ formula computes a [monotone] Boolean function $\{0,1\}^N \to \{0,1\}$ in the usual way.
  
We measure the {\em size} of an $\ACzero$ formula by the number of leaves that are labeled by literals.
{\em Depth} (resp.\ {\em $\bigwedge$-depth}) is the maximum number of gates (resp.\ $\bigwedge$-gates) on a root-to-leaf-branch.  
{\em Fan-in} (resp.\ {\em $\bigwedge$-fan-in}) is the maximum number of children of any gate (resp.\ $\bigwedge$-gate). 

Depth $0$ $\ACzero$ formulas (i.e.,\ literals) are also known as {\em $\Sigma_0$ formulas} and {\em $\Pi_0$ formulas}. For $d \ge 1$, a {\em $\Sigma_d$ formula} (resp.\ {\em $\Pi_d$ formula}) is either a $\Pi_{d-1}$ formula (resp.\ $\Sigma_{d-1}$ formula) or a depth $d$ $\ACzero$ formula whose output gate is labeled $\bigvee$ (resp.\ $\bigwedge$). Note that $\Sigma_{2d+1}$ formulas are precisely the class of $\ACzero$ formulas with $\bigwedge$-depth $d$.

{\em $\SACzero$ formulas} (where ``\textsl{S}'' stands for ``semi-unbounded fan-in'') are $\ACzero$ formulas with bounded $\bigwedge$-fan-in.  This usually refers to $\bigwedge$-fan-in $2$ (i.e.,\ formulas with unbounded fan-in $\bigvee$ gates and fan-in~$2$ $\wedge$ gates).  However, 
when speaking of problems $\BMM_{n,k}$, $\SPMM_{n,k}$ and $\PMM_{n,k}$, 
we allow 
$\bigwedge$-fan-in $O(k^{1/d})$ in the context of upper bounds, and at most $n^{1/k}$ in the context of lower bounds. Note that $k^{1/d} \ll n^{1/k}$ in the regime of parameters that we consider (i.e.,\ $k \le \log\log n$ and $d \le \log k$).

\subsection{Subgraphs of paths}\label{sec:subgraphs-of-path}

The graphs that we will mainly consider are {\em disjoint unions of paths}.  It is convenient to regard these as finite subgraphs of the infinite path graph $\Path_\Z$ with vertex set $\Z$ and edge set $\{\{i-1,i\} : i \in \Z\}$.

For integers $s < t$, let $\Path_{s,t}$ denote subgraph of $\Path_\Z$ with vertex set $\{i \in \Z : s \le i \le t\}$ and edge set $\{\{i-1,i\} \in \Z : s < i \le t\}$.
For positive integers $k$, we write $\Path_k$ for $\Path_{0,k}$.

Note that every finite subgraph of $\Path_\Z$ is
is a union of paths $\Path_{s_1,t_1} \cup \dots \cup \Path_{s_c,t_c}$ for unique integers $c \ge 0$ and $s_1 < t_1 < s_2 < \dots < t_{c-1} < s_c < t_c$. 

The next three definitions introduce notation that will be heavily used throughout the paper.

\begin{df}[Parameters $\|G\|$, $\Delta(G)$, $\lambda(G)$]
For a finite graph $G \subset \Path_\Z$, let
\[
  &&&&&&&&
  \|G\| &\defeq\#\{\tu{edges of }G\},\\
  &&&&&&&&
  \vphantom{\bigg|}
  \Delta(G) &\defeq \#\{\tu{connected components of }G\}
  &&= \#\{\tu{vertices of }G\} - \#\{\tu{edges of }G\},\\
  &&&&&&&&
  \lambda(G) &\defeq 
  \max_{\tu{connected components $C$ of }G}\,\|C\|
  &&= 
  \tu{length of the longest path in }G.&&&&&&&&
\]
Note that if $G = \Path_{s_1,t_1} \cup \dots \cup \Path_{s_c,t_c}$ where $s_1 < t_1 < s_2 < \dots < t_{c-1} < s_c < t_c$, then 
\[
  \|G\| = \sum_{i=1}^c\  (t_i - s_i),\qquad 
  \Delta(G) = c,\qquad
  \lambda(G) = \max_{i\in\{1,\dots,c\}}\ (t_i-s_i).
\]
\end{df}

\begin{df}[The operation $G \ominus F$]
For graphs $G,F \subset \Path_\Z$, we denote by $G \ominus F$ the subgraph of $G$ comprised of the connected components of $G$ that are vertex-disjoint from $F$.
\end{df}

\begin{df}
For a sequence of finite graphs $G_1,\dots,G_m \subset \Path_\Z$, let
\[
  \vv\Delta(G_1,\dots,G_m)
  &\defeq
  \sum_{j\in[m]} \Delta(G_j \ominus (G_1 \cup \dots \cup G_{j-1})),\\
  \vv{\lambda}(G_1,\dots,G_m) &\defeq \sum_{j\in[m]} 
  \lambda(G_j \ominus (G_1 \cup \dots \cup G_{j-1})),\\
  \vv{\lambda\Delta}(G_1,\dots,G_m) &\defeq \sum_{j\in[m]} 
  \lambda(G_j \ominus (G_1 \cup \dots \cup G_{j-1}))
  \cdot
  \Delta(G_j \ominus (G_1 \cup \dots \cup G_{j-1})).
\]
For $F \subset \Path_\Z$, we also introduce a convenient ``conditional $\vv\Delta$'' notation:
\[
  \vv\Delta(G_1,\dots,G_m \mid F)
  &\defeq
  \sum_{j\in[m]} \Delta(G_j \ominus (F \cup G_1 \cup \dots \cup G_{j-1})).
\]
\end{df}

The value of $\vv\Delta(G_1,\dots,G_m)$ is sensitive to the order of $G_1,\dots,G_m$, as the following example shows.

\begin{ex}\label{ex:25-0}
Let $E_1,E_2,E_3,\dots,E_{25}$ be the single-edge subgraphs of $\Path_{25}$ in the standard order. We have
\[
  \vv\Delta(E_1,E_2,E_3,
  \dots,E_{25}) = 1,
\]
since $\vv\Delta(E_1) = 1$ and $\vv\Delta(E_j \ominus (E_1 \cup \dots \cup E_{j-1})) = 0$ for $j = 2,\dots,25$. 
However, we can increase the $\vv\Delta$-value by permuting the sequence $E_1,\dots,E_{25}$ as follows:
\[  
  \vv\Delta(\underbrace{E_1,E_3,E_5,
  \dots,
  E_{25}},\,
  \underbrace{E_2,E_4,E_6,
  \dots,
  E_{24}})
  =
  13.
\]
It is not hard to check that $13$ is the maximum possible $\vv\Delta$-value among all permutations of $E_1,\dots,E_{25}$.
Generalizing this example from $25$ to any $k$, we get a maximum $\vv\Delta$-value of $\lceil k/2 \rceil$.
\end{ex}

For later reference, the following lemma records some basic properties of $\vv\Delta(\,\cdot\,|\,\cdot\,)$, which follow  
directly
from 
definitions.

\begin{la}\label{la:delta-props}
Let $F,G_1,\dots,G_m,H_1,\dots,H_m$ be finite subgraphs of $\Path_\Z$.
\begin{enumerate}[\quad\normalfont(a)]
\item
$\vv\Delta(G_1,\dots,G_m\mid F) \le \vv\Delta(G_1,\dots,G_m\mid F_0)$ for all $F_0 \subseteq F$.
\item
$\vv\Delta(G_1,\dots,G_m\mid\emptyset) = \vv\Delta(G_1,\dots,G_m)$. More generally, 
\[
\vv\Delta(G_1,\dots,G_m\mid F)
=
\vv\Delta(F,G_1,\dots,G_m) - \Delta(F).
\]
\item\label{la:delta-props1}
$\Delta(G_1\cup\dots\cup G_m) \le \vv\Delta(G_1,\dots,G_m)$. More generally,
\[
\Delta((G_1 \cup \dots \cup G_m) \ominus F) 
\le \vv\Delta(G_1,\dots,G_m\mid F)
\le
\vv\Delta(G_1\ominus F,\dots,G_m\ominus F)
.
\]
\item
$\vv\Delta(\,\cdot\,|\,\cdot\,)$ decomposes via the following ``chain rule'':
\[
  \vv\Delta(G_1,\dots,G_m,H_1,\dots,H_n\mid F)
  =
  \vv\Delta(G_1,\dots,G_m\mid F) + \vv\Delta(H_1,\dots,H_n\mid F \cup G_1 \cup \dots \cup G_m).
\]
\item
$\vv\Delta(G_1,G_2,G_3,\dots,G_m) = \vv\Delta(G_1,\,G_1 \cup G_2,\,G_1 \cup G_2 \cup G_3,\,\dots,\,G_1 \cup \dots \cup G_m)$.
\qed
\end{enumerate}
\end{la}

\subsection{Join trees}\label{sec:join-trees}

We define the notion of {\em join trees} with respect to arbitrary finite graphs $G$. However, our focus will soon narrow to join trees for subgraphs of $\Path_k$.

\begin{df}[Join tree]
For any finite graph $G$, a {\em $G$-join tree} is a 
binary tree $T$ together a labeling of nodes of $T$ by subgraphs of $G$ such that
\begin{itemize}
  \item
    each leaf is labeled by a single-edge subgraph of $G$ or the empty graph,
  \item
    each non-leaf is labeled by the union of the graphs labeling its children, and
  \item
    the root is labeled by $G$.
\end{itemize}
Note that the labeling function is induced by its value on leaves of $T$, given by a surjective partial function from $\{$leaves of $T\}$ to $E(G)$.
\end{df}

We consider non-standard notions of ``depth'' and ``size'' for join trees.  Our combinatorial main theorem gives tradeoffs between these parameters, which we define next.

\begin{df}[Operations $\un{T_1}{T_2}$, $\sqq{T_1,\dots,T_m}$ and $\sem{T_1,\dots,T_m}$]
If $T_1$ is a $G_1$-join tree and $T_2$ is a $G_2$-join tree, then we denote by $\un{T_1}{T_2}$ the $G_1 \cup G_2$-join tree consisting of a root with left and right subtrees $T_1$ and $T_2$ with the naturally induced labeling of nodes.

Operations $\sqq{\cdot}$ and $\sem{\cdot}$ are defined inductively on finite sequences of join trees. In the base case $m = 1$, we have $\sqq{T_1}=\sem{T_1}\defeq T_1$.
For $m \ge 2$, we define
\begin{alignat*}{4}
  \sqq{T_1,\dots,T_m} &\defeq \un{T_1}{\sqq{T_2,\dots,T_m}} 
  &&= 
  \un{T_1}{\unp{T_2}{
  \ldots \unp{T_{m-2}}{\unp{T_{m-1}}{T_m}}\ldots
  }},
  \vphantom{\Big|}\\
  \sem{T_1,\dots,T_m} &\defeq 
  \un{\,\sem{T_1,\dots,
  T_{m-1}}}{\sem{T_1,\dots,T_{m-2},T_m}\,}\vphantom{\big|}
  \quad&&=
  \sem{\,\un{T_1}{T_m},\un{T_2}{T_m},\dots,\un{T_{m-1}}{T_m}\,}.
  \vphantom{\Big|}
\end{alignat*}
\end{df}

Note that $\sqq{T_1,T_2} = \sem{T_1,T_2} = \un{T_1}{T_2}$ in the case $m = 2$.  However, note that operations $\sqq{\cdot}$ and $\sem{\cdot}$ are different for $m \ge 3$.  For visual reference, the case $m=5$ is picture below where $T_1,\dots,T_5$ are $G_1,\dots,G_5$-join trees:\medskip

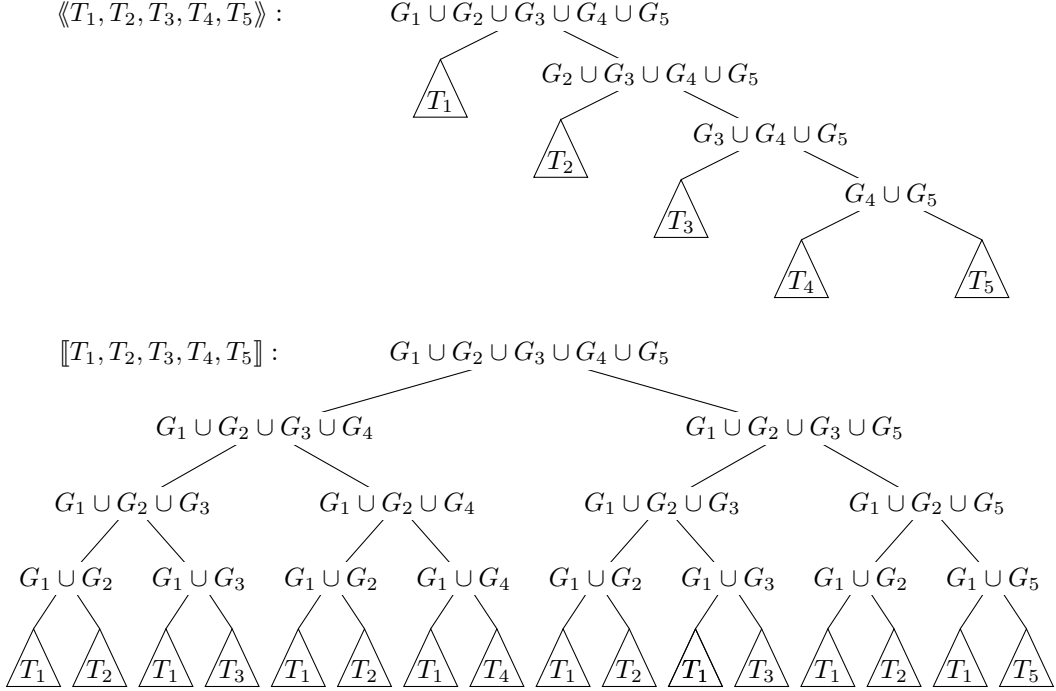
\begin{center}
\begin{tikzpicture}[scale = .4]
  \tikzstyle{v}=[circle, draw, fill, inner sep=0pt, minimum width=2.5pt]
  \tikzstyle{w}=[fill=white, inner sep=2pt]
  
  \tikzstyle{ww}=[draw, fill=white, inner sep=2pt]
  \tikzstyle{www}=[draw, circle, fill=white, inner sep=2.5pt]
  \tikzstyle{wwww}=[draw, inner sep=2pt]
  
  \def \shiftor {.35}
  \def \Tdown {1.75+\shiftor}
  \def \tridown {1.7+\shiftor} 
 
  \tikzstyle{tri}=[draw,
  shape border uses incircle,
  isosceles triangle,
  isosceles triangle apex angle=50,
  scale=1.1,
  shape border rotate=90
  ] 

  \def \x {2}
  \def \z {4}

  \def \y {.25}
 
  \draw (0,0*\x) -> (\z,-1*\x);
  \draw (\z,-1*\x) -> (2*\z,-2*\x);
  \draw (2*\z,-2*\x) -> (3*\z,-3*\x);
  \draw (3*\z,-3*\x) -> (4*\z-\y*\z,-4*\x+\y*\x);
  
  \draw     (0,0*\x) -> (-\z+\y*\z,-1*\x+\y*\x);
  \draw   (\z,-1*\x) -> (0+\y*\z,-2*\x+\y*\x);
  \draw (2*\z,-2*\x) -> (\z+\y*\z,-3*\x+\y*\x);
  \draw (3*\z,-3*\x) -> (2*\z+\y*\z,-4*\x+\y*\x);

  \draw (-12,0) node [w] {$\sqq{T_1,T_2,T_3,T_4,T_5}:$};
  
  \draw (0,0) node [w] {$G_1\cup G_2\cup G_3\cup G_4 \cup G_5$};
  \draw (\z,-1*\x) node [w] {$G_2\cup G_3\cup G_4\cup G_5$};
  \draw (2*\z,-2*\x) node [w] {$G_3\cup G_4\cup G_5$};
  \draw (3*\z,-3*\x) node [w] {$G_4\cup G_5$};
  
  \draw (-\z+\y*\z,-1*\x+\y*\x-\tridown) node [tri] {\phantom\,};
  \draw (0+\y*\z,-2*\x+\y*\x-\tridown) node [tri] {\phantom\,};
  \draw (\z+\y*\z,-3*\x+\y*\x-\tridown) node [tri] {\phantom\,};
  \draw (2*\z+\y*\z,-4*\x+\y*\x-\tridown) node [tri] {\phantom\,};
  \draw (4*\z-\y*\z,-4*\x+\y*\x-\tridown) node [tri] {\phantom\,};
  
  \draw (-\z+\y*\z,-1*\x+\y*\x-\Tdown) node {$T_1$};
  \draw (0+\y*\z,-2*\x+\y*\x-\Tdown) node {$T_2$};
  \draw (\z+\y*\z,-3*\x+\y*\x-\Tdown) node {$T_3$};
  \draw (2*\z+\y*\z,-4*\x+\y*\x-\Tdown) node {$T_4$};
  \draw (4*\z-\y*\z,-4*\x+\y*\x-\Tdown) node {$T_5$};
  
\end{tikzpicture}\bigskip
\ \\
\begin{tikzpicture}[scale = .4]
  \tikzstyle{v}=[circle, draw, fill, inner sep=0pt, minimum width=2.5pt]
  \tikzstyle{w}=[fill=white, inner sep=2pt]
  
  \tikzstyle{ww}=[draw, fill=white, inner sep=2pt]
  \tikzstyle{www}=[draw, circle, fill=white, inner sep=2.5pt]
  \tikzstyle{wwww}=[draw, inner sep=2pt]
  
  \def \shiftor {.35}
  \def \Tdown {1.75+\shiftor}
  \def \tridown {1.7+\shiftor}  
 
  \tikzstyle{tri}=[draw,
  shape border uses incircle,
  isosceles triangle,
  isosceles triangle apex angle=50,
  scale=1.1,
  shape border rotate=90
  ] 
  
  \def \multip {1.1}
  
  \def \a {8*\multip}
  \def \b {4*\multip}
  \def \c {2*\multip}
  \def \d {1*\multip}
  
  \def \Z {2.5}
  
  \def \O {5*\Z}
  \def \A {4*\Z}
  \def \B {3*\Z}
  \def \C {2*\Z}
  \def \D {1.3334*\Z}
 
  \def \e {0}
  
  \def \p {-\a+\b-\c+\d+\e}
  \def \q {\D-\e*\Z}
  
  \draw (-\a,\A) -- (0,\O);
  \draw (+\a,\A) -- (0,\O);
  
  \draw (-\a-\b,\B) -- (-\a,\A);
  \draw (-\a+\b,\B) -- (-\a,\A);
  \draw (+\a-\b,\B) -- (+\a,\A);
  \draw (+\a+\b,\B) -- (+\a,\A);
  
  \draw (-\a-\b-\c,\C) -- (-\a-\b,\B);
  \draw (-\a-\b+\c,\C) -- (-\a-\b,\B);
  \draw (-\a+\b-\c,\C) -- (-\a+\b,\B);
  \draw (-\a+\b+\c,\C) -- (-\a+\b,\B);
  \draw (+\a-\b-\c,\C) -- (+\a-\b,\B);
  \draw (+\a-\b+\c,\C) -- (+\a-\b,\B);
  \draw (+\a+\b-\c,\C) -- (+\a+\b,\B);
  \draw (+\a+\b+\c,\C) -- (+\a+\b,\B);
  
  \draw (-\a-\b-\c-\d,\D) -- (-\a-\b-\c,\C);
  \draw (-\a-\b-\c+\d,\D) -- (-\a-\b-\c,\C);
  \draw (-\a-\b+\c-\d,\D) -- (-\a-\b+\c,\C);
  \draw (-\a-\b+\c+\d,\D) -- (-\a-\b+\c,\C);
  \draw (-\a+\b-\c-\d,\D) -- (-\a+\b-\c,\C);
  \draw (-\a+\b-\c+\d,\D) -- (-\a+\b-\c,\C);
  
  \draw (-\a+\b+\c-\d,\D) -- (-\a+\b+\c,\C);
  \draw (-\a+\b+\c+\d,\D) -- (-\a+\b+\c,\C);
  \draw (+\a-\b-\c-\d,\D) -- (+\a-\b-\c,\C);
  \draw (+\a-\b-\c+\d,\D) -- (+\a-\b-\c,\C);
  \draw (+\a-\b+\c-\d,\D) -- (+\a-\b+\c,\C);
  \draw (+\a-\b+\c+\d,\D) -- (+\a-\b+\c,\C);
  \draw (+\a+\b-\c-\d,\D) -- (+\a+\b-\c,\C);
  \draw (+\a+\b-\c+\d,\D) -- (+\a+\b-\c,\C);
  \draw (+\a+\b+\c-\d,\D) -- (+\a+\b+\c,\C);
  \draw (+\a+\b+\c+\d,\D) -- (+\a+\b+\c,\C);

  \draw (-12,\O) node [w] {$\sem{T_1,T_2,T_3,T_4,T_5}:$};

  \draw (0,\O) node [w] {$G_1\cup G_2\cup G_3\cup G_4 \cup G_5
  $};
  
  \draw (-\a,\A) node [w] {$G_1\cup G_2\cup G_3\cup G_4$};
  \draw (+\a,\A) node [w] {$G_1\cup G_2\cup G_3\cup G_5$};
  
  \draw (-\a-\b,\B) node [w] {$G_1\cup G_2\cup G_3$};
  \draw (-\a+\b,\B) node [w] {$G_1\cup G_2\cup G_4$};
  \draw (+\a-\b,\B) node [w] {$G_1\cup G_2\cup G_3$};
  \draw (+\a+\b,\B) node [w] {$G_1\cup G_2\cup G_5$};
  
  \draw (-\a-\b-\c,\C) node [w] {$G_1 \cup G_2$};
  \draw (-\a-\b+\c,\C) node [w] {$G_1 \cup G_3$};
  \draw (-\a+\b-\c,\C) node [w] {$G_1 \cup G_2$};
  \draw (-\a+\b+\c,\C) node [w] {$G_1 \cup G_4$};
  \draw (+\a-\b-\c,\C) node [w] {$G_1 \cup G_2$};
  \draw (+\a-\b+\c,\C) node [w] {$G_1 \cup G_3$};
  \draw (+\a+\b-\c,\C) node [w] {$G_1 \cup G_2$};
  \draw (+\a+\b+\c,\C) node [w] {$G_1 \cup G_5$};
  
  \draw (-\a-\b-\c-\d,\D-\tridown) node [tri] {\phantom\,};
  \draw (-\a-\b-\c+\d,\D-\tridown) node [tri] {\phantom\,};
  \draw (-\a-\b+\c-\d,\D-\tridown) node [tri] {\phantom\,};
  \draw (-\a-\b+\c+\d,\D-\tridown) node [tri] {\phantom\,};
  \draw (-\a+\b-\c-\d,\D-\tridown) node [tri] {\phantom\,};
  \draw (-\a+\b-\c+\d,\D-\tridown) node [tri] {\phantom\,};
  \draw (-\a+\b+\c-\d,\D-\tridown) node [tri] {\phantom\,};
  \draw (-\a+\b+\c+\d,\D-\tridown) node [tri] {\phantom\,};
  \draw (+\a-\b-\c-\d,\D-\tridown) node [tri] {\phantom\,};
  \draw (+\a-\b-\c+\d,\D-\tridown) node [tri] {\phantom\,};
  \draw (+\a-\b+\c-\d,\D-\tridown) node [tri] {\phantom\,};
  \draw (+\a-\b+\c-\d,\D-\tridown) node [tri] {\phantom\,};
  \draw (+\a-\b+\c+\d,\D-\tridown) node [tri] {\phantom\,};
  \draw (+\a+\b-\c-\d,\D-\tridown) node [tri] {\phantom\,};
  \draw (+\a+\b-\c+\d,\D-\tridown) node [tri] {\phantom\,};
  \draw (+\a+\b+\c-\d,\D-\tridown) node [tri] {\phantom\,};
  \draw (+\a+\b+\c+\d,\D-\tridown) node [tri] {\phantom\,};

  \draw (-\a-\b-\c-\d,\D-\Tdown) node {$T_1$};
  \draw (-\a-\b-\c+\d,\D-\Tdown) node {$T_2$};
  \draw (-\a-\b+\c-\d,\D-\Tdown) node {$T_1$};
  \draw (-\a-\b+\c+\d,\D-\Tdown) node {$T_3$};
  \draw (-\a+\b-\c-\d,\D-\Tdown) node {$T_1$};
  \draw (-\a+\b-\c+\d,\D-\Tdown) node {$T_2$};
  \draw (-\a+\b+\c-\d,\D-\Tdown) node {$T_1$};
  \draw (-\a+\b+\c+\d,\D-\Tdown) node {$T_4$};
  \draw (+\a-\b-\c-\d,\D-\Tdown) node {$T_1$};
  \draw (+\a-\b-\c+\d,\D-\Tdown) node {$T_2$};
  \draw (+\a-\b+\c-\d,\D-\Tdown) node {$T_1$};
  \draw (+\a-\b+\c-\d,\D-\Tdown) node {$T_1$};
  \draw (+\a-\b+\c+\d,\D-\Tdown) node {$T_3$};
  \draw (+\a+\b-\c-\d,\D-\Tdown) node {$T_1$};
  \draw (+\a+\b-\c+\d,\D-\Tdown) node {$T_2$};
  \draw (+\a+\b+\c-\d,\D-\Tdown) node {$T_1$};
  \draw (+\a+\b+\c+\d,\D-\Tdown) node {$T_5$};
  
\end{tikzpicture}\bigskip
\end{center}

\begin{df}[Depth measures]
We define
{\em $\joinop$-depth}, {\em $\sqq{}$-depth} and {\em $\semempty$-depth} as the pointwise minimum functions $\{$join trees$\} \to \{0,1,2,\dots\}$ satisfying
\[
  \depth_{\joinop}(\un{T_1}{T_2}) 
  &\le 
  1+\max\{\depth_{\joinop}(T_1),\,\depth_{\joinop}(T_2)\},\\
  \depth_{\sqq{}}(\sqq{T_1,\dots,T_m}) 
  &\le 
  1+\max\{\depth_{\sqq{}}(T_1),\,\dots,\,\depth_{\sqq{}}(T_m)\},
  \vphantom{\Big|}\\
  \depth_{\semempty{}}(\sem{T_1,\dots,T_m}) 
  &\le 
  1+\max\{\depth_{\semempty{}}(T_1),\,\dots,\,\depth_{\semempty{}}(T_m)\}.
\]
That is, the $\joinop$-depth (resp.\ $\sqq{}$, $\semempty{}$-depth) of $T$ is the minimum nesting depth of binary ${\cup}$ operations (resp.\ unbounded $\sqq{}$, $\semempty{}$ operations) required to express $T$ in terms of individual edges.
\end{df}

$\joinop$-depth is the standard notion of ``depth'' for binary trees, that is, the maximum length (= number of parent-to-child descents) of a root-to-leaf branch. 
$\sqq{}$-depth is what might be called ``left-depth'', that is, the maximum number of left descents on a root-to-leaf-branch.
$\semempty$-depth is perhaps a new notion, which arises in connection to our tradeoffs for unbounded fan-in $\ACzero$ formulas.\medskip

In order to define our ``size'' measure for join trees, we first introduce the notion of {\em branch coverings}.

\begin{df}[Branch coverings]
Let $T$ be a $G$-join tree. We associate each root-to-leaf branch $(b_1,\dots,b_\ell)$ in $T$ (where $b_1$ is the root and $b_\ell$ is a leaf) with the sequence of graphs $B_1,\dots,B_\ell$ where
\begin{itemize}
\item
$B_j$ labels the sibling of $b_{j+1}$ (i.e.,\ the opposite-side child of $b_j$) for each $j \in \{1,\dots,\ell-1\}$, 
and 
\item
$B_\ell$ labels the leaf $b_\ell$ (note that $\|B_\ell\| \le 1$).
\end{itemize}
Note that $B_1 \cup \dots \cup B_\ell = G$, that is, $B_1,\dots,B_\ell$ is a covering of $G$. 
We refer to sets $\{B_1,\dots,B_\ell\}$ arising from root-to-leaf branches as {\em $T$-branch coverings} of $G$.
\end{df}

One standard way of measuring the ``size'' of a binary tree is by the number of leaves ($=$ number of root-to-leaf branches). For join trees $T$ with graph $G \subset \Path_k$, we consider a completely different ``size'' measure defined as the \underline{maximum} of a certain complexity measure over all $T$-branch coverings $\{B_1,\dots,B_\ell\}$ of $G$.  This complexity measure is itself the \underline{maximum} $\vv\Delta$-value among all orderings $B_{\pi(1)},\dots,B_{\pi(\ell)}$. 
This rather complicated ``size'' measure, denoted by $\Psi$, arises naturally in the context of {\em pathset complexity} (see \S\ref{sec:pathset-tradeoff} and in particular Corollary \ref{cor:pathset-density} and Lemma \ref{la:chi2}, which explain the roles of $\vec\Delta$ and $\Psi$).

\begin{df}[$\Psi$-size of join trees]
Let $T$ be a $G$-join tree where $G \subset \Path_\Z$. We define the {\em $\Psi$-size} of $T$ by
\[
  \Psi(T)
  &\defeq
  \max_{\substack{
    \tu{$\vphantom{\Big|}T$-branch coverings $\{B_1,\dots,B_\ell\}$ of $G$}
    \\
    \tu{and permutations $\smash{\pi : [\ell] \stackrel\cong\to [\ell]}$}
  }} 
  \ 
  \vv\Delta(B_{\pi(1)},\dots,B_{\pi(\ell)}).
\]
Stated differently, $\Psi(T)$ is the maximum value of $\vv\Delta(C_1,\dots,C_m)$ over all $T$-branch covering $\mc C$ of $G$ and enumerations $C_1,\dots,C_m$ of $\mc C$. 
This alternative description of $\Psi(T)$ emphasizes the fact that we treat $T$-branch coverings as (unordered) sets.
\end{df}

\section{Tradeoffs for join trees}\label{sec:tradeoffs-for-join-trees}

In this section we state the combinatorial main theorem of this paper, Theorem \ref{thm:tradeoff}, which gives tight tradeoffs between the $\Psi$-size and the $\sqq{}$- and $\semempty$-depth of $\Path_k$-join trees.
We also state two auxiliary results, Lemmas \ref{la:pre-pi-sigma} and \ref{la:pi-sigma}, which concern coverings of $\Path_k$.
Proofs of these results are postponed to \S\ref{sec:tradeoff1} and \S\ref{sec:tradeoff2} (the combinatorial core of the paper, which may be read independently of \S\ref{sec:pathset-tradeoff} and \S\ref{sec:tradeoffs-for-formulas}).
We also state and prove a key numerical inequality, Lemma \ref{la:numerical}, which will be used later on in the inductive proofs of our tradeoffs for join trees.

\subsection{Shift permutations}\label{sec:shift}

Given a sequence $G_1,\dots,G_m \subset \Path_\Z$, we will be interested in the question of maximizing $\vv\Delta(G_{\pi(1)},\dots,G_{\pi(m)})$ over permutations $\pi$. 
We will also be studying a variant of this question over a restricted class of permutations, which we call ``shift permutations''. (The author is unaware of any standard terminology for this class of permutations.)

\begin{df}
A {\em shift permutation} is a permutation $\sigma : [m] \stackrel\cong\to [m]$ satisfying $\sigma(j) \ge j-1$ for all $j \in [m]$. 
(For clarity of notation, we will consistently write $\sigma$ for shift permutations and $\pi$ for general permutations.)
\end{df}

There are exactly $2^{m-1}$ shift permutations, which we will index via a bijection 
\[
  I \mapsto \sigma_I : \{\tu{subsets of $[m]$ containing }m\} \to \{\tu{shift permutations}\}.
\]
Suppose $I = \{i_1,\dots,i_p\}$ with $0 =: i_0 \le i_1 < \dots < i_p = m$. Then we define 
$\sigma_I : [m] \stackrel\cong\to [m]$ by
\[
  \sigma_I : (1,\dots,m) \mapsto
  (
    \underbrace{i_1,1,2,\dots,i_1-1}_{\tu{first }i_1},\,
    \underbrace{i_2,i_1+1,i_2+2\dots,i_2-1}_{\tu{next }i_2-i_1},\,
    \dots,\,
    \underbrace{i_p,i_{p-1}+1,i_{p-1}+2,\dots,m-1}_{\tu{last }m-i_{p-1}}
  ).
\]
That is,
\[
  \sigma_I(j)
  \defeq
  \begin{cases}
    i_h &\text{if } j = i_{h-1} + 1 \text{ for some } h \in [p],\\
    j-1 &\text{otherwise.}
  \end{cases}
\]
Note that $\sigma_{[m]}$ is the identity permutation and $\sigma_{\{m\}}$ is the $m$-cycle $(1,\dots,m) \mapsto (m,1,\dots,m-1)$. 
Also note that inverse of the bijection $I \mapsto \sigma_I$ is given by
\[
  \sigma 
  \mapsto 
  \Big\{j \in [m] : \{\sigma(1),\sigma(2),\dots,\sigma(j)\} = \{1,2,\dots,j\}\Big\}
  \quad
  \Big({=}\ \{\sigma(j) : j \in [m] \text{ such that } \sigma(j) \ge j\}\Big).
\]

\begin{ex}\label{ex:25-1}
Again consider the sequence $E_1,E_2,E_3,\dots,E_{25}$ of the single-edge subgraphs of $\Path_{25}$.
Example \ref{ex:25-0} observes that $\vv\Delta(E_{\pi(1)},\dots,E_{\pi(25)}) = 13$ where $\pi$ is the permutation 
\[
\pi : (1,\dots,25) \mapsto (\underbrace{1,3,5,\dots,25},\,\underbrace{2,4,6,\dots,24}).
\]
Note that $\pi$ is \underline{not} a shift permutation. 
However, we can achieve the same value $\vv\Delta(E_{\sigma(1)},\dots,E_{\sigma(m)}) = 13$ via the shift permutation $\sigma = \sigma_{\{1,3,5,\dots,25\}}$. 
That is, we have
\[
  \vv\Delta(E_1,\, \underbrace{E_3,E_2},\, \underbrace{E_5,E_4},\,   \underbrace{E_{25},E_{24}})
  =
  13.
\]
Generalizing from $25$ to any odd $k$, the shift permutation $\sigma_{\{1,3,5,\dots,k\}}$ produces a $\vv\Delta$-value of $\frac{k+1}{2}$.
\end{ex}

\begin{ex}\label{ex:25-2}
Suppose that instead of $E_1,\dots,E_{25}$ in the usual order, we are given the sequence
\[
  \underbrace{E_1,E_6,E_{11},E_{16},E_{21}},\,
  \underbrace{E_2,E_7,E_{12},E_{17},E_{22}},\,
  \underbrace{E_3,E_8,E_{13},E_{18},E_{23}},\,
  \underbrace{E_4,E_9,E_{14},E_{19},E_{24}},\,
  \underbrace{E_5,E_{10},E_{15},E_{20},E_{25}}.
\]
Note that the $\vv\Delta$-value of this sequence equals $5$.

What is the best $\vv\Delta$-value we can obtain by applying a shift permutation to this sequence of graphs?  It turns out that the answer is $7$.  Among many possibilities, this is achieved by the shift permutation $\sigma_{\{15,25\}}$, which produces the sequence
\[
  E_{23},
  \underbrace{E_1,E_6,E_{11},E_{16},E_{21}},\,
  \underbrace{E_2,E_7,E_{12},E_{17},E_{22}},\,
  \underbrace{E_3,E_8,E_{13},E_{18}},
  E_{25},
  \underbrace{E_4,E_9,E_{14},E_{19},E_{24}},\,
  \underbrace{E_5,E_{10},E_{15},E_{20}}.
\]

Generalizing this example from $25$ to any odd square $k$, we get an ordering $E_1,E_{\smash{\sqrt k}+1},E_{2\smash{\sqrt k}+1},\dots,E_k$ where applying the best shift permutation increases the $\vv\Delta$-value from $\sqrt k$ to $\sqrt k + \frac{\sqrt k-1}{2}$.
\end{ex}

\subsection{Combinatorial results of this paper}\label{sec:tradeoffs}

We now state the main combinatorial results of this paper, whose proofs are given in \S\ref{sec:tradeoff1} and \S\ref{sec:tradeoff2}.
Parts (I) and (II) in the lemmas and theorem below yield lower bounds which match upper bounds (I) and (II) of Proposition \ref{prop:upper}.

\begin{la}[Pre-Main Lemma]
\label{la:pre-pi-sigma}
Suppose that $G_1 \cup \dots \cup G_m = \Path_k$. 
\begin{enumerate}[\quad\normalfont(I)]
\item
If 
$\max_{j\in[m]} \lambda(G_j) = 1$, then there exists a {\bf\em permutation} 
$\pi$ such that
$\ds
  \vv{\Delta}(G_{\pi(1)},\dots,G_{\pi(m)})
  \ge
  k/6.
$
\item
\mbox{If $\vv\Delta(G_1,\dots,G_m) = 1$, then there exists a {\bf\em shift permutation} $\sigma$
s.t.\ 
$\ds
  \vv{\lambda}(G_{\sigma(1)},\dots,G_{\sigma(m)})
  \ge
  k/4.
$}\medskip
\end{enumerate}
\end{la}

\begin{la}[Main Lemma]
\label{la:pi-sigma}
Suppose that $G_1 \cup \dots \cup G_m = \Path_k$.
\begin{enumerate}[\quad\normalfont(I)]
\item
There exists a {\bf\em permutation} $\pi$ 
such that
$\ds
  \vv{\lambda\Delta}(G_{\pi(1)},\dots,G_{\pi(m)})
  \ge
  k/30.
$
\item
There exists a {\bf\em shift permutation} $\sigma$ 
s.t.\ 
$\ds
  \vv{\lambda\Delta}(G_{\sigma(1)},\dots,G_{\sigma(m)})
  \ge
  \sqrt{k/8}.
$\medskip
\end{enumerate}
\end{la}

Lemma \ref{la:pre-pi-sigma} is used to prove Lemma \ref{la:pi-sigma}, which is in turn used to prove the main combinatorial result of this paper:

\begin{thm}\label{thm:tradeoff}
Let $T$ be any $\Path_k$-join tree.
\begin{enumerate}[\quad\normalfont(I)]
  \item
    \parbox{1.5in}{$\ds
      \Psi(T) 
      \ge 
      \frac{1}{30\Exp} dk^{1/d} - d
    $}
    where $d$ is the $\sqq{}$-depth of $T$.
  \item
    \parbox{1.5in}{$\ds
      \Psi(T)
      \ge 
      \frac{1}{\sqrt{32\Exp}} dk^{1/2d} - d
    $}
    where $d$ is the $\semempty{}$-depth of $T$.
\end{enumerate}
\end{thm}

Note that inequalities (I) and (II) of Theorem \ref{thm:tradeoff} are respectively $\Omega(dk^{1/d})$ and $\Omega(dk^{1/2d})$ up to some $d = O(\log k)$, yet
both inequalities become trivial for some slightly larger $d = \Omega(\log k)$.
However, this doesn't bother us since a different bound $\Phi(T) \ge 0.35\log k$ is known by a result of \cite{kush2023tree} for different parameter $\Phi(T)$ satisfying $\Phi(T) \ge \Psi(T)$.
As we describe later in Corollary \ref{cor:full-range}, we end up with tradeoffs $n^{\Omega(d(k^{1/d}-1))}$ and $n^{\Omega(d(k^{1/2d}-1))}$ for pathset complexity, which are tight for all $d$.  In particular, these tradeoffs both converge to $n^{\Omega(\log k)}$ as $d \to \infty$.

\subsection{Matching upper bounds}

The following 
Examples \ref{ex:tight1} and \ref{ex:tight2} show that 
the lower bounds of Theorem \ref{thm:tradeoff} are tight up to a constant factor.

\begin{ex}[Tightness of Theorem \ref{thm:tradeoff}(I)]\label{ex:tight1}
Let $P$ be a path of length $k$, where we assume $\ell \defeq k^{1/d}$ is an integer.  Consider the covering $P_1 \cup P_2 \cup \dots \cup P_\ell = P$ where $P_1,P_2,\dots,P_\ell$ are consecutive edge-disjoint paths of length $k/\ell$ ($=k^{(d-1)/d}$). 

Let $T = \sqq{T_1,T_2,\dots,T_k}$ be the $P$-join tree of $\sqq{}$-depth $d$ where each $T_i$ is the $P_i$-join tree of $\sqq{}$-depth $d-1$ constructed recursively by the same process.
Note that $\Psi(T) \le \ell/2 + \max_i \Psi(T_i)$. It follows by induction that $\Psi(T) \le d\ell/2 = dk^{1/d}/2$.
\end{ex}

\begin{ex}[Tightness of Theorem \ref{thm:tradeoff}(II)]\label{ex:tight2}
Let $P$ be a path of length $k$, where we assume $\ell \defeq k^{1/2d}$ is an integer. 
Consider the covering $\bigcup_{(i,j) \in [\ell]^2} P_{i,j} = P$ where  
\[
  \underbrace{P_{1,1},P_{1,2},\dots,P_{1,\ell}},\ 
  \underbrace{P_{2,1},P_{2,2},\dots,P_{2,\ell}},\ 
  \dots,\ 
  \underbrace{P_{\ell,1},P_{\ell,2},\dots,P_{\ell,\ell}}
\]
are consecutive edge-disjoint paths of length $k/\ell^2$ ($= k^{(d-1)/d}$). Note that $P_{i,1} \cup P_{i,2} \cup \dots \cup P_{i,\ell}$ is a path of length $k/\ell$, whilel $P_{1,j} \cup P_{2,j} \cup \dots \cup P_{\ell,j}$ is a vertex-disjoint union of $\ell$ paths of length $k/\ell^2$.

Let 
\[
  T = \sem{
  \underbrace{T_{1,1},T_{2,1},\dots,T_{\ell,1}},\ 
  \underbrace{T_{1,2},T_{2,2},\dots,T_{\ell,2}},\ 
  \dots,\ 
  \underbrace{T_{1,\ell},T_{2,\ell},\dots,T_{\ell,\ell}}
  }
\]
be the $P$-join tree of $\sqq{}$-depth $d$ where each $T_{i,j}$ is the $P_{i,j}$-join tree of $\semempty$-depth $d-1$ constructed recursively by the same process. 
Note that $\Psi(T) \le 2\ell + \max_{i,j} \Psi(T_{i,j})$. It follows by induction that $\Psi(T) \le 2d\ell = 2dk^{1/2d}$.
\end{ex}

\subsection{Numerical inequality used in the induction}

Both inequalities of Theorem \ref{thm:tradeoff} are proved by induction on the parameter $d$.  The following numerical inequality plays a key role in the induction step.

\begin{la}\label{la:numerical}
For all real numbers $x_1,\dots,x_m,y_1,\dots,y_m \ge 0$ and $d > 1$, 
\[
  \max_{j \in \{1,\dots,m\}}
  \Bigg(
  (d-1)x_j^{1/(d-1)} + 
  \sum_{i=1}^j y_i
  \Bigg)
  \ge
  d\Bigg(\frac{1}{\Exp}\sum_{j=1}^m x_j y_j\Bigg)^{1/d}.
\]
\end{la}

\begin{proof} 
Let
\[
  \vphantom{\big|} 
  Y_j &\defeq 
  \sum_{i=1}^j y_i \quad\text{ for $j \in \{0,\dots,m\}$},\\
  \vphantom{\big|}
  Z &\defeq x_1y_1 + \dots + x_m y_m.
\]
Assume that $Z > 0$, since otherwise the inequality is trivially valid.

Toward a contradiction, assume that for all $j \in \{1,\dots,m\}$, we have
\[
  (d-1)x_j^{1/(d-1)} + Y_j
  <
  d\left(\frac{Z}{\Exp}\right)^{1/d}.
\]
It follows that
\[
  x_j y_j
  &\le 
  \frac{1}{(d-1)^{d-1}}
  \bigg(d\bigg(\frac{Z}{\Exp}\bigg)^{1/d} - Y_j\bigg)^{d-1} y_j\\
  &=
  \frac{1}{(d-1)^{d-1}}  
  \int_{Y_{j-1}}^{Y_j} \bigg(d\bigg(\frac{Z}{\Exp}\bigg)^{1/d} - Y_j\bigg)^{d-1}\ \mr dw\\
  &\le
  \frac{1}{(d-1)^{d-1}} 
  \int_{Y_{j-1}}^{Y_j} \bigg(d\bigg(\frac{Z}{\Exp}\bigg)^{1/d} - w\bigg)^{d-1}\ \mr dw.
\]

We now get a contradiction $Z < Z$ as follows:
\[
  Z 
  = 
  \sum_{j=1}^m x_jy_j
  &\le 
  \frac{1}{(d-1)^{d-1}}
  \int_{0}^{Y_m} \bigg(d\bigg(\frac{Z}{\Exp}\bigg)^{1/d} - w\bigg)^{d-1}\ \mr dw\\
  &\le
  \frac{1}{(d-1)^{d-1}}
  \int_{0}^{d(Z/\Exp)^{1/d}} \bigg(d\bigg(\frac{Z}{\Exp}\bigg)^{1/d} - w\bigg)^{d-1}\ \mr dw
  &&\begin{aligned}
  &\text{(since $Y_m \le d(Z/\Exp)^{1/d}$ by}\\
  &\text{\ case $j=m$ of our assumption)} 
  \end{aligned}\\
  &=
  \frac{1}{(d-1)^{d-1}}
  \frac{1}{d}
  \bigg(d\bigg(\frac{Z}{\Exp}\bigg)^{1/d}\bigg)^d
  &&\begin{aligned}
  &\text{(since $\ts\int_0^C (C-w)^{d-1}\: \mr dw = \frac{1}{d}C^{1/d}$}\\
  &\text{ for  $C\ge 0$)}
  \end{aligned}\\
  &=
  \bigg(1 + \frac{1}{d-1}\bigg)^{d-1}
  \frac{Z}{\Exp}\\
  &<
  Z
  &&\text{(since $1+C \le \exp(C)$ for $C \ge 0$).}\qedhere
\]
\end{proof}

Lemma \ref{la:numerical} is invoked once and twice, respectively, in the induction step of lower bounds (I) and (II) of Theorem \ref{thm:tradeoff}.  This accounts for the different values $1/d$ vs.\ $1/2d$ in the exponent of $k$ in these two bounds.

We remark that in the simplest case $m=1$ of Lemma \ref{la:numerical}, one can show a stronger bound (without the constant $1/\Exp$) by elementary calculus.
We record this claim as a separate lemma, since we will use it later on (as one of the two applications of Lemma \ref{la:numerical} in the proof of Theorem \ref{thm:tradeoff}(II)).

\begin{la}[The $m=1$ case of Lemma \ref{la:numerical}]\label{la:numerical0}
For all real numbers $x,y \ge 0$ and $d > 1$, we have
\[
  (d-1)x^{1/(d-1)} + y
  &\ge
  d(xy)^{1/d}.\qedhere
\]
\end{la}

We remark that Lemma \ref{la:numerical0} is a common inequality that shows up in formula size-depth tradeoffs including \cite{lagarde2019lower,rossman2018average}.
We suspect that the more general inequality given by Lemma \ref{la:numerical} 
might find applications beyond the present paper to additional formula tradeoffs which involve the function $dk^{1/d}$.

\section{Join tree tradeoff (I)}\label{sec:tradeoff1}

In this section we prove inequalities (I) of Lemmas~\ref{la:pre-pi-sigma}--\ref{la:pi-sigma} and Theorem~\ref{thm:tradeoff}, which we eventually use to prove tradeoffs (I)$^+$ and (I)$^-$ of Theorem \ref{thm:AC0tradeoffs} for monotone and non-monotone $\SACzero$ circuits.

\subsection{Proof of Pre-Main Lemma \ref{la:pre-pi-sigma}(I)}\label{sec:greedy}

The following notion of $\vv\Delta$-greedy sequences of subgraphs of $\Path_\Z$ plays a key role in our proof of Lemma \ref{la:pre-pi-sigma}(I).

\begin{df}[$\vv\Delta$-greedy sequence]
For finite subgraphs $F$ and $G_1,\dots,G_m$ of $\Path_\Z$, we say that the sequence $(G_1,\dots,G_m)$ is {\em $\vv\Delta$-greedy over $F$} if 
\[
  \Delta(G_j \ominus (F \cup G_1 \cup \dots \cup G_{j-1})) \ge \Delta(G_l \ominus (F \cup G_1 \cup \dots \cup G_{j-1}))
  \text{ for all } 1 \le j < l \le m.
\]
In the case $F = \emptyset$, we simply say that $(G_1,\dots,G_m)$ is {\em $\vv\Delta$-greedy}.
\end{df}

Note that for any $F$ and finite family $\mc G$ of finite subgraphs of $\Path_\Z$, there exists an enumeration of $(G_1,\dots,G_m)$ of $\mc G$ which is $\vv\Delta$-greedy with respect to $F$.
We next point out that if $(G_1,\dots,G_m)$ is $\vv\Delta$-greedy over $F$, then 
\[
  \Delta(G_1 \ominus F) \ge \Delta(G_2 \ominus (F \cup G_1)) 
  \ge \Delta(G_3 \ominus (F \cup G_1 \cup G_2)) \ge \dots \ge 
  \Delta(G_m \ominus (F \cup G_1 \cup \dots \cup G_{m-1})).
\]
That is, $\vv\Delta$-greediness implies that the sequence $(t_1,\dots,t_m)$ is non-decreasing where $t_j = \Delta(G_j \ominus (F \cup G_1 \cup \dots \cup G_{j-1}))$. However, this is not sufficient, since $\Delta$-greediness also requires that no transposition of two graphs in the sequence $G_1,\dots,G_m$ increases the value of $(t_1,\dots,t_m)$ under the lexicographic order.

\begin{ex}\label{ex:greedy}
Illustrated below is a graph $P \subset \Path_\Z$ (on top) with $41$ edges, formed by a union of edge-disjoint subgraphs $G_1,\dots,G_{27}$. 
The sequence $G_1,\dots,G_{27}$ is $\vv\Delta$-greedy 
and has $\Delta(G_1,\dots,G_{27}) = 10$. 
Graphs $G_1,\dots,G_7$ (below $P$) each have three edges, while $G_8,\dots,G_{27}$ each have a single edge.
The components of $\Delta(G_j \ominus (G_1 \cup \dots \cup G_{j-1}))$, $j \in \{1,\dots,7\}$, that contribute to $\Delta(G_1,\dots,G_{27})$ are circled.\medskip
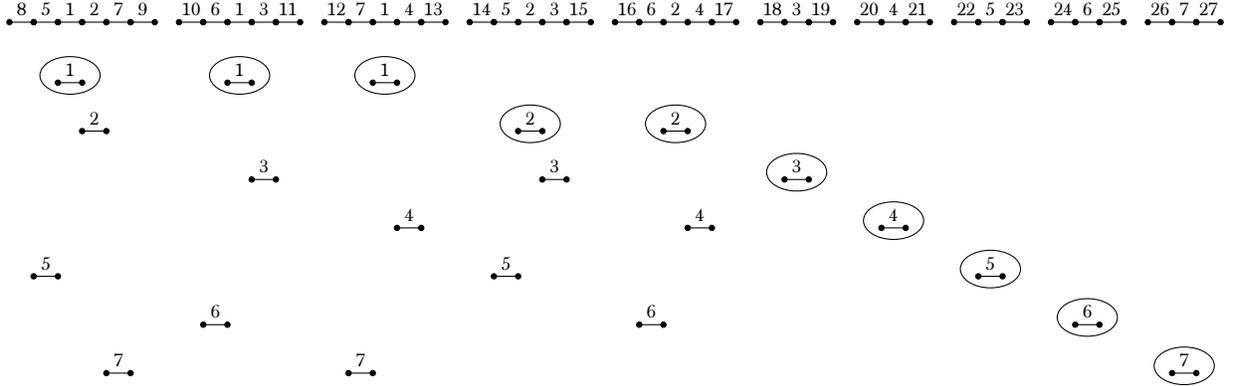
\begin{center}
\scriptsize
\begin{tikzpicture}[scale = .322]
  \tikzstyle{v}=[circle, draw, fill, 
  inner sep=0pt, minimum width=2pt]

  \def \x {.5}
  \def \y {-2}
  \def \z {-2*1.3}
  \def \q {1}
 
  \def \t {1} 
  \def \a {3} \def \b {9+\q} \def \c {14+2*\q}
  \draw[color=black] (\a-1,\x) node [v] {} -- (\a,\x) node [v] {}  node [midway, above] {\t}; 
  \draw[color=black] (\b-1,\x) node [v] {} -- (\b,\x) node [v] {}  node [midway, above] {\t}; 
  \draw[color=black] (\c-1,\x) node [v] {} -- (\c,\x) node [v] {}  node [midway, above] {\t}; 
  \def \t {2} 
  \def \a {4} \def \b {19+3*\q} \def \c {24+4*\q}
  \draw[color=black] (\a-1,\x) node [v] {} -- (\a,\x) node [v] {}  node [midway, above] {\t}; 
  \draw[color=black] (\b-1,\x) node [v] {} -- (\b,\x) node [v] {}  node [midway, above] {\t}; 
  \draw[color=black] (\c-1,\x) node [v] {} -- (\c,\x) node [v] {}  node [midway, above] {\t}; 
  \def \t {3} 
  \def \a {10+\q} \def \b {20+3*\q} \def \c {28+5*\q}
  \draw[color=black] (\a-1,\x) node [v] {} -- (\a,\x) node [v] {}  node [midway, above] {\t}; 
  \draw[color=black] (\b-1,\x) node [v] {} -- (\b,\x) node [v] {}  node [midway, above] {\t}; 
  \draw[color=black] (\c-1,\x) node [v] {} -- (\c,\x) node [v] {}  node [midway, above] {\t}; 
  \def \t {4} 
  \def \a {15+2*\q} \def \b {25+4*\q} \def \c {31+6*\q}
  \draw[color=black] (\a-1,\x) node [v] {} -- (\a,\x) node [v] {}  node [midway, above] {\t}; 
  \draw[color=black] (\b-1,\x) node [v] {} -- (\b,\x) node [v] {}  node [midway, above] {\t}; 
  \draw[color=black] (\c-1,\x) node [v] {} -- (\c,\x) node [v] {}  node [midway, above] {\t}; 
  \def \t {5} 
  \def \a {2} \def \b {18+3*\q} \def \c {34+7*\q}
  \draw[color=black] (\a-1,\x) node [v] {} -- (\a,\x) node [v] {}  node [midway, above] {\t}; 
  \draw[color=black] (\b-1,\x) node [v] {} -- (\b,\x) node [v] {}  node [midway, above] {\t}; 
  \draw[color=black] (\c-1,\x) node [v] {} -- (\c,\x) node [v] {}  node [midway, above] {\t}; 
  \def \t {6} 
  \def \a {8+\q} \def \b {23+4*\q} \def \c {37+8*\q}
  \draw[color=black] (\a-1,\x) node [v] {} -- (\a,\x) node [v] {}  node [midway, above] {\t}; 
  \draw[color=black] (\b-1,\x) node [v] {} -- (\b,\x) node [v] {}  node [midway, above] {\t}; 
  \draw[color=black] (\c-1,\x) node [v] {} -- (\c,\x) node [v] {}  node [midway, above] {\t}; 
  \def \t {7} 
  \def \a {5} \def \b {13+2*\q} \def \c {40+9*\q}
  \draw[color=black] (\a-1,\x) node [v] {} -- (\a,\x) node [v] {}  node [midway, above] {\t}; 
  \draw[color=black] (\b-1,\x) node [v] {} -- (\b,\x) node [v] {}  node [midway, above] {\t}; 
  \draw[color=black] (\c-1,\x) node [v] {} -- (\c,\x) node [v] {}  node [midway, above] {\t}; 
  
  \def\n{{8,9,10,11,12,
  13,14,15,16,17,
  18,19,20,21,22,
  23,24,25,26,27}}
  \def\A{{1,6,7+\q,11+\q,12+2*\q,
  16+2*\q,17+3*\q,21+3*\q,22+4*\q,26+4*\q,
  27+5*\q,29+5*\q,30+6*\q,32+6*\q,33+7*\q,
  35+7*\q,36+8*\q,38+8*\q,39+9*\q,41+9*\q}}
  
  \foreach \i in {0,1,...,19}
  {
  \pgfmathsetmacro{\a}{\A[\i]}
  \pgfmathsetmacro{\t}{\n[\i]}
  \draw[color=black] (\a-1,\x) node [v] {} -- (\a,\x) [] node [v] {}  node [midway, above] {\t}; 
  }
  
  \def \sss {.3}
  \def \ttt {.8cm}
  \def \uuu {.5cm}
 
  \def \t {1} 
  \def \a {3} \def \b {9+\q} \def \c {14+2*\q}
  \draw[color=black] (\a-1,\t*\y) node [v] {} -- (\a,\t*\y) node [v] {}  node [midway, above] {\t}; 
  
  \node[ellipse,
    draw, 
    minimum width = \ttt, 
    minimum height = \uuu] (e) at (\a-1+0.5,\t*\y+\sss) {};
  
  \draw[color=black] (\b-1,\t*\y) node [v] {} -- (\b,\t*\y) node [v] {}  node [midway, above] {\t}; 
  
    \node[ellipse,
    draw, 
    minimum width = \ttt, 
    minimum height = \uuu] (e) at (\b-1+0.5,\t*\y+\sss) {};
    
  \draw[color=black] (\c-1,\t*\y) node [v] {} -- (\c,\t*\y) node [v] {}  node [midway, above] {\t}; 
  
    \node[ellipse,
    draw, 
    minimum width = \ttt, 
    minimum height = \uuu] (e) at (\c-1+0.5,\t*\y+\sss) {};
  
  \def \t {2} 
  \def \a {4} \def \b {19+3*\q} \def \c {24+4*\q}
  \draw (\a-1,\t*\y) node [v] {} -- (\a,\t*\y) node [v] {}  node [midway, above] {\t}; 
  \draw[color=black] (\b-1,\t*\y) node [v] {} -- (\b,\t*\y) node [v] {}  node [midway, above] {\t}; 
  
    \node[ellipse,
    draw, 
    minimum width = \ttt, 
    minimum height = \uuu] (e) at (\b - 0.5,\t*\y+\sss) {};
  
  \draw[color=black] (\c-1,\t*\y) node [v] {} -- (\c,\t*\y) node [v] {}  node [midway, above] {\t}; 
  
    \node[ellipse,
    draw, 
    minimum width = \ttt, 
    minimum height = \uuu] (e) at (\c - 0.5,\t*\y+\sss) {};
    
  \def \t {3} 
  \def \a {10+\q} \def \b {20+3*\q} \def \c {28+5*\q}
  \draw (\a-1,\t*\y) node [v] {} -- (\a,\t*\y) node [v] {}  node [midway, above] {\t}; 
  \draw (\b-1,\t*\y) node [v] {} -- (\b,\t*\y) node [v] {}  node [midway, above] {\t}; 
  \draw[color=black] (\c-1,\t*\y) node [v] {} -- (\c,\t*\y) node [v] {}  node [midway, above] {\t}; 
  
    \node[ellipse,
    draw, 
    minimum width = \ttt, 
    minimum height = \uuu] (e) at (\c - 0.5,\t*\y+\sss) {};
  
  \def \t {4} 
  \def \a {15+2*\q} \def \b {25+4*\q} \def \c {31+6*\q}
  \draw (\a-1,\t*\y) node [v] {} -- (\a,\t*\y) node [v] {}  node [midway, above] {\t}; 
  \draw (\b-1,\t*\y) node [v] {} -- (\b,\t*\y) node [v] {}  node [midway, above] {\t}; 
  \draw[color=black] (\c-1,\t*\y) node [v] {} -- (\c,\t*\y) node [v] {}  node [midway, above] {\t}; 
  
    \node[ellipse,
    draw, 
    minimum width = \ttt, 
    minimum height = \uuu] (e) at (\c - 0.5,\t*\y+\sss) {};
  \def \t {5} 
  \def \a {2} \def \b {18+3*\q} \def \c {34+7*\q}
  \draw (\a-1,\t*\y) node [v] {} -- (\a,\t*\y) node [v] {}  node [midway, above] {\t}; 
  \draw (\b-1,\t*\y) node [v] {} -- (\b,\t*\y) node [v] {}  node [midway, above] {\t}; 
  \draw[color=black] (\c-1,\t*\y) node [v] {} -- (\c,\t*\y) node [v] {}  node [midway, above] {\t}; 
  
    \node[ellipse,
    draw, 
    minimum width = \ttt, 
    minimum height = \uuu] (e) at (\c - 0.5,\t*\y+\sss) {};
  \def \t {6} 
  \def \a {8+\q} \def \b {23+4*\q} \def \c {37+8*\q}
  \draw (\a-1,\t*\y) node [v] {} -- (\a,\t*\y) node [v] {}  node [midway, above] {\t}; 
  \draw (\b-1,\t*\y) node [v] {} -- (\b,\t*\y) node [v] {}  node [midway, above] {\t}; 
  \draw[color=black] (\c-1,\t*\y) node [v] {} -- (\c,\t*\y) node [v] {}  node [midway, above] {\t}; 
  
    \node[ellipse,
    draw, 
    minimum width = \ttt, 
    minimum height = \uuu] (e) at (\c - 0.5,\t*\y+\sss) {};
  \def \t {7} 
  \def \a {5} \def \b {13+2*\q} \def \c {40+9*\q}
  \draw (\a-1,\t*\y) node [v] {} -- (\a,\t*\y) node [v] {}  node [midway, above] {\t}; 
  \draw (\b-1,\t*\y) node [v] {} -- (\b,\t*\y) node [v] {}  node [midway, above] {\t}; 
  \draw[color=black] (\c-1,\t*\y) node [v] {} -- (\c,\t*\y) node [v] {}  node [midway, above] {\t}; 
  
    \node[ellipse,
    draw, 
    minimum width = \ttt, 
    minimum height = \uuu] (e) at (\c - 0.5,\t*\y+\sss) {};
\end{tikzpicture}
\medskip
\end{center}

To check that $(G_1,\dots,G_{27})$ is $\vv\Delta$-greedy, it may be seen by inspection that
\[
  &&&&&&&&\vphantom{\big|}
  3 &= \Delta(G_1)
   &&\ge 
   \Delta(G_j) 
   &&\text{for } 2\le j \le 27,&&&&&&&&\\
  \vphantom{\big|}
  &&&&&&&&2 
  &= \Delta(G_2 \ominus G_1) 
  &&\ge \Delta(G_j \ominus G_1)  
  &&\text{for } 3 \le j \le 27,\\
  \vphantom{\big|}
  &&&&&&&&1 &= \Delta(G_3 \ominus (G_1 \cup G_2)) 
  &&\ge \Delta(G_j \ominus (G_1 \cup G_2))
  &&\text{for } 4 \le j \le 27,\\
  \vphantom{\big|}
  &&&&&&&&1 &= \Delta(G_4 \ominus (G_1 \cup G_2 \cup G_3)) 
  &&\ge 
  \Delta(G_j \ominus (G_1 \cup G_2 \cup G_3))
  &&\text{for } 5 \le j \le 27,
\]
and so on.

This example is the $t=3$ case of a more general construction. For every positive integer $t$, there is a $\vv\Delta$-greedy sequence $G_1,\dots,G_m$ 
with $\lambda(G_1)=\dots=\lambda(G_m)=1$ and $\Delta(G_1) = t$ and
\[
  m &= t + (t+1) + (t+2)(t+1) &&= t^2+5t+3, 
  \vphantom{\bigg|}\\
  \|G_1 \cup \dots \cup G_m\| &= \sum_{j=1}^m \|G_j\| = 
  \underbrace{t+\dots+t}_{t \tu{ times}}\ \ +\ \   
  \underbrace{t+\dots+t}_{t+1 \tu{ times}}\ \ +\ \ 
  \underbrace{1+\dots+1}_{\!\!\!\!(t+2)(t+1) \tu{ times}\!\!\!\!}
  &&= 3t^2 + 4t + 2,
  \vphantom{\big|}\\
  \Delta(G_1,\dots,G_m) &= 
  t + (t-1) + (t-2) + \dots + 2 + 1 \ \ +\ \  
  \underbrace{1+\dots+1}_{t+1 \tu{ times}}\ \ +\ \   
  \underbrace{0+\dots+0}_{\!\!\!\!(t+2)(t+1) \tu{ times}\!\!\!\!}
  &&=
  \frac{(t+2)(t+1)}{2}.\vphantom{\big|}\qquad\quad
\]
As the following theorem shows, this construction achieves the maximum possible ratio of $\|G_1 \cup \dots \cup G_m\|$ to $\Delta(G_1,\dots,G_m)$
for any given $t$.
\end{ex}

\begin{thm}\label{thm:greedy}
Suppose $G_1,\dots,G_m$ are graphs such that $\lambda(G_1)=\dots=\lambda(G_m)=1$ and $(G_1,\dots,G_m)$ is $\vv\Delta$-greedy over $\emptyset$.
Let $t = \Delta(G_1)$.
Then
\[
  \|G_1 \cup \dots \cup G_m\|
  \le
  \frac{2(3t^2 + 4t + 2)}{(t+2)(t+1)} 
  \vv\Delta(G_1,\dots,G_m).
\]
Moreover, 
this inequality is tight for all $t \ge 1$.
\end{thm}

Note that Lemma \ref{la:pre-pi-sigma}(I) follows directly from this theorem, since $\frac{2(3t^2 + 4t + 2)}{(t+2)(t+1)} > 6$ for all $t \ge 1$.

\begin{proof}
The theorem is proved by analyzing a linear program with $t+1$ constraints and $C_1+\dots+C_{t+1} = \Theta(t^{-3/2} 4^t)$ variables, where $C_n = \frac{1}{n+1}\binom{2n}{n}$ is the $n$th Catalan number.
This linear program turns out to have an integral solution that corresponds to the construction given in Example \ref{ex:greedy}.

Without loss of generality, assume that for all $j \in [m]$, each edge of $G_j$ has at most one endpoint in $V(G_1 \cup \dots \cup G_{j-1})$. In particular, $G_1,\dots,G_m$ are edge-disjoint and no two connected components of $G_1 \cup \dots \cup G_{j-1}$ belong to the same connected component of $G_1 \cup \dots \cup G_j$.

For $j \in [m]$ and $s \in \{0,\dots,t\}$, we say that a graph $G_j$ has {\em level $s$} if $\Delta(G_j \ominus (G_1 \cup \dots \cup G_{j-1})) = t-s$. In the sequence $G_1,\dots,G_m$, the level $0$ graphs come first, followed by level $1$, etc.
For a graph $G_j$ of level $s$, we define the {\em profile} of $G_j$ as the sequence $(a_1,\dots,a_s) \in \N^s$ where $a_r$ is the number of vertices that $G_j$ shares with graphs of level $s-r$. For example, in Example \ref{ex:greedy} we have $t=3$ and  graphs $G_1,G_2,G_3,G_4$ have profiles $(),(1),(1,1),(1,1)$ respectively.

For $s \in \{0,\dots,t\}$ and $(a_1,\dots,a_s) \in \N^s$, let
\[
  \mb w_{a_1,\dots,a_s}
  &=
  \#\{\text{graphs of level $s$ with profile } (a_1,\dots,a_s)\}.
\]
Note that
\[
  \vv\Delta(G_1,\dots,G_m)
  &=
  \sum_{s=0}^t\ (t-s)
  \cdot \#\{
  \text{graphs of level } s\}
  =
  \sum_{s=0}^t \sum_{a_1,\dots,a_s} (t-s) \mb w_{a_1,\dots,a_s},\\
  \|G_1 \cup \dots \cup G_m\|
  &=
  \sum_{s=0}^t \sum_{a_1,\dots,a_s} (t-s+a_1+\dots+a_s)\mb w_{a_1,\dots,a_s}.
\]s

Let $\gamma = \gamma(t)$ be the constant
\[
\gamma \defeq 5 + \frac{2}{t+1} - \frac{12}{t+2}
\quad \bigg({=}\  \frac{2(3t^2 + 4t + 2)}{(t+2)(t+1)} - 1\bigg).
\] 
To prove the theorem (i.e.,\ the inequality $\|G_1\cup\dots\cup G_m\| \le (\gamma+1)\Delta(G_1,\dots,G_m)$), it suffices to show that
\[
  \sum_{s=0}^t \sum_{a_1,\dots,a_s} (t-s+a_1+\dots+a_s)  \mb w_{a_1,\dots,a_s}
  &\le
  \sum_{s=0}^t \sum_{a_1,\dots,a_s} (\gamma+1)(t-s) \mb w_{a_1,\dots,a_s}.
\]

The variables in our program will be indexed by Dyck sequences of length $\le t$. We say that $(a_1,\dots,a_s) \in \N^s$ is a {\em Dyck sequence} if $a_1+\dots+a_r \le r$ for all $r \in [s]$. (Note: The number of Dyck sequences of length $s$ is the Catalan number $C_{s+1}$.) 

For all $s \in \{0,\dots,t\}$ and $(a_1,\dots,a_s) \in \N^s$, we claim that
\begin{align}\label{eq:Dyck}
  \mb w_{a_1,\dots,a_s} \ne 0 
  \ &\Longrightarrow\ 
  (a_1,\dots,a_s) \text{ is a Dyck sequence}.
\end{align}
To see why, assume for contradiction that $\mb w_{a_1,\dots,a_s} \ne 0$ and $a_1 + \dots + a_r \ge r + 1$ for some $r \in [s]$. 
Fix the \underline{minimum} such $r$ and 
note that $a_1 + \dots + a_{r-1} \le r-1$ and $a_r \ge 2$.
Since $\mb w_{a_1,\dots,a_s} \ne 0$, there exists $j \in [m]$ such that $G_j$ has level $s$ and profile $(a_1,\dots,a_r)$.
Since $a_r$ is positive, there exists at least at least $i \in [m]$ such that $G_i$ has level $s-r$. 
Now fix the \underline{minimum} such $i$ and note that $i < j$.  Observe that 
\[
  \Delta(G_j \ominus (G_1 \cup \dots \cup G_{i-1}))
  &=
  \Delta(G_j \ominus (G_1 \cup \dots \cup G_{j-1})) + a_1 + \dots + a_r,
\]
since each edge of $G_j$ that contributes to the sum $a_1+\dots+a_r$ (i.e.,\ shares an endpoint with a graph of level between $s-1$ and $s-r$) contributes $0$ to $\Delta(G_j \ominus (G_1 \cup \dots \cup G_{i-1}))$ and $1$ to $\Delta(G_j \ominus (G_1 \cup \dots \cup G_{j-1}))$.
Since $a_1+\dots+a_r \ge r+1$, it follows that
\[
  \Delta(G_j \ominus (G_1 \cup \dots \cup G_{i-1}))
  &\ge
  \Delta(G_j \ominus (G_1 \cup \dots \cup G_{j-1})) + r + 1\\
  &=
  t - s + r + 1\\
  &>
  t - s + r\\
  &=
  \Delta(G_i \ominus (G_1 \cup \dots \cup G_{i-1})).
\]
However, this contradicts $\Delta$-greediness of the sequence $(G_1,\dots,G_m)$, proving the implication (\ref{eq:Dyck}).

We next observe $t+1$ linear inequalities among the $\mb w$-numbers:
\begin{align}
  \tag{$\ast_0$}
  \mb w_{()} &\ge 1,\\
  \tag{$\ast_r$}
    \qquad\sum_{s=r}^t \sum_{a_1,\dots,a_s} 
      a_{s-r} \mb w_{a_1,\dots,a_s}
  &\le
      \sum_{a_1,\dots,a_{r-1}}
      \Big(a_1+\dots+a_{r-1}+2(t-r+1)\Big)
      \mb w_{a_1,\dots,a_{r-1}}
      \quad\text{ for } r \in [t].
\end{align}
Here constraint ($\ast_0$) expresses that ``$\Delta(G_1) \ge t$''. For $r \in [t]$, the constraint ($\ast_r$) expresses ``the number of endpoints of graphs of level $\ge r$ that match with endpoints of graph of level $r-1$ is at most the number of {\em available} endpoints of level $s$'' (i.e.,\ $2(t-r+1)$ for each graph of level $r-1$, plus $1$ for each edge of level $r-1$ that has an endpoint in level $\le r-2$).

Let us now forget the definition of numbers $\mb w_{\vec a}$ in terms of graph $G_1,\dots,G_m$ and instead view $\mb w_{\vec a}$ as a family of real-valued variables indexed by Dyck sequences of length $\le t$.
Consider the linear program
\begin{align}
\tag{LP}
  &\max_{\vphantom{\big|}\mb w_{\vec a} \,\in\, \R_{\ge 0}
  \tu{ for each Dyck sequence $\vec a$ of length $\le t$}
  \,:\, 
  (\ast_0),\dots,(\ast_t)}\
  \sum_{s=0}^t \sum_{a_1,\dots,a_s}
  \Big(a_1+\dots+a_s-(t-s)\gamma\Big) \mb w_{a_1,\dots,a_s}.
\end{align}
Our goal is to show that this (LP) has value $0$. 

In fact, an optimal solution is given by the following integer-valued $\mb w$-vector. Let $\vec\alpha_0,\dots,\vec\alpha_t$ be the following Dyck sequences of length $0,\dots,t$ respectively:
\[
  \vec\alpha_0 = (),\quad
  \vec\alpha_1 = (1),\quad
  \vec\alpha_2 = (1,1),\quad
  \cdots,\quad
  \vec\alpha_{t-1} = (\underbrace{1,\dots,1}_{t-1}),\quad
  \vec\alpha_t = (1,\underbrace{0,\dots,0}_{t-1}).
\]
Then $\mb w$ is given by
\[
    \mb w_{\vec\alpha_0} = \mb w_{\vec\alpha_1} = \dots = \mb w_{\alpha_{t-2}} &= 1,\\
    \mb w_{\vec\alpha_{t-1}} &= t+1,\\
    \mb w_{\vec\alpha_t} &= (t+2)(t+1),\\
    \mb w_{\vec\beta} &= 0
    \quad
    \text{ for all } \vec\beta \notin \{\vec\alpha_0,\dots,\vec\alpha_k\}.\quad
\]
This vector corresponds to the greedy construction given in Example \ref{ex:greedy}. A simple calculation shows that (LP) has value $0$ under this $\mb w$.

To prove optimality, we consider the dual linear program given by
\begin{align}
\tag{dual LP}
  &\min_{\vphantom{\big|}\mb y_0,\dots,\mb y_t \,\in\, \R_{\ge 0} \,:\, 
  (\star_{\vec a}) \tu{ for each Dyck sequence $\vec a$ of length $\le t$}}\
  - \mb y_0
\end{align}
where constraints ($\star_{\vec a}$) are as follows:
\begin{align}
  \tag{$\star_{()}$}
  \mb y_0 + 2t\mb y_1 &
  \le t\gamma,\\
  \tag{$\star_{a_1,\dots,a_t}$}
  \sum_{r=1}^t a_{t-r}\mb y_r
      &
      \ge
      a_1+\dots+a_t,\\
  \tag{$\star_{a_1,\dots,a_s}$}
      \sum_{r=1}^s a_{s-r}\mb y_r
      -
      \Big(a_1+\dots+a_s+2(t-s)\Big) \mb y_{s+1}
      &
      \ge
      a_1+\dots+a_s - (t-s)\gamma
      \quad\text{ for } s \in [t-1].\qquad
\end{align}

We claim that 
the following vector $\mb y = (\mb y_0,\dots,\mb y_t) \in \R_{\ge 0}^{t+1}$ is an optimal solution to the dual LP:
\[
    \mb y_0
    &=
    0
    ,\\
    \mb y_r
    &=
    \frac{\gamma}{2}
    -
    \frac{(r - 1) (4 t + 2 - r)}{2 (2t + 2 - r) (2t + 1 - r)}
    \quad \text{ for } r \in [t]. 
\]
To see this, we first note that
\[
  \frac{5}{2} > \frac{\gamma}{2} = \mb y_1 > \mb y_2 > \dots > \mb y_t = 1.
\]
Next, note that constraint ($\star_{()}$) clearly holds with equality (i.e.,\ $\mb y_0 + 2t\mb y_1 = 0 + 2t\frac{\gamma}{2} = \gamma$). 
To check that $\mb y$ satisfies ($\star_{a_1,\dots,a_t}$) for Dyck sequences of length $t$, we have
\[
  \sum_{r=1}^t a_{t-r}\mb y_r - (a_1+\dots+a_t)
  &=
  \sum_{s=1}^t (\mb y_{t-s+1}-1) a_s\\
  &=
  \sum_{s=1}^t
  \left(\frac{\gamma}{2} - \frac{(t - s) (3 t + s + 1)}{2 (t + s + 1) (t + s)} - 1\right) a_s\\
  &\ge
  \sum_{s=1}^t
  \left(\frac{\gamma}{2} - \frac{(t - s) (3 t + s + 1)}{2 (t + s + 1) (t + s)} - 1\right)\\
  &=
  \frac{t^2 (t - 1)}{(t + 1) (t + 2)}\\
  &\ge
  0.
\]
The first inequality above follows from Karamata's inequality, using the fact that $(1,\dots,1)$ majorizes $(a_1,\dots,a_t)$ and the rational function 
\[
  \rho(x) \defeq \frac{(t - x) (3 t + x + 1)}{2 (t + x + 1) (t + x)}
\]
is convex over $x \in [0,t]$.  (Karamata's inequality implies that $\sum_{s=1}^t \rho(s) a_s \le \sum_{s=1}^t \rho(s)$ for every convex function $\rho$.)  This shows that $\mb y$ satisfies ($\star_{a_1,\dots,a_t}$). The argument that $\mb y$ satisfies ($\star_{a_1,\dots,a_s}$) for each Dyck sequence of length $s < t$ uses a similar application of Karamata's inequality.

Finally, 
note that the objective function $-\mb y_0$ of the dual LP at the feasible point $\mb y$ equals $0$. Therefore, both the primal and dual LPs have value $0$, which completes the proof.
\end{proof}

\begin{rmk}
Optimality of $\mb w$ and $\mb y$ may be seen by considering 
the primal LP in standard matrix form with columns restricted to the support of $\mb w$ (i.e.,\ with columns $\vec\alpha_0,\dots,\vec\alpha_t$). This looks like
\[
&\max_{\mb w_{\vec\alpha_0},\dots,\mb w_{\vec\alpha_t} \,\ge\, 0}\ 
f^\top 
\left(
\begin{array}{c}
\mb w_{\vec\alpha_0}\\
\mb w_{\vec\alpha_1}\\
\vdots\\
\mb w_{\vec\alpha_t}
\end{array}
\right)
\quad\text{ subject to }\quad
M
\left(
\begin{array}{c}
\mb w_{\vec\alpha_0}\\
\mb w_{\vec\alpha_1}\\
\vdots\\
\mb w_{\vec\alpha_t}
\end{array}
\right)
\le
\left(
\begin{array}{c}
-1\\
0\\
\vdots\\
0
\end{array}
\right)
\]
where
\[
f =
\left(
\begin{array}{c}
-t\gamma\\
1-(t-1)\gamma\\
2-(t-2)\gamma\\
\vdots\\
t-2-2\gamma\\
t-1-\gamma\\
1
\end{array}
\right),\qquad
M = 
\left(
\begin{array}{ccccccccc}
\ -1\ & 0 & \ \: 0 & \ \ \ \cdots  & 0 & 0 & \ 0 \\
-2t & 1 & \ \: 1 & \ \ \ \cdots & 1 & 1 & \ 0 \\
0 & -2t+1 & \ \: 1 & \ \ \ \cdots  & 1 & 1 & \ 0 \\
\vdots & \vdots & \ \: \vdots & \ \ \ \ddots  & \vdots & \vdots & \ \vdots \\
0 & 0 & \ \: 0 & \ \ \ \cdots  & 1 & 1 & \ 0 \\
0 & 0 & \ \: 0 & \ \ \ \cdots  & -t-2 & 1 & \ 0 \\
0 & 0 & \ \: 0 & \ \ \ \cdots  & 0 & -t-1 & \  \, \ 1\, \  
\end{array}
\right).
\]
Optimality of 
$\mb w$ and $\mb y$ is then certified by checking that
\[
M
\left(
\begin{array}{c}
\mb w_{\vec\alpha_0}\\
\mb w_{\vec\alpha_1}\\
\vdots\\
\mb w_{\vec\alpha_t}
\end{array}
\right)
=
\left(
\begin{array}{c}
-1\\
0\\
\vdots\\
0
\end{array}
\right),\qquad
M^\top
\left(
\begin{array}{c}
\mb y_0\\
\mb y_1\\
\vdots\\
\mb y_t
\end{array}
\right)
=
f,\qquad
f^\top 
\left(
\begin{array}{c}
\mb w_{\vec\alpha_0}\\
\mb w_{\vec\alpha_1}\\
\vdots\\
\mb w_{\vec\alpha_t}
\end{array}
\right)
=
-\mb y_0
&=
0.\qedhere
\] 
\end{rmk}

\begin{cor}\label{cor:greedy}
If $(G_1,\dots,G_m)$ is $\vv\Delta$-greedy over $F$, then
\[
  \frac{\|(G_1 \cup \dots \cup G_m) \ominus F\|}{\max_{j\in[m]} \lambda(G_j \ominus F)}
  \le
  6\vv\Delta(G_1,\dots,G_m\mid F).
\]
\end{cor}

\begin{proof}
Suppose $(G_1,\dots,G_m)$ is $\vv\Delta$-greedy over $F$. We can easily modify these graphs to obtains a sequence $(G_1',\dots,G_m')$ which is $\vv\Delta$-greedy over $\emptyset$ and satisfies
\begin{gather*}
\lambda(G_1') = \dots = \lambda(G_m') = 1,\\
\vv\Delta(G_1',\dots,G_m') = \Delta(G_1,\dots,G_m \mid F),\vphantom{\Big|}\\
\|G_1 \cup \dots \cup G_m\| \le \|G_1' \cup \dots \cup G_m'\|
\cdot
\max_{j\in[m]} \lambda(G_j \ominus F).
\end{gather*}
The idea is simple: we form $G_j'$ by contracting each component of $G_j \ominus (F \cup G_1 \cup \dots \cup G_{j-1})$ that is contained in $(G_1 \cup \dots \cup G_m) \ominus F$ to a single edge.
The bound then follows from Theorem \ref{thm:greedy}. 
\end{proof}
 
Corollary \ref{cor:greedy} is not exactly what we need for our proof of Lemma \ref{la:pi-sigma}(I) in the next subsection.  Instead of a bound on $\|(G_1 \cup \dots \cup G_m) \ominus F\|$, we require a bound on the potentially larger quantity $\|(G \ominus F) \cup \dots \cup (G_m \ominus F)\|$. 
Fortunately, LP duality again yields a lower bound. 

\begin{la}[Strengthening of Lemma \ref{la:pre-pi-sigma}(I)]\label{la:greedy2}
If $(G_1,\dots,G_m)$ is $\vv\Delta$-greedy over $F$, then
\[
  \frac{\|(G_1 \ominus F) \cup \dots \cup (G_m \ominus F)\|}{\max_{j\in[m]} \lambda(G_j \ominus F)}
  &\le
  5\Delta(F) + 
  6\vv\Delta(G_1,\dots,G_m\mid F).
\]
\end{la}

\begin{proof}
We slightly modify the linear program considered in Theorem \ref{thm:greedy}.
For $t = \Delta(G_1 \ominus F)$, we now index variables by sequences $(a_1,\dots,a_s;b)$ where $a_1,\dots,a_s$ is a Dyck sequence of length $\le t$ and $b \in \{0,\dots,2\Delta(F)\}$ (representing the number of endpoints of $F$ that a given $G_j$ covers).
In addition to (appropriately modified) constraints ($\ast_0$),\dots,($\ast_t$), we introduce an extra constraint
\begin{equation}\tag{$\ast_{t+1}$}
  \sum_{a_1,\dots,a_s}\,\sum_b\, \mb w_{a_1,\dots,a_s;b} \le 2\Delta(F).
\end{equation}
We maximize essentially the same objective function as before:
\[
  \sum_{s=0}^t \sum_{a_1,\dots,a_s}\,\sum_b\,
  \Big(a_1+\dots+a_s-(t-s)\gamma\Big) \mb w_{a_1,\dots,a_s;b}
\]
for the same constant $\gamma \defeq 5 + \frac{2}{t+1} - \frac{12}{t+2}$ ($<5$). The value of this LP gives an upper bound on
\[
  (\gamma+1)\cdot\vv\Delta(G_1,\dots,G_m\mid F) 
  - \frac{\|(G_1 \ominus F) \cup \dots \cup (G_m \ominus F)\|}{\max_{j\in[m]} \lambda(G_j \ominus F)}.
\]

It turns out that this value is $\gamma\cdot\Delta(F)$. However, unlike the LP of Theorem \ref{thm:greedy} (or this modified version in the case $\Delta(F)=0$), the primal optimal solution is no longer integral. 
Nevertheless, we can show that $\gamma\cdot\Delta(F)$ is an upper bound by considering the dual LP. 
The dual optimal solution $\mb y = (\mb y_0,\dots,\mb y_{t+1})$ happens to coincide on $\mb y_0,\dots,\mb y_t$ with the dual optimal solution of Theorem \ref{thm:greedy}; the extra coordinate $\mb y_{t+1}$ equals $\gamma/2$. We omit the calculations showing this, as the details are similar to Theorem \ref{thm:greedy}.
\end{proof}

\subsection{Proof of Main Lemma \ref{la:pi-sigma}(I)}
\label{sec:lambda-delta}

\newtheorem*{lambda-delta}{Main Lemma \ref{la:pi-sigma}(I)}

We shall now prove:

\begin{lambda-delta}[rephrased]
For every covering $\mc G$ of $\Path_k$, there exists a sequence $G_1,\dots,G_m$ of graphs in $\mc G$ (w.l.o.g.\ an enumeration of $\mc G$) such that
\[
  \vv{\lambda\Delta}(G_1,\dots,G_m)
  \ge
  \frac{k}{30}.
\]
\end{lambda-delta}

\begin{proof}
Let $r = \lceil \log(k+1)\rceil$.  We greedily construct integers $m_1,\dots,m_r \ge 0$ and a sequence
\[
  G_{1,1},\dots,G_{1,m_1},G_{2,1},\dots,G_{2,m_2},\cdots,G_{r,1},\dots,G_{r,m_r} \in \mc G
\]
such that, letting $F_0 = \emptyset$ and
\[
  \vphantom{\Big|}
  F_i &\defeq F_{i-1} \cup G_{i,1} \cup \dots \cup G_{i,m_i},\\
  G^{i,j} &\defeq G_{i,j} \ominus (F_{i-1} \cup G_{i,1} \cup \dots \cup G_{i,j-1}),
\]
the following properties hold for all $i \in [r]$ and $j \in [m_i]$:
\begin{enumerate}[\quad (i)\,]
  \item
    $\lambda(G \ominus F_i) < 2^{r-i}$ 
    for all 
    $G \in \mc G$, 
    $\vphantom{\big|}$
  \item
    $\lambda(G^{i,j}) \ge 2^{r-i}$,
    $\vphantom{\big|}$
  \item
    $\ds\Delta(G^{i,j})
    =
    \max_{\vphantom{\big|}
      G \in \mc G \,:\, \lambda(G \,\ominus\, (F_{i-1} \,\cup\, G_{i,1} \,\cup\, \dots \,\cup\, G_{i,j-1})) \,\ge\, 2^{r-i}
      } 
    \Delta(G \ominus (F_{i-1} \cup G_{i,1} \cup \dots \cup G_{i,j-1}))$.
    $\vphantom{\big|}$
\end{enumerate}
To see that such a sequence exists,
assume for $i \in [r]$ and $j \ge 1$ that we have constructed integers $m_1,\dots,m_{i-1}$ (but not yet $m_i$) and graphs $G_{1,1},\dots,G_{i,j-1}$ (but not $G_{i,j}$) which so far satisfy properties (i), (ii), (iii). There are now two possibilities for the next step in the construction:
\begin{itemize} 
  \item
    Suppose that there exists $G \in \mc G$ such that $\lambda(G \ominus (F_{i-1} \cup G_{i,1} \cup \dots \cup G_{i,j-1})) \ge 2^{r-i}$. 
    
    In this case, we pick $G_{i,j}$ to be any such $G$ which maximizes $\Delta(G \ominus (F_{i-1} \cup G_{i,1} \cup \dots \cup G_{i,j-1}))$.
    
    Note that properties (ii) and (iii) are satisfied with respect to $i$ and $j$.
  \item
    Suppose that for all $G \in \mc G$, we have $\lambda(G \ominus (F_{i-1} \cup G_{i,1} \cup \dots \cup G_{i,j-1})) < 2^{r-i}$. 
    
    In this case, we set $m_i \defeq j-1$.
    
    Note that property (i) is satisfied with respect to $i$.
\end{itemize}
For the base case $i=j=1$ of this inductive procedure, note that $\lambda(G \ominus F_0) = \lambda(G) \le k < 2^r$ for all $G \in \mc G$.\bigskip

Having constructed the sequence $G_{1,1},\dots,G_{r,m_r}$, we turn to showing that
\[
  \vv{\lambda\Delta}(G_{1,1},\dots,G_{r,m_r}) 
  \
  \bigg(\!{=}\ 
  \sum_{i\in[r]} \sum_{j\in[m_i]}\lambda(G^{i,j})\Delta(G^{i,j})
  \bigg)
  \ge \frac{k}{30}.
\]

Note that by property (ii), for all $i \in [r]$,
\begin{align}
\label{eq:ld3}
  2^{r-i} \vv\Delta(G_{i,1},\dots,G_{i,m_i} \mid F_i) 
  &=
  2^{r-i} \sum_{j \in [m_i]} \Delta(G^{i,j})
  \le
  \sum_{j \in [m_i]} \lambda(G^{i,j})\Delta(G^{i,j}).
\end{align}
We next prove two inequalities: 
\begin{align}
\label{eq:ld2}
  \sum_{i\in[r]} 2^{r-i} \Delta(F_i) 
  &\le
  \phantom{2}2\sum_{i\in[r]} \sum_{j\in[m_i]}\lambda(G^{i,j})\Delta(G^{i,j}),\\
\label{eq:ld1}
  \sum_{i\in[r]} \Big\|\bigcup_{j \in [m_i]} (G_{i,j} \ominus F_{i-1}) \Big\|
  &\le 
  22\sum_{i\in[r]} \sum_{j\in[m_i]}\lambda(G^{i,j})\Delta(G^{i,j}).
\end{align}

For inequality (\ref{eq:ld2}), noting that $\emptyset = F_0 \subseteq F_1 \subseteq \dots \subseteq F_r$, we have
\[
  \sum_{i\in[r]} 2^{r-i} \Delta(F_i) 
  &=
  \sum_{i\in[r]} 2^{r-i} \sum_{h \in [i]} \Delta(F_h \ominus F_{h-1})\\
  &\le 
  2\sum_{i\in[r]} 2^{r-i} \Delta(F_i \ominus F_{i-1})
  \\
  &=
  2\sum_{i\in[r]} 2^{r-i} \Delta((G_{i,1}\cup\dots\cup G_{i,m_i}) \ominus F_{i-1})\\
  &\le
  2\sum_{i\in[r]} 2^{r-i} \Delta(G_{i,1},\dots,G_{i,m_i} \mid F_{i-1})
  &&\text{(by Lemma \ref{la:delta-props}
  (a))}\\\
  &\le
  2\sum_{i\in[r]} \sum_{j \in [m_i]} \lambda(G^{i,j}) \Delta(G^{i,j})
  &&\text{(by (\ref{eq:ld3}))}.
\]

For inequality (\ref{eq:ld1}), note that condition (iii) implies that $(G_{i,1},\dots,G_{i,m_i})$ is $\vv\Delta$-greedy over $F_{i-1}$ for all $i \in [r]$. Therefore, by Lemma \ref{la:greedy2} and property (i), 
\[
  \Big\|\bigcup_{j \in [m_i]} (G_{i,j} \ominus F_{i-1}) \Big\|
  &\le
  2^{r-i+1}\Big(
    5\Delta(F_{i-1}) + 6\Delta(G_{i,1},\dots,G_{i,m_i}\mid F_{i-1})
  \Big).
\]
Inequality (\ref{eq:ld1}) now follows from (\ref{eq:ld3}) and (\ref{eq:ld2}).\bigskip

In the next part of the argument, we consider the {\em $2^{r-i+1}$-neighborhood} of each graph $F_i$. 
To define this formally: for a finite graph $P \subset \Path_\Z$, let us write $\Nbd_1(P)$ for the graph consisting of all edges $\Path_\Z$ incident to a vertex of $P$. 
For $t \in \N$, let 
\[
\Nbd_t(P) \defeq \underbrace{\Nbd_1(\ldots\Nbd_1}_{t \text{ times}}(P)\ldots).
\]
Note the inequalities $\Delta(\Nbd_t(P)) \le \Delta(P)$ and $\|\Nbd_t(P)\| \le 2t\Delta(P) + \|P\|$.

For finite graphs $P,Q \subset \Path_\Z$, let $P \setminus Q$ denote the graph with edge set $E(P)\setminus E(Q)$ (and no isolated vertices).  
Here is another obvious inequality, numbered for later reference:
\begin{equation}\label{eq:nbd}
   \|\Nbd_t(P) \setminus Q\| \le 2t\Delta(P) + \|P \setminus Q\|.
\end{equation}
Note that $P \ominus Q \subseteq P \setminus Q$. In the other direction, note that $P \setminus Q_1 \subseteq P \ominus Q_2$ if and only if 
($C \subseteq Q_1$ or $V(C) \cap V(Q_2) = \emptyset$) for every connected component $C$ of $P$.
\bigskip

Now comes a key claim: for all $i \in [r]$ and $G \in \mc G$, we have
\begin{equation}\label{eq:containment}
  G \setminus 
  \bigcup_{h \in [i]} 
  \Nbd_{2^{r-h+1}}(F_h)
  \subseteq
  G \ominus F_i.
\end{equation}
To see why, suppose $C$ is a connected component of $G$ such that $V(C) \cap V(F_i)$ is nonempty. 
We must show that $C \subseteq \bigcup_{h \in [i]} \Nbd_{2^{r-h+1}}(F_h)$. There are two cases:
\begin{itemize}
\item
If $\|C\| < 2^{r-i+1}$, then clearly $C \subseteq \Nbd_{2^{r-i+1}}(F_i)$ since $C$ is connected and shares a vertex with $F_i$.
\item
If $\|C\| \ge 2^{r-i+1}$, then we consider the unique $h \in [i-1]$ such that $2^{r-h} \le \|C\| < 2^{r-h+1}$.
(Note that $h$ is well-defined since $\|C\| \le k < 2^r$.)

By property (i), we have $\lambda(G \ominus F_h) < 2^{r-h}$. It follows that $V(C) \cap V(F_h)$ is nonempty, since otherwise $C \subseteq G \ominus F_h$ and we would have a contradictory inequality $2^{r-h} \le \|C\| \le \lambda(G \ominus F_h) < 2^{r-h}$.
Since $C$ is connected of length $< 2^{r-h+1}$ and shares a vertex with $F_h$, we conclude that $C \subseteq \Nbd_{2^{r-h+1}}(F_h)$.
\end{itemize}

We next observe that for all $G \in \mc G$, case $i=r$ of property (i) implies that $\lambda(G \ominus F_r) < 1$ and hence $G \ominus F_r = \emptyset$. Case $i=r$ of the containment (\ref{eq:containment}) now yields
\[
  G \setminus 
  \bigcup_{i \in [r]} 
  \Nbd_{2^{r-i+1}}(F_i)
  \subseteq
  G \ominus F_r
  =
  \emptyset,\quad\text{ that is, }\quad
  G \subseteq 
  \bigcup_{i \in [r]} 
  \Nbd_{2^{r-i+1}}(F_i).
\]
Since $\bigcup_{G \in \mc G} G = \Path_k$, it follows that
$
  \Path_k \subseteq \bigcup_{i \in [r]} 
  \Nbd_{2^{r-i+1}}(F_i).
$

The proof now finishes as follows:
\[
  k
  &\le
  \Big\|
  \bigcup_{i \in [r]} 
  \Nbd_{2^{r-i+1}}(F_i)
  \Big\|
 \\
  &\le
  \sum_{i\in[r]}
  \Big\|
  \Nbd_{2^{r-i+1}}(F_i)
  \setminus 
  \bigcup_{h \in [i-1]} 
  \Nbd_{2^{r-h+1}}(F_h)
  \Big\|\\
  &\le
  \sum_{i\in[r]}
  \Big(
  2^{r-i+2}\Delta(F_i)
  +
  \Big\|
  F_i
  \setminus 
  \bigcup_{h \in [i-1]} 
  \Nbd_{2^{r-h+1}}(F_h)
  \Big\|
  \Big)
  &&\text{(by (\ref{eq:nbd}))}
  \\
  &=
  4\sum_{i\in[r]}
  2^{r-i}\Delta(F_i)
  +
  \sum_{i\in[r]}
  \Big\|
  \bigcup_{j\in[m_j]}
  \Big(
  G_{i,j}
  \setminus 
  \bigcup_{h \in [i-1]} 
  \Nbd_{2^{r-h+1}}(F_h)
  \Big)
  \Big\|
  \\
  &\le
  4\sum_{i\in[r]}
  2^{r-i}\Delta(F_i)
  +
  \sum_{i\in[r]}
  \Big\|
  \bigcup_{j\in[m_j]}
  (
  G_{i,j}
  \ominus
  F_i 
  )
  \Big\|
  &&\text{(by (\ref{eq:containment}))}
  \\
  &\le
  30\sum_{i\in[r]} \sum_{j\in[m_i]}\lambda(G^{i,j})\Delta(G^{i,j})
  &&\text{(by (\ref{eq:ld2}) and (\ref{eq:ld1}))}\\
  &=
  30\,\vv{\lambda\Delta}(G_{1,1},\dots,G_{r,m_r}).&&\qedhere
\]
\end{proof}

\subsection{Proof of Theorem \ref{thm:tradeoff}(I)}
\label{sec:proof}\label{sec:proof1}

We will actually prove the following slightly more general version of Theorem \ref{thm:tradeoff}(I), whose statement is better suited to induction on $d$.

\newtheorem*{tradeoff1}{Main Theorem \ref{thm:tradeoff}(I)}

\begin{tradeoff1}[slightly more general restatement]
Let $P$ be a finite subgraph of $\Path_\Z$ and let $T$ be a $P$-join tree of $\sqq{}$-depth $d$.
Then
\[
  \Psi(T)
  \ge 
  \frac{1}{30\Exp} d\lambda(P)^{1/d} + \Delta(P) - d.
\]
\end{tradeoff1}

\begin{proof}
We argue by induction on $d$. To clarify the exposition, we give headers to the various steps in the argument.

\subsubsection*{The base case $d=1$}

(We could actually treat $d=0$ as the base case of the induction, but instead choose $d=1$ since it will help to illustrate the reduction to the case $\Delta(P)=1$ in the next step.)

For the base case $d=1$, let $\mc G$ be the rightmost branch covering of $T$. Note that $\mc G$ is the set of single-edge subgraphs of $P$. 

The graph $P$ is a disjoint union of paths $P_1,\dots,P_c$ where $c \defeq \Delta(P)$ and $\lambda(P) = \|P_1\|$ (without loss of generality). 
Let $S_1$ be a set consisting of $\lceil
\lambda(P)/2
\rceil$ vertex-disjoint edges in $P_1$, and let $S_2$ be a set containing one edge from each of $P_2,\dots,P_c$. Let $e_1,\dots,e_m$ be any enumeration of the edges of $P$ with respect to which the edges in $S_1 \cup S_2$ come before edges outside this set, and let $G_1,\dots,G_m$ be the corresponding enumeration of $\mc G$. We have
\[
  \vv\Delta(G_1,\dots,G_m) 
  \ge 
  |S_1| + |S_2|
  \ge 
  \frac{1}{2}\lambda(P) + c
  =
  \frac{1}{2}d\lambda(P)^{1/d} + \Delta(P) - d,
\]
which proves the base case since 
$1/2 > 1/30\Exp$.

\subsubsection*{Induction step: reduction to the case $\Delta(P) = 1$}

For the induction step, let $d \ge 2$ and assume that the theorem holds for $d-1$. We argue that it suffices to prove the induction step in the case that $\Delta(P) = 1$. To see this, we again decompose $P$ as a disjoint union of paths $P_1,\dots,P_c$ where $c \defeq \Delta(P)$ and $\lambda(P) = \|P_1\|$. 

Let $T_1$ be the $P_1$-join tree obtained from $T$ by relabeling to $\emptyset$ (the empty graph) any leaf originally labeled by an edge of $P_2 \cup \dots \cup P_c$. 
Assume that we are given a $T_1$-branch covering $\mc G_1$ of $P_1$ and an enumeration $H_1,\dots,H_m$ of $\mc G_1$ such that
\[
  \vv\Delta(H_1,\dots,H_m) \ge 
  \frac{1}{30\Exp} d\lambda(P_1)^{1/d} - d + 1.
\]

Let $G_1,\dots,G_m$ be the graphs corresponding to $H_1,\dots,H_m$ in the original tree $T$, and let $\mc G = \{G_1,\dots,G_m\}$. Note that $\mc G$ is a $T$-branch covering of $P$.  Also note that $H_j = G_j \cap P_1$ for all $j \in [m]$.  

Writing $G^j$ for $G_j \ominus (G_1 \cup \dots \cup G_{j-1})$ and $H^j$ for $H_j \ominus (G_1 \cup \dots \cup H_{j-1})$, we next observe that
\[
  \Delta(G^j)
  \ge
  \Delta(H^j)
  +
  \#\{i \in \{2,\dots,c\} : 
  j \tu{ is the first index such that } G^j \cap P_i \tu{ is nonempty}\}.
\]
It follows that
\[
  \Psi(T)
  \ge
  \vv\Delta(G_1,\dots,G_m) 
  =
  \sum_{j\in[m]} \Delta(G^j)
  \ge 
  \sum_{j\in[m]} \Delta(H^j)
  + c - 1
  =
  \vv\Delta(H_1,\dots,H_m)
  + 
  \Delta(P) - 1.
\]
Therefore, we have
$
  \vv\Delta(G_1,\dots,G_m) \ge 
  \frac{1}{30\Exp} d\lambda(P)^{1/d} + \Delta(P) - d
$
as required.
\bigskip

Having established that it suffices to prove the induction step in the case that $\Delta(P) = 1$, let us proceed under the assumption $P$ is connected.  Without loss of generality, let $P = \Path_k$ where $k = \lambda(P)$.  Thus, we assume that $T$ is a $\Path_k$-join tree of $\sqq{}$-depth $d$.  Our goal is to find a $T$-branch covering $\mc G$ of $\Path_k$ together with an enumeration $G_1,\dots,G_m$ of $\mc G$ satisfying
\begin{equation}
\label{eq:goal}
  \vv\Delta(G_1,\dots,G_m) 
  \ge
  \frac{1}{30\Exp} dk^{1/d} - d + 1.
\end{equation}

\subsubsection*{The right spine of $T$}

As a first step, we consider the ``right spine'' 
$T = \sqq{T_1,\dots,T_\ell}$ where $T_\ell$ is a leaf and each $T_j$ is a $B_j$-join tree of $\sqq{}$-depth $d-1$ for graphs $B_1 \cup \dots \cup B_\ell = \Path_\ell$.
(Note: $\{B_1,\dots,B_\ell\}$ is not necessarily the $T$-branch covering of $\Path_k$ that we will eventually produce.)

Applying our Main Lemma \ref{la:pi-sigma}(II)
to graphs $B_1,\dots,B_\ell$, we obtain a permutation $\pi : [\ell] \stackrel\cong\to [\ell]$ satisfying
\begin{equation}\label{eq:first-lambda-delta}
  \vv{\lambda\Delta}(B_{\pi(1)},\dots,B_{\pi(\ell)}) 
  \ge \frac{k}{30}.
\end{equation}
For $j \in [\ell]$, it will be convenient to write 
\[
  B^j \defeq B_j \ominus (B_{\pi(1)} \cup \dots \cup B_{\pi(\pi^{-1}(j)-1)}).
\]
(Note: $B^j$ is not to be confused with $B_j \ominus (B_1 \cup \dots \cup B_{j-1})$.) 
Under this reindexing, inequality (\ref{eq:first-lambda-delta}) is equivalent to
\begin{equation}\label{eq:lambda-delta}
  \sum_{j\in[\ell]} \lambda(B^j)  
  \Delta(B^j) 
  \ge \frac{k}{30}. 
\end{equation}

\subsubsection*{Graph sequences $G^1_1,\dots,G^j_{m_j}$}

For each $j \in [\ell]$, let $T^j$ be the $B^j$-join tree obtained from $T_j$ by 
relabeling to $\emptyset$ (the empty graph) any leaf originally labeled by an edge of $B_j \setminus B^j$.
Note that $T^j$ has $\sqq{}$-depth $d-1$, since relabeling does not increase $\sqq{}$-depth.

Applying the induction hypothesis to each $T^j$, 
we obtain a $T^j$-branch covering $\mc H^j$ of $B^j$ and an enumeration $\smash{H^j_1,\dots,H^j_{n_j}}$ of $\mc H^j$ such that 
\begin{equation}\label{eq:Delta-H}
  \vv\Delta(H^j_1,\dots,H^j_{n_j}) 
  \ge 
  \frac{1}{30\Exp} (d-1)\lambda(B^j)^{1/(d-1)} + \Delta(B^j) - d + 1.
\end{equation}

We now define a sequence of graph $G^j_1,\dots,G^j_{m_j}$ where $m_j = n_j + j - 1$
via the three bullet points below.
(See the diagram following the definition for a helpful example.)
\begin{itemize}
  \item
    Note that $\pi^{-1}(1),\dots,\pi^{-1}(j)$ are $j$ distinct numbers in $[m]$.
    
    Let $\tau_j : [j] \stackrel\cong\to [j]$ be the permutation defined by $\tau_j^{-1}(a) < \tau_j^{-1}(b)$ iff $\pi^{-1}(a) < \pi^{-1}(b)$ for all $a,b \in [j]$.
    
    For example, if $j = 4$ and $\pi^{-1}(1,2,3,4) = (4,7,2,6)$, then $\tau_j^{-1}(1,2,3,4) = (2,4,1,3)$ and $\tau_j(1,2,3,4) = (3,1,4,2)$.
  \item
    Let $j^\star = \tau_j^{-1}(j) \in [j]$.
    
    In our example, $j^\star = 3$.
  \item
    Let $G^j_1,\dots,G^j_{m_j}$ be the sequence 
    \[
      B_{\tau_j(1)},
      \dots,B_{\tau_j(j^\star - 1)},
      H^j_1,
      \dots,H^j_{n_j},
      B_{\tau_j(j^\star+1)},
      \dots,B_{\tau_j(j)}.
    \]
    Said differently: $G^j_1,\dots,G^j_{m_j}$ is obtained from the sequence $B_{\tau_j(1)},\dots,B_{\tau_j(j)}$ by removing the $j^\ast$-th entry (i.e.,\ the graph $B_{\tau_j(j^\star)} = B_j$) and inserting the subsequence $H^j_1,\dots,H^j_{n_j}$ in its place.
    
    In our example (with $j=4$ and $j^\star=3$), $G^j_1,\dots,G^j_{m_j}$ is the sequence $B_3,B_1,H^4_1,\dots,H^4_{n_4},B_2$.\bigskip
\end{itemize}

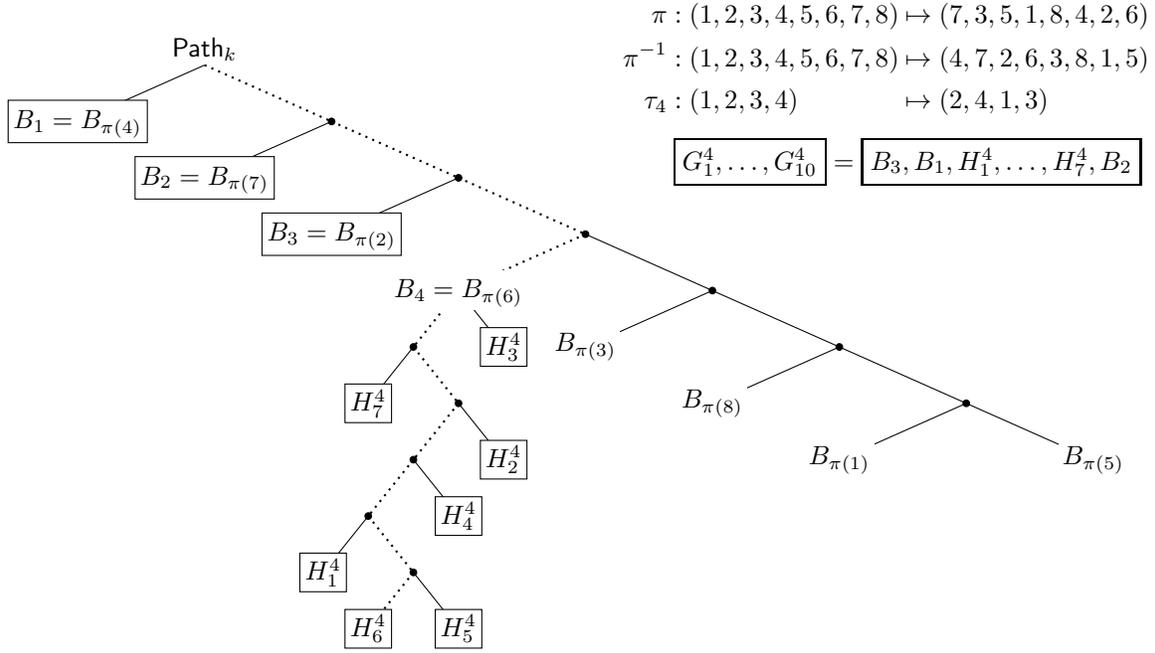
\begin{center}
\begin{tikzpicture}[scale = .75]
  \tikzstyle{v}=[circle, draw, fill, 
  inner sep=0pt, minimum width=2.5pt]
  \tikzstyle{w}=[ 
  fill=white, inner sep=1.5pt]
  
  \tikzstyle{ww}=[draw, fill=white, inner sep=2pt] 
  
  \def \z {2.25}
  \def \y {.8}
  \def \x {1}
  \def \a {2*\z}
  \def \b {-4*\x}
  
  \draw (3*\z+.5,-.5*\x+.6) node [right] {
  $\begin{aligned}[t]
    \vphantom{\big|}
    \pi : (1,2,3,4,5,6,7,8) &\mapsto (7,3,5,1,8,4,2,6)\\
    \vphantom{\big|}
    \pi^{-1} : (1,2,3,4,5,6,7,8) &\mapsto (4,7,2,6,3,8,1,5)\\
    \vphantom{\big|}
    \tau_4 : (1,2,3,4)\phantom{5,,6,7,8} &\mapsto (2,4,1,3)
  \end{aligned}$
  };

  \draw (3.35*\z+.4,-2.32*\x+.6) node [right] {\:\:$\fbox{$G^4_1,\dots,G^4_{10}$} = 
  \fbox{$B_3,B_1,H^4_1,\dots,H^4_7,B_2$}$};
  
  \draw (\z,-1*\x) node [v] {};
  \draw (2*\z,-2*\x) node [v] {};
  \draw (3*\z,-3*\x) node [v] {};
  \draw (4*\z,-4*\x) node [v] {};
  \draw (5*\z,-5*\x) node [v] {};
  \draw (6*\z,-6*\x) node [v] {};
 
  \draw[thick, dotted] (0,0*\x) -> (\z,-1*\x);
  \draw[thick, dotted] (\z,-1*\x) -> (2*\z,-2*\x);
  \draw[thick, dotted] (2*\z,-2*\x) -> (3*\z,-3*\x);
  \draw (3*\z,-3*\x) -> (4*\z,-4*\x);
  \draw (4*\z,-4*\x) -> (5*\z,-5*\x);
  \draw (5*\z,-5*\x) -> (6*\z,-6*\x);
  \draw (6*\z,-6*\x) -> (7*\z,-7*\x);
  
  \draw     (0,0*\x) -> (-\z,-1*\x);
  \draw   (\z,-1*\x) -> (0,-2*\x);
  \draw (2*\z,-2*\x) -> (\z,-3*\x);
  \draw[thick, dotted] (3*\z,-3*\x) -> (2*\z,-4*\x);
  \draw (4*\z,-4*\x) -> (3*\z,-5*\x);
  \draw (5*\z,-5*\x) -> (4*\z,-6*\x);
  \draw (6*\z,-6*\x) -> (5*\z,-7*\x);
  
  \draw (\a-\y,\b-1) node [v] {};
  \draw[thick, dotted] (\a,\b) -> (\a-\y,\b-1);
  \draw (\a+\y,\b-1) node [v] {};
  \draw (\a,\b) -> (\a+\y,\b-1);
  
  \draw (\a-2*\y,\b-2) node [v] {};
  \draw (\a-\y,\b-1) -> (\a-2*\y,\b-2);
  \draw (\a,\b-2) node [v] {};
  \draw[thick, dotted] (\a-\y,\b-1) -> (\a,\b-2);
  
  \draw (\a+\y,\b-3) node [v] {};
  \draw (\a,\b-2) -> (\a+\y,\b-3);
  \draw (\a-\y,\b-3) node [v] {};
  \draw[thick, dotted] (\a,\b-2) -> (\a-\y,\b-3);
  
  \draw (\a,\b-4) node [v] {};
  \draw (\a-\y,\b-3) -> (\a,\b-4);
  \draw (\a-2*\y,\b-4) node [v] {};
  \draw[thick, dotted] (\a-\y,\b-3) -> (\a-2*\y,\b-4);
  
  \draw (\a-3*\y,\b-5) node [v] {};
  \draw (\a-2*\y,\b-4) -> (\a-3*\y,\b-5);
  \draw (\a-\y,\b-5) node [v] {};
  \draw[thick, dotted] (\a-2*\y,\b-4) -> (\a-\y,\b-5);
  
  \draw (\a,\b-6) node [v] {};
  \draw (\a-\y,\b-5) -> (\a,\b-6);
  \draw (\a-2*\y,\b-6) node [v] {};
  \draw[thick, dotted] (\a-\y,\b-5) -> (\a-2*\y,\b-6);
  
  \draw  (-\z,-1*\x) node [ww] {$B_1 = B_{\pi(4)}\vphantom{\big|}$};
  \draw    (0,-2*\x) node [ww] {$B_2 = B_{\pi(7)}\vphantom{\big|}$};
  \draw   (\z,-3*\x) node [ww] {$B_3 = B_{\pi(2)}\vphantom{\big|}$};
  \draw (2*\z,-4*\x) node [w] {$B_4 = B_{\pi(6)}\vphantom{\big|}$};
  \draw (3*\z,-5*\x) node [w] {$B_{\pi(3)}$};
  \draw (4*\z,-6*\x) node [w] {$B_{\pi(8)}$};
  \draw (5*\z,-7*\x) node [w] {$B_{\pi(1)}$};
  \draw (7*\z,-7*\x) node [w] {$B_{\pi(5)}$};
  
  \draw (\a+\y,\b-1) node [ww] {$H^4_3$};
  \draw (\a-2*\y,\b-2) node [ww] {$H^4_7$};
  \draw (\a+\y,\b-3) node [ww] {$H^4_2$};
  \draw (\a,\b-4) node [ww] {$H^4_4$};
  \draw (\a-3*\y,\b-5) node [ww] {$H^4_1$};
  \draw (\a,\b-6) node [ww] {$H^4_5$};
  \draw (\a-2*\y,\b-6) node [ww] {$H^4_6$};
  
  \draw (0,0) node [w,above] {$\Path_k$};
\end{tikzpicture}\bigskip
\end{center}

\subsubsection*{The $T$-branch covering $\mc G_j = \{G_1,\dots,G_{m_j+1}\}$}

For each $j \in [\ell]$ and $s \in [m_j]$, let $H_{j,s}$ be the label in $T_j$ of the node labeled by $H^j_s$ in $T^j$. Note that $H^j_s = H_{j,s} \ominus (B_{\pi(1)} \cup \dots \cup B_{\pi(\pi^{-1}(j)-1)})$.

We now define a $T$-branch covering $\mc G_j = \{G_{j,1},\dots,G_{j,m_j+1}\}$ of $\Path_k$ as follows:
\[
  (G_{j,1},\dots,G_{j,m_j})
  &\defeq 
  (B_{\tau_j(1)},\dots,B_{\tau_j(j^\star - 1)},
  H_{j,1},\dots,H_{j,n_j},
  B_{\tau_j(j^\star+1)},\dots,B_{\tau_j(j)}),\\
  G_{j,m_j+1} 
  &\defeq
  B_{j+1} \cup \dots \cup B_\ell.
\]
That is, $G_{j,m_j+1}$ is the label in $T$ of the $j$th right descendant of the root; we will ignore the contribution of this graph when bounding $\vv\Delta(G_{j,1},\dots,G_{j,m_j+1})$.
Note that $\mc G_j$ is indeed a $T$-branch covering of $\Path_k$.

\subsubsection*{Lower bound on $\vec{\Delta}(G_{j,1},\dots,G_{j,m_j})$}

By the ``chain rule'' for $\Delta$ 
(Lemma \ref{la:delta-props}), 
\[
  \vphantom{\Big|}\vv\Delta(G_{j,1},\dots,G_{j,m_j}) 
  =\mbox{}
  \mathrel{\phantom{+}}
  \mbox{}&\vv\Delta(B_{\tau_j(1)},\dots,B_{\tau_j(j^\star - 1)}) 
  \\
  +\ &
  \vv\Delta(H_{j,1},\dots,H_{j,n_j}\mid B_{\tau_j(1)} \cup \dots \cup B_{\tau_j(j^\star - 1)})
  \\
  +\ &
  \vv\Delta(B_{\tau_j(j^\star+1)},\dots,B_{\tau_j(j)}
  \mid
  B_{\tau_j(1)} \cup \dots \cup B_{\tau_j(j^\star - 1)} \cup 
  \smash{\overbrace{\vphantom{\big|}H_{j,1}\cup \dots \cup H_{j,n_j}}^{\ts= B_j = B_{\tau_j(j^\star)}}})
  \vphantom{\Big|}\\
  =\mbox{}
  \mathrel{\phantom{+}}
  \mbox{}&
  \vv\Delta(H_{j,1},\dots,H_{j,n_j}\mid B_{\tau_j(1)} \cup \dots \cup B_{\tau_j(j^\star - 1)})\\
  +\ &
  \sum_{h \in [j] \setminus \{j^\star\}}
  \Delta(B_{\tau_j(h)} \ominus (B_{\tau_j(1)} \cup \dots \cup B_{\tau_j(h-1)})).\vphantom{\Big|}
\]
Since $B_{\tau_j(1)} \cup \dots \cup B_{\tau_j(j^\star - 1)} \subseteq B_{\pi(1)} \cup \dots \cup B_{\pi(\pi^{-1}(j)-1)}$, we have 
\[
  \vv\Delta(H_{j,1},\dots,H_{j,n_j}\mid B_{\tau_j(1)} \cup \dots \cup B_{\tau_j(j^\star - 1)})
  &\ge
  \vv\Delta(H_{j,1},\dots,H_{j,n_j}\mid B_{\pi(1)} \cup \dots \cup B_{\pi(\pi^{-1}(j)-1)})\\
  &=
  \vv\Delta(H^j_1,\dots,H^j_{n_j}).\vphantom{\Big|}
\]
Next, we have
\bigskip
\[
  \sum_{h \in [j] \setminus \{j^\star\}}
  \Delta(B_{\tau_j(h)} \ominus (B_{\tau_j(1)} \cup \dots \cup B_{\tau_j(h-1)}))
  &=
  \sum_{i \in [j-1]}
  \Delta(B_i \ominus (
  \smash{\overbrace{\vphantom{\big|}
    B_{\smash{\tau_j(1)}} \cup \dots \cup B_{\smash{\tau_j(\tau_j^{-1}(i)-1)})}
  }^{\ts
  \!\!\!\!
  \vphantom{|_j}
  \subseteq  B_{\pi(1)} \cup \cdots \cup B_{\smash{\pi(\pi^{-1}(i)-1)}}
  \!\!\!\!
  }}
  ))\\
  &\ge
  \sum_{i \in [j-1]}
  \Delta(\underbrace{\vphantom{\big|}B_i \ominus (
  B_{\pi(1)} \cup \dots \cup B_{\pi(\pi^{-1}(i)-1)}
  )}_{\ts
  = B^i}).
\]

The above inequalities combine with bound (\ref{eq:Delta-H}) on $\vv\Delta(H^j_1,\dots,H^j_{n_j})$ to yield:
\begin{align}\notag
  \vv\Delta(G_{j,1},\dots,G_{j,m_j}) 
  &\ge
  \vv\Delta(H_{j,1},\dots,H_{j,n_j}) + 
  \sum_{i \in [j-1]}
  \Delta(B^i)
  \\
 \label{eq:Gj} 
 &\ge
  \frac{1}{30\Exp} 
  (d-1)\lambda(B^j)^{1/(d-1)} + 
  \sum_{i\in[j]} \Delta(B^i) - d + 1.
\end{align}

\subsubsection*{Completing the proof via the numerical inequality Lemma \ref{la:numerical}}

Combining bounds (\ref{eq:lambda-delta}) and (\ref{eq:Gj}) with the numerical inequality Lemma \ref{la:numerical}, we get
\[
  \max_{j \in [\ell]}
  \vv\Delta(G^j_1,\dots,G^j_{m_j+1})
  &\ge
  \max_{j \in [\ell]}
  \vv\Delta(G^j_1,\dots,G^j_{m_j})
  &&\text{(ignoring $G^j_{m_j+1}$)}\\
  &\ge
  \max_{j \in [\ell]} 
  \bigg(
  (d-1)
  \bigg(\frac{\lambda(B^j)}{(30\Exp)^{d-1}}\bigg)^{1/(d-1)} + 
  \sum_{i\in[j]} \Delta(B^i)
  \bigg)
  - d + 1
  &&\text{(by (\ref{eq:Gj}))}\\ 
  &\ge
  d\bigg(\frac{1}{30^{d-1}\Exp^d} 
  \sum_{j\in[\ell]} \lambda(B^j)  
  \Delta(B^j)\bigg)^{1/d} 
  - d + 1
  &&\text{(by Lemma \ref{la:numerical})}\\
  &\ge
  \frac{1}{30\Exp} d k^{1/d} - d + 1
  &&\text{(by (\ref{eq:lambda-delta}))}.
\]
The proof is completed by picking the optimal $j \in [\ell]$, as inequality (\ref{eq:goal}) is satisfied by the $T$-branch covering $\mc G = \mc G_j$ and enumeration $(G_1,\dots,G_m) = (G_{j,1},\dots,G_{j,m_j+1})$.
\end{proof}

\section{Join tree tradeoff (II)}\label{sec:tradeoff2}

In this section we prove inequalities (II) of Lemmas~\ref{la:pre-pi-sigma}--\ref{la:pi-sigma} and Theorem~\ref{thm:tradeoff}, which we eventually use to prove tradeoffs (II)$^+$ and (II)$^-$ of Theorem \ref{thm:AC0tradeoffs} for monotone and non-monotone $\ACzero$ circuits.

\subsection{Lemma relating the $\semempty$ operation and shift permutations}\label{sec:Psi-shift}

Our proof of Theorem \ref{thm:tradeoff}(I) was based on the following inequality, which relates $\Psi(\sqq{T_1,\dots,T_m})$ with permutations of $[j]$ where $j \le m$. 
(We omit the proof, which is implicit in the argument of \S\ref{sec:tradeoff1}.)

\begin{la}\label{la:meat}
Suppose $T = \sqq{T_1,\dots,T_m}$ where $T_j$ are $G_j$-join trees with $G_j \subset \Path_\Z$.
Then for every $j \in [m]$ and permutation $\tau : [j] \stackrel\cong\to [j]$ and $j^\star \defeq \tau^{-1}(j)$,
\[
\Psi(T)
  \ge
  \Psi(T_j \ominus (G_{\tau(1)} \cup \dots \cup G_{\tau(j^\star-1)}))
  -
  \Delta(T_j \ominus (G_{\tau(1)} \cup \dots \cup G_{\tau(j^\star-1)}))
  +
  \vv\Delta(G_{\tau(1)},\dots,G_{\tau(j)}).
\]
\end{la}

In this section, we prove an analogous result (Lemma \ref{la:psi-bound}) for join trees $T = \sem{T_1,\dots,T_m}$, which will be the backbone of our proof of Theorem \ref{thm:tradeoff}(II) in \S\ref{sec:tradeoff2}. In order to state this lemma, we require one definition.

\begin{df}[Induced shift permutations $\wt\sigma_1,\dots,\wt\sigma_m$]
Consider a shift permutation $\sigma = \sigma_I : [m] \stackrel\cong\to [m]$ where $I = \{i_1,\dots,i_p\}$ with $0 =: i_0 < i_1 < \dots < i_p = m$.  
For each $j \in [m]$, we define $\wt\sigma{}_j = \sigma_{\wt I_j} : [m] \stackrel\cong\to [m]$ where
\[
  \wt I_j \defeq 
  \begin{cases}
    I \cup [i_{h-1}]
    &\text{if } j = i_h \in I,\\
    I \cup [j-1]
    &\text{if } j \notin I.
  \end{cases}
\]
For example, if $\sigma = \sigma_{\{3,5\}} : (1,2,3,4,5) \mapsto (3,1,2,5,4)$, then
\[
  \wt\sigma_1 = \wt\sigma_3 = \sigma_{\{3,5\}} : (1,2,3,4,5) &\mapsto (3,1,2,5,4),\\
  \wt\sigma_2 = \sigma_{\{1,3,5\}} : (1,2,3,4,5) &\mapsto (1,3,2,5,4),\\
  \wt\sigma_4 = \wt\sigma_5 = \sigma_{\{1,2,3,5\}} : (1,2,3,4,5) &\mapsto (1,2,3,5,4).
\]
\end{df}

For future reference (in the proof of Lemma \ref{la:2} in the next subsection), we record a simple property of shift permutations $\wt\sigma_j$.

\begin{la}\label{la:simple-property}
For all $\sigma = \sigma_I : [m] \stackrel\cong\to [m]$ and $j \in [m]$, we have
\[
   \vv\Delta(G_{\wt\sigma_j(1)},\dots,G_{\wt\sigma_j(m)}) 
   &\ge
   \sum_{i \in I}
   \Delta(G_i \ominus (G_1 \cup \dots \cup G_{i-1})).
\]
\end{la}

\begin{proof}
Consider any $\{m\} \subseteq I \subseteq [m]$ and $j \in [m]$. For each $i \in I$, we have $i \in \wt I_j$. Let $i^\leftarrow$ denote the predecessor of $i$ in the set the $\{0\} \cup \wt I_j$, that is, $i^\leftarrow \defeq \max(\{h \in \{0\} \cup \wt I_j : h < i\})$. Note that
\[
  \wt\sigma_j(i^\leftarrow + 1) = i
  \quad\text{ and }\quad
  \{\wt\sigma_j(1),\wt\sigma_j(2),\dots,\wt\sigma_j(i^\leftarrow)\} = \{1,2,\dots,i^\leftarrow\}.
\]
We now obtain the desired inequality as follows:
\[
  \vv\Delta(G_{\wt\sigma_j(1)},\dots,G_{\wt\sigma_j(m)}) 
  &=
  \sum_{l=1}^m
  \Delta(G_{\wt\sigma_j(l)}
  \ominus (G_{\wt\sigma_j(1)} \cup \dots G_{\wt\sigma_j(l-1)}))\\
  &\ge
  \sum_{i \in I}
  \Delta(G_{\wt\sigma_j(i^\leftarrow+1)}
  \ominus (G_{\wt\sigma_j(1)} \cup \dots G_{\wt\sigma_j(i^\leftarrow)}))\\
  &=
  \sum_{i \in I}
  \Delta(G_i \ominus (G_1 \cup \dots G_{i^\leftarrow}))\\
  &\ge
  \sum_{i \in I}
  \Delta(G_i \ominus (G_1 \cup \dots G_{i-1})).\qedhere
\]
\end{proof}

The next lemma relates the $\semempty{}$ operation and shift permutations.  

\begin{la}\label{la:psi-bound}
Suppose $T = \sem{T_1,\dots,T_m}$ where $T_j$ are $G_j$-join trees with $G_j \subset \Path_\Z$. 
Then for every shift permutation $\sigma\,:\,[m] \,\stackrel\cong\to\, [m]$ and $h \in [m]$, we have
\begin{enumerate}[\quad\normalfont(i)]
\item
  $\ds
  \Psi(T)
  \ge
  \Psi(T_{\sigma(1)}) 
  + 
  \vv\Delta(G_{\sigma(2)},\dots,G_{\sigma(m)} \mid G_{\sigma(1)})  
  $,
\item
  $\ds
  \Psi(T)
  \ge
  \Psi(T_{\sigma(h)}) 
  + 
  \vv\Delta(G_{\sigma(h+1)},\dots,G_{\sigma(m)} \mid G_{\sigma(1)} \cup \dots \cup G_{\sigma(h)})
  $,
\item
  $\ds
  \Psi(T)
  \ge
  \Psi(T_j \ominus F_j)
  -
  \Delta(G_j \ominus F_j)
  + 
  \vv\Delta(G_{\wt\sigma_j(1)},\dots,G_{\wt\sigma_j(m)})
  $
  where 
  \[
    j \defeq \sigma(h)
    \quad\text{ and }\quad
    F_j \defeq G_{\sigma(1)} \cup \dots \cup G_{\sigma(h-1)}.
  \]
\end{enumerate}
\end{la}

\begin{proof}
We illustrate bounds (i), (ii), (iii) by considering a specific example $\sigma : [5] \stackrel\cong\to [5]$ with $m = 5$: 
\[
  \sigma \defeq \sigma_{\{3,5\}} : (1,2,3,4,5) \mapsto (3,1,2,5,4).
\] 
Following the proof of each bound (i), (ii), (iii) for this specific $\sigma$, we will comment on the 
generalization to arbitrary shift permutations.

In this case, bound (i) states
\[
  \Psi(T) 
  \ge
  \Psi(T_{\sigma(1)}) + \vv\Delta(G_{\sigma(2)},G_{\sigma(3)},G_{\sigma(4)},G_{\sigma(5)} \mid G_{\sigma(1)})
  =
  \Psi(T_3) + \vv\Delta(G_1,G_2,G_5,G_4 \mid G_3),
\]
To obtain this bound,
we consider the dotted branch in $T$ depicted below (ignoring for now arrows pointing to the sub-trees labeled $h=1,\dots,5$).

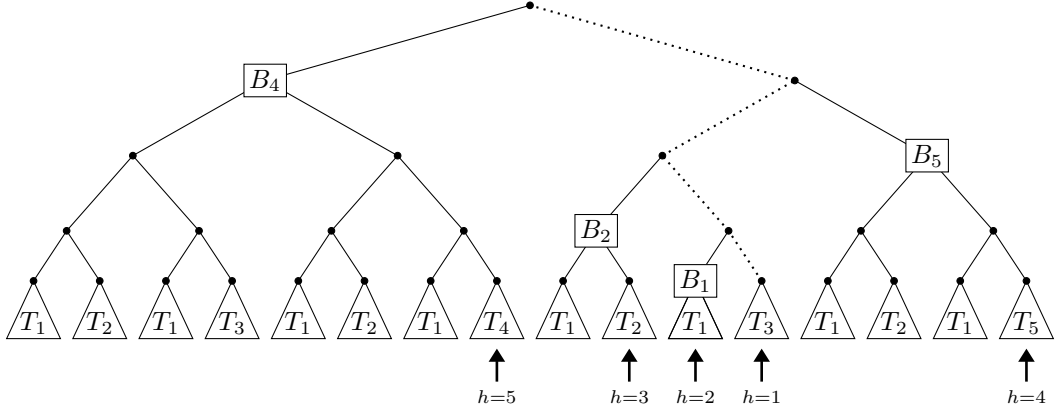
\begin{center}
\begin{tikzpicture}[scale = .4]
  \tikzstyle{v}=[circle, draw, fill, inner sep=0pt, minimum width=2.5pt]
  \tikzstyle{w}=[fill=white, inner sep=2pt]
  
  \tikzstyle{ww}=[draw, fill=white, inner sep=2pt]
  
  \tikzstyle{www}=[
  fill=white, 
  inner sep=2pt]
  
  \tikzstyle{wwww}=[draw, inner sep=2pt]
  
  \tikzstyle{reg}=[thick,dotted]
  \tikzstyle{spec}=[]
  
  \def \shiftor {.35}
  \def \Tdown {1.75+\shiftor}
  \def \tridown {1.7+\shiftor}  
 
  \tikzstyle{tri}=[draw,
  shape border uses incircle,
  isosceles triangle,
  isosceles triangle apex angle=50,
  scale=1.1,
  shape border rotate=90
  ] 
  
  \def \multip {1.1}
  
  \def \a {8*\multip}
  \def \b {4*\multip}
  \def \c {2*\multip}
  \def \d {1*\multip}
  
  \def \Z {2.5}
  
  \def \O {5*\Z}
  \def \A {4*\Z}
  \def \B {3*\Z}
  \def \C {2*\Z}
  \def \D {1.3334*\Z}
 
  \def \e {0}
  
  \def \p {-\a+\b-\c+\d+\e}
  \def \q {\D-\e*\Z}
  
  \draw (-\a,\A) -- (0,\O);
  \draw[thick,dotted] (+\a,\A) -- (0,\O);
  
  \draw (-\a-\b,\B) -- (-\a,\A);
  \draw (-\a+\b,\B) -- (-\a,\A);
  \draw[thick,dotted] (+\a-\b,\B) -- (+\a,\A);
  \draw (+\a+\b,\B) -- (+\a,\A);
  
  \draw (-\a-\b-\c,\C) -- (-\a-\b,\B);
  \draw (-\a-\b+\c,\C) -- (-\a-\b,\B);
  \draw (-\a+\b-\c,\C) -- (-\a+\b,\B);
  \draw (-\a+\b+\c,\C) -- (-\a+\b,\B);
  \draw (+\a-\b-\c,\C) -- (+\a-\b,\B);
  \draw[thick,dotted] (+\a-\b+\c,\C) -- (+\a-\b,\B);
  \draw (+\a+\b-\c,\C) -- (+\a+\b,\B);
  \draw (+\a+\b+\c,\C) -- (+\a+\b,\B);
  
  \draw (-\a-\b-\c-\d,\D) -- (-\a-\b-\c,\C);
  \draw (-\a-\b-\c+\d,\D) -- (-\a-\b-\c,\C);
  \draw (-\a-\b+\c-\d,\D) -- (-\a-\b+\c,\C);
  \draw (-\a-\b+\c+\d,\D) -- (-\a-\b+\c,\C);
  \draw (-\a+\b-\c-\d,\D) -- (-\a+\b-\c,\C);
  \draw (-\a+\b-\c+\d,\D) -- (-\a+\b-\c,\C);
  \draw (-\a+\b+\c-\d,\D) -- (-\a+\b+\c,\C);
  \draw (-\a+\b+\c+\d,\D) -- (-\a+\b+\c,\C);
  \draw (+\a-\b-\c-\d,\D) -- (+\a-\b-\c,\C);
  \draw (+\a-\b-\c+\d,\D) -- (+\a-\b-\c,\C);
  \draw (+\a-\b+\c-\d,\D) -- (+\a-\b+\c,\C);
  \draw[thick,dotted] (+\a-\b+\c+\d,\D) -- (+\a-\b+\c,\C);
  \draw (+\a+\b-\c-\d,\D) -- (+\a+\b-\c,\C);
  \draw (+\a+\b-\c+\d,\D) -- (+\a+\b-\c,\C);
  \draw (+\a+\b+\c-\d,\D) -- (+\a+\b+\c,\C);
  \draw (+\a+\b+\c+\d,\D) -- (+\a+\b+\c,\C);
  
  \draw (0,\O) node [v] {};
  
  \draw (-\a,\A) node [ww] {$B_4$};
  \draw (+\a,\A) node [v] {};
  
  \draw (-\a-\b,\B) node [v] {};
  \draw (-\a+\b,\B) node [v] {};
  \draw (+\a-\b,\B) node [v] {};
  \draw (+\a+\b,\B) node [ww] {$B_5$};
  
  \draw (-\a-\b-\c,\C) node [v] {};
  \draw (-\a-\b+\c,\C) node [v] {};
  \draw (-\a+\b-\c,\C) node [v] {};
  \draw (-\a+\b+\c,\C) node [v] {};
  \draw (+\a-\b-\c,\C) node [ww] {$B_2$};
  \draw (+\a-\b+\c,\C) node [v] {};
  \draw (+\a+\b-\c,\C) node [v] {};
  \draw (+\a+\b+\c,\C) node [v] {};
  
  \draw (-\a-\b-\c-\d,\D) node [v] {};
  \draw (-\a-\b-\c+\d,\D) node [v] {};
  \draw (-\a-\b+\c-\d,\D) node [v] {};
  \draw (-\a-\b+\c+\d,\D) node [v] {};
  \draw (-\a+\b-\c-\d,\D) node [v] {};
  \draw (-\a+\b-\c+\d,\D) node [v] {};
  \draw (-\a+\b+\c-\d,\D) node [v] {};
  \draw (-\a+\b+\c+\d,\D) node [v] {};
  \draw (+\a-\b-\c-\d,\D) node [v] {};
  \draw (+\a-\b-\c+\d,\D) node [v] {};
  \draw (+\a-\b+\c-\d,\D) node [v] {};
  \draw (+\a-\b+\c-\d,\D) node [v] {};
  \draw (+\a-\b+\c+\d,\D) node [v] {};
  \draw (+\a+\b-\c-\d,\D) node [v] {};
  \draw (+\a+\b-\c+\d,\D) node [v] {};
  \draw (+\a+\b+\c-\d,\D) node [v] {};
  \draw (+\a+\b+\c+\d,\D) node [v] {};
  
  \draw (-\a-\b-\c-\d,\D-\tridown) node [tri] {\phantom\,};
  \draw (-\a-\b-\c+\d,\D-\tridown) node [tri] {\phantom\,};
  \draw (-\a-\b+\c-\d,\D-\tridown) node [tri] {\phantom\,};
  \draw (-\a-\b+\c+\d,\D-\tridown) node [tri] {\phantom\,};
  \draw (-\a+\b-\c-\d,\D-\tridown) node [tri] {\phantom\,};
  \draw (-\a+\b-\c+\d,\D-\tridown) node [tri] {\phantom\,};
  \draw (-\a+\b+\c-\d,\D-\tridown) node [tri] {\phantom\,};
  \draw (-\a+\b+\c+\d,\D-\tridown) node [tri] {\phantom\,};
  
\draw[-{Triangle[width=6pt,length=5pt]}, line width=1pt](-\a+\b+\c+\d,\D-\tridown-2) -- (-\a+\b+\c+\d,\D-\tridown-.9);

\draw (-\a+\b+\c+\d,\D-\tridown-2.5) node [w] {$\scriptstyle h=5$};
  
  \draw (+\a-\b-\c-\d,\D-\tridown) node [tri] {\phantom\,};
  \draw (+\a-\b-\c+\d,\D-\tridown) node [tri] {\phantom\,};
  
\draw[-{Triangle[width=6pt,length=5pt]}, line width=1pt](\a-\b-\c+\d,\D-\tridown-2) -- (\a-\b-\c+\d,\D-\tridown-.9);
\draw (\a-\b-\c+\d,\D-\tridown-2.5) node [w] {$\scriptstyle h=3$};

  \draw (+\a-\b+\c-\d,\D-\tridown) node [tri] {\phantom\,};
  
\draw[-{Triangle[width=6pt,length=5pt]}, line width=1pt](\a-\b+\c-\d,\D-\tridown-2) -- (\a-\b+\c-\d,\D-\tridown-.9);
\draw (\a-\b+\c-\d,\D-\tridown-2.5) node [w] {$\scriptstyle h=2$};

  \draw 
  (+\a-\b+\c+\d,\D-\tridown) node [tri] {\phantom\,};
  
\draw[-{Triangle[width=6pt,length=5pt]}, line width=1pt](\a-\b+\c+\d,\D-\tridown-2) -- (\a-\b+\c+\d,\D-\tridown-.9);

\draw (\a-\b+\c+\d,\D-\tridown-2.5) node [w] {$\scriptstyle h=1$};

  \draw (+\a+\b-\c-\d,\D-\tridown) node [tri] {\phantom\,};
  \draw (+\a+\b-\c+\d,\D-\tridown) node [tri] {\phantom\,};
  \draw (+\a+\b+\c-\d,\D-\tridown) node [tri] {\phantom\,};
  \draw (+\a+\b+\c+\d,\D-\tridown) node [tri] {\phantom\,};
  
\draw[-{Triangle[width=6pt,length=5pt]}, line width=1pt](\a+\b+\c+\d,\D-\tridown-2) -- (\a+\b+\c+\d,\D-\tridown-.9);

\draw (\a+\b+\c+\d,\D-\tridown-2.5) node [w] {$\scriptstyle h=4$};
  
  \draw (-\a-\b-\c-\d,\D-\Tdown) node {$T_1$};
  \draw (-\a-\b-\c+\d,\D-\Tdown) node {$T_2$};
  \draw (-\a-\b+\c-\d,\D-\Tdown) node {$T_1$};
  \draw (-\a-\b+\c+\d,\D-\Tdown) node {$T_3$};
  \draw (-\a+\b-\c-\d,\D-\Tdown) node {$T_1$};
  \draw (-\a+\b-\c+\d,\D-\Tdown) node {$T_2$};
  \draw (-\a+\b+\c-\d,\D-\Tdown) node {$T_1$};
  \draw (-\a+\b+\c+\d,\D-\Tdown) node {$T_4$};
  \draw (+\a-\b-\c-\d,\D-\Tdown) node {$T_1$};
  \draw (+\a-\b-\c+\d,\D-\Tdown) node {$T_2$};
  \draw (+\a-\b+\c-\d,\D-\Tdown) node {$T_1$};
  
  \draw (+\a-\b+\c-\d,\D-\tridown) node [tri] {\phantom\,};
  
  \draw (+\a-\b+\c+\d,\D-\Tdown) node {$T_3$};
  \draw (+\a+\b-\c-\d,\D-\Tdown) node {$T_1$};
  \draw (+\a+\b-\c+\d,\D-\Tdown) node {$T_2$};
  \draw (+\a+\b+\c-\d,\D-\Tdown) node {$T_1$};
  \draw (+\a+\b+\c+\d,\D-\Tdown) node {$T_5$};
  
  \draw (+\a-\b+\c-\d,\D) node [ww] {$B_1$};
  
\end{tikzpicture}
\end{center}

This branch 
descends right-left-right-right from the root before continuing as the optimal branch inside the $T_3$-subtree (i.e.,\ the branch that witnesses the value of $\Psi(T_3)$). Observe that the four branch-sibling graphs above this $T_3$-subtree are, in order from the bottom up:
\[
  B_1 \defeq G_1,\qquad
  B_2 \defeq G_1 \cup G_2,\qquad
  B_5 \defeq G_1 \cup G_2 \cup G_5,\qquad
  B_4 \defeq G_1 \cup G_2 \cup G_3 \cup G_4.
\]
Note that in general we have
\begin{equation}\label{eq:fact}
  G_{\sigma(1)} \cup B_{\sigma(2)} \cup \dots \cup B_{\sigma(h)} = 
G_{\sigma(1)} \cup G_{\sigma(2)} \cup \dots \cup G_{\sigma(h)}
  \tu{ for all }h \in [m].
\end{equation}

Now consider the $T$-branch covering of $G_1 \cup \dots \cup G_5$ that consists of an optimal enumeration of the $T_3$-branch covering of $G_3$, followed by graphs $B_1,B_2,B_5,B_4$. 
This shows that
\[
  \Psi(T) \ge \Psi(T_3) + \vv\Delta(B_1,B_2,B_5,B_4 \mid G_3)
  = \Psi(T_3) + \vv\Delta(G_1,G_2,G_5,G_4 \mid G_3),
\]
that is, the required bound (i). 

For general shift permutations $\sigma : [m] \stackrel\cong\to [m]$, we have
\[
  \Psi(T) \ge \Psi(T_{\sigma(1)}) + 
\vv\Delta(B_{\sigma(2)},\dots,B_{\sigma(m)} \mid G_3) = \Psi(T_{\sigma(1)}) +\vv\Delta(G_{\sigma(2)},\dots,G_{\sigma(m)} \mid G_3)
\]
by fact (\ref{eq:fact}) and Lemma \ref{la:delta-props}(e).
\bigskip

To obtain bound (ii), for each $h \in [5]$, consider the $T_{\sigma(h)}$-subtree indicated by arrows in the above diagram.
Observe that the branch-sibling graphs above this $T_{\sigma(h)}$-subtree include $B_{\sigma(h+1)},\dots,B_{\sigma(5)}$.
Now consider a $T$-branch covering of $G_1 \cup \dots \cup G_5$ that begins with an optimal enumeration of an optimal $T_{\sigma(h)}$-branch covering of $G_{\sigma(h)}$ (i.e.\ witnessing the value of $\Psi(T_{\sigma(h)})$), followed by graphs $B_{\sigma(h+1)},\dots,B_{\sigma(5)}$. 
This yields the required bound (ii), respectively for $h = 1,\dots,5$:
\[
  \Psi(T) 
  &\ge
  \Psi(T_3) + \vv\Delta(B_1,B_2,B_5,B_4 \mid G_3)
  =
  \Psi(T_3) + \vv\Delta(G_1,G_2,G_5,G_4 \mid G_3),\\
  \Psi(T) 
  &\ge
  \Psi(T_1) + \,\phantom{G_3,}\vv\Delta(B_2,B_5,B_4 \mid G_1)
  \ge
  \Psi(T_1) + \,\phantom{G_3,}\vv\Delta(G_2,G_5,G_4 \mid G_3 \cup G_1),\\
  \Psi(T) 
  &\ge
  \Psi(T_2) + \,\phantom{G_3,G_1,}\vv\Delta(B_5,B_4 \mid G_2)
  \ge
  \Psi(T_2) + \,\phantom{G_3,G_1,}\vv\Delta(G_5,G_4 \mid G_3 \cup G_1 \cup G_2),\\
  \Psi(T) 
  &\ge
  \Psi(T_5) + \,\phantom{G_3,G_1,G_2,}\vv\Delta(B_4 \mid G_5)
  \ge
  \Psi(T_5) + \,\phantom{G_3,G_1,G_2,}\vv\Delta(G_4 \mid G_3 \cup G_1 \cup G_2 \cup G_5),\\ 
  \Psi(T) 
  &\ge
  \Psi(T_4) + \,\!\!\phantom{G_3,G_1,G_2,G_5,}\hspace{-1pt}\vv\Delta(\: \mid G_4)
  \ge
  \Psi(T_4)
  +
  \,\!\!\phantom{G_3,G_1,G_2,G_5,}\vv\Delta(\:\mid G_3 \cup G_1 \cup G_2 \cup G_5 \cup G_4).
\]

In general, we have 
\[
\Psi(T) &\ge \Psi(T_{\sigma(h)}) +
\vv\Delta(B_{\sigma(h+1)},\dots,B_{\sigma(m)} \mid G_{\sigma(1)} \cup \dots \cup G_{\sigma(h)})\\
&= 
\Psi(T_{\sigma(h)}) + \vv\Delta(G_{\sigma(h+1)},\dots,G_{\sigma(m)} \mid G_{\sigma(1)} \cup \dots \cup G_{\sigma(h)})
\]
again by fact (\ref{eq:fact}) and Lemma \ref{la:delta-props}.\bigskip

For the final bound (iii), for $j = 1,\dots,5$,
consider the $T_j$-subtrees indicated by arrows in the diagram below.
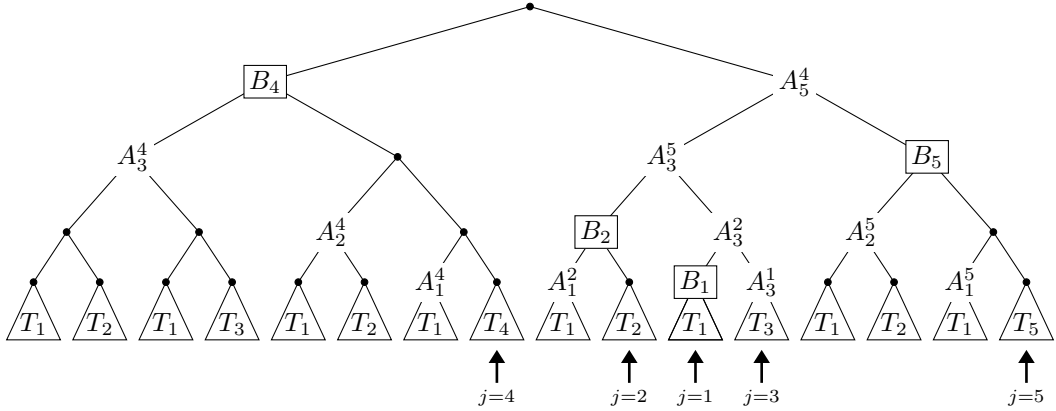
\begin{center}
\begin{tikzpicture}[scale = .4]
  \tikzstyle{v}=[circle, draw, fill, inner sep=0pt, minimum width=2.5pt]
  \tikzstyle{w}=[fill=white, inner sep=2pt]
  
  \tikzstyle{ww}=[draw, fill=white, inner sep=2pt]
  
  \tikzstyle{www}=[
  fill=white, 
  inner sep=2pt]
  
  \tikzstyle{wwww}=[draw, inner sep=2pt]
  
  \tikzstyle{reg}=[thick,dotted]
  \tikzstyle{spec}=[]

  \def \shiftor {.35}
  \def \Tdown {1.75+\shiftor}
  \def \tridown {1.7+\shiftor}  
 
  \tikzstyle{tri}=[draw,
  shape border uses incircle,
  isosceles triangle,
  isosceles triangle apex angle=50,
  scale=1.1,
  shape border rotate=90
  ] 
  
  \def \multip {1.1}
  
  \def \a {8*\multip}
  \def \b {4*\multip}
  \def \c {2*\multip}
  \def \d {1*\multip}
  
  \def \Z {2.5}
  
  \def \O {5*\Z}
  \def \A {4*\Z}
  \def \B {3*\Z}
  \def \C {2*\Z}
  \def \D {1.3334*\Z}
 
  \def \e {0}
  
  \def \p {-\a+\b-\c+\d+\e}
  \def \q {\D-\e*\Z}
  
  \draw (-\a,\A) -- (0,\O);
  \draw (+\a,\A) -- (0,\O);
  
  \draw (-\a-\b,\B) -- (-\a,\A);
  \draw (-\a+\b,\B) -- (-\a,\A);
  \draw (+\a-\b,\B) -- (+\a,\A);
  \draw (+\a+\b,\B) -- (+\a,\A);
  
  \draw (-\a-\b-\c,\C) -- (-\a-\b,\B);
  \draw (-\a-\b+\c,\C) -- (-\a-\b,\B);
  \draw (-\a+\b-\c,\C) -- (-\a+\b,\B);
  \draw (-\a+\b+\c,\C) -- (-\a+\b,\B);
  \draw (+\a-\b-\c,\C) -- (+\a-\b,\B);
  \draw (+\a-\b+\c,\C) -- (+\a-\b,\B);
  \draw (+\a+\b-\c,\C) -- (+\a+\b,\B);
  \draw (+\a+\b+\c,\C) -- (+\a+\b,\B);
  
  \draw (-\a-\b-\c-\d,\D) -- (-\a-\b-\c,\C);
  \draw (-\a-\b-\c+\d,\D) -- (-\a-\b-\c,\C);
  \draw (-\a-\b+\c-\d,\D) -- (-\a-\b+\c,\C);
  \draw (-\a-\b+\c+\d,\D) -- (-\a-\b+\c,\C);
  \draw (-\a+\b-\c-\d,\D) -- (-\a+\b-\c,\C);
  \draw (-\a+\b-\c+\d,\D) -- (-\a+\b-\c,\C);
  \draw (-\a+\b+\c-\d,\D) -- (-\a+\b+\c,\C);
  \draw (-\a+\b+\c+\d,\D) -- (-\a+\b+\c,\C);
  \draw (+\a-\b-\c-\d,\D) -- (+\a-\b-\c,\C);
  \draw (+\a-\b-\c+\d,\D) -- (+\a-\b-\c,\C);
  \draw (+\a-\b+\c-\d,\D) -- (+\a-\b+\c,\C);
  \draw (+\a-\b+\c+\d,\D) -- (+\a-\b+\c,\C);
  \draw (+\a+\b-\c-\d,\D) -- (+\a+\b-\c,\C);
  \draw (+\a+\b-\c+\d,\D) -- (+\a+\b-\c,\C);
  \draw (+\a+\b+\c-\d,\D) -- (+\a+\b+\c,\C);
  \draw (+\a+\b+\c+\d,\D) -- (+\a+\b+\c,\C);

  \draw (0,\O) node [v] {};
  
  \draw (-\a,\A) node [ww] {$B_4$};
  \draw (+\a,\A) node [v] {};
  
  \draw (-\a-\b,\B) node [v] {};
  \draw (-\a+\b,\B) node [v] {};
  \draw (+\a-\b,\B) node [v] {};
  \draw (+\a+\b,\B) node [ww] {$B_5$};
  
  \draw (-\a-\b-\c,\C) node [v] {};
  \draw (-\a-\b+\c,\C) node [v] {};
  \draw (-\a+\b-\c,\C) node [v] {};
  \draw (-\a+\b+\c,\C) node [v] {};
  \draw (+\a-\b-\c,\C) node [ww] {$B_2$};
  \draw (+\a-\b+\c,\C) node [v] {};
  \draw (+\a+\b-\c,\C) node [v] {};
  \draw (+\a+\b+\c,\C) node [v] {};
  
  \draw (-\a-\b-\c-\d,\D) node [v] {};
  \draw (-\a-\b-\c+\d,\D) node [v] {};
  \draw (-\a-\b+\c-\d,\D) node [v] {};
  \draw (-\a-\b+\c+\d,\D) node [v] {};
  \draw (-\a+\b-\c-\d,\D) node [v] {};
  \draw (-\a+\b-\c+\d,\D) node [v] {};
  \draw (-\a+\b+\c-\d,\D) node [v] {};  
  \draw (-\a+\b+\c+\d,\D) node [v] {};
  \draw (+\a-\b-\c-\d,\D) node [v] {};
  \draw (+\a-\b-\c+\d,\D) node [v] {};
  \draw (+\a-\b+\c-\d,\D) node [v] {};
  \draw (+\a-\b+\c-\d,\D) node [v] {};
  \draw (+\a-\b+\c+\d,\D) node [v] {};
  \draw (+\a+\b-\c-\d,\D) node [v] {};
  \draw (+\a+\b-\c+\d,\D) node [v] {};
  \draw (+\a+\b+\c-\d,\D) node [v] {};
  \draw (+\a+\b+\c+\d,\D) node [v] {};
  
  \draw (-\a-\b-\c-\d,\D-\tridown) node [tri] {\phantom\,};
  \draw (-\a-\b-\c+\d,\D-\tridown) node [tri] {\phantom\,};
  \draw (-\a-\b+\c-\d,\D-\tridown) node [tri] {\phantom\,};
  \draw (-\a-\b+\c+\d,\D-\tridown) node [tri] {\phantom\,};
  \draw (-\a+\b-\c-\d,\D-\tridown) node [tri] {\phantom\,};
  \draw (-\a+\b-\c+\d,\D-\tridown) node [tri] {\phantom\,};
  \draw (-\a+\b+\c-\d,\D-\tridown) node [tri] {\phantom\,};
  \draw (-\a+\b+\c+\d,\D-\tridown) node [tri] {\phantom\,};
  
\draw[-{Triangle[width=6pt,length=5pt]}, line width=1pt](-\a+\b+\c+\d,\D-\tridown-2) -- (-\a+\b+\c+\d,\D-\tridown-.9);

\draw (-\a+\b+\c+\d,\D-\tridown-2.5) node [w] {$\scriptstyle j=4$};
  
  \draw (+\a-\b-\c-\d,\D-\tridown) node [tri] {\phantom\,};
  \draw (+\a-\b-\c+\d,\D-\tridown) node [tri] {\phantom\,};
  
\draw[-{Triangle[width=6pt,length=5pt]}, line width=1pt](\a-\b-\c+\d,\D-\tridown-2) -- (\a-\b-\c+\d,\D-\tridown-.9);
\draw (\a-\b-\c+\d,\D-\tridown-2.5) node [w] {$\scriptstyle j=2$};

  \draw (+\a-\b+\c-\d,\D-\tridown) node [tri] {\phantom\,};
  
\draw[-{Triangle[width=6pt,length=5pt]}, line width=1pt](\a-\b+\c-\d,\D-\tridown-2) -- (\a-\b+\c-\d,\D-\tridown-.9);
\draw (\a-\b+\c-\d,\D-\tridown-2.5) node [w] {$\scriptstyle j=1$};

  \draw (+\a-\b+\c+\d,\D-\tridown) node [tri] {\phantom\,};
  
\draw[-{Triangle[width=6pt,length=5pt]}, line width=1pt](\a-\b+\c+\d,\D-\tridown-2) -- (\a-\b+\c+\d,\D-\tridown-.9);

\draw (\a-\b+\c+\d,\D-\tridown-2.5) node [w] {$\scriptstyle j=3$};

  \draw (+\a+\b-\c-\d,\D-\tridown) node [tri] {\phantom\,};
  \draw (+\a+\b-\c+\d,\D-\tridown) node [tri] {\phantom\,};
  \draw (+\a+\b+\c-\d,\D-\tridown) node [tri] {\phantom\,};
  \draw (+\a+\b+\c+\d,\D-\tridown) node [tri] {\phantom\,};
  
\draw[-{Triangle[width=6pt,length=5pt]}, line width=1pt](\a+\b+\c+\d,\D-\tridown-2) -- (\a+\b+\c+\d,\D-\tridown-.9);

\draw (\a+\b+\c+\d,\D-\tridown-2.5) node [w] {$\scriptstyle j=5$};

  \draw (-\a-\b-\c-\d,\D-\Tdown) node {$T_1$};
  \draw (-\a-\b-\c+\d,\D-\Tdown) node {$T_2$};
  \draw (-\a-\b+\c-\d,\D-\Tdown) node {$T_1$};
  \draw (-\a-\b+\c+\d,\D-\Tdown) node {$T_3$};
  \draw (-\a+\b-\c-\d,\D-\Tdown) node {$T_1$};
  \draw (-\a+\b-\c+\d,\D-\Tdown) node {$T_2$};
  \draw (-\a+\b+\c-\d,\D-\Tdown) node {$T_1$};
  \draw (-\a+\b+\c+\d,\D-\Tdown) node {$T_4$};
  \draw (+\a-\b-\c-\d,\D-\Tdown) node {$T_1$};
  \draw (+\a-\b-\c+\d,\D-\Tdown) node {$T_2$};
  \draw (+\a-\b+\c-\d,\D-\Tdown) node {$T_1$};
  
  \draw (+\a-\b+\c-\d,\D-\tridown) node [tri] {\phantom\,};
  
  \draw (+\a-\b+\c+\d,\D-\Tdown) node {$T_3$};
  \draw (+\a+\b-\c-\d,\D-\Tdown) node {$T_1$};
  \draw (+\a+\b-\c+\d,\D-\Tdown) node {$T_2$};
  \draw (+\a+\b+\c-\d,\D-\Tdown) node {$T_1$};
  \draw (+\a+\b+\c+\d,\D-\Tdown) node {$T_5$};
 
  \draw (+\a-\b-\c-\d,\D) node [www] {$A^2_1$};
  \draw (+\a-\b+\c,\C) node [www] {$A^2_3$};
  
  \draw (+\a-\b+\c-\d,\D) node [ww] {$B_1$};
  
  \draw (+\a-\b+\c+\d,\D) node [www] {$A^1_3$};
  
  \draw (-\a+\b+\c-\d,\D) node [www] {$A^4_1$};
  \draw (-\a+\b-\c,\C) node [www] {$A^4_2$};
  \draw (-\a-\b,\B) node [www] {$A^4_3$};
  \draw (\a,\A) node [www] {$A^4_5$};

  \draw (\a+\b+\c-\d,\D) node [www] {$A^5_1$};
  \draw (\a+\b-\c,\C) node [www] {$A^5_2$};
  \draw (\a-\b,\B) node [www] {$A^5_3$};
\end{tikzpicture}
\end{center}

Letting $h \defeq \sigma^{-1}(j)$, the branch-sibling graphs above this $T_j$-subtree are, in order from the bottom up:
\[
  A^j_{\wt\sigma_j(1)},\dots,A^j_{\wt\sigma_j(h-1)},B_{\sigma(h+1)},\dots,B_{\sigma(5)}.
\]
Note that, since $\{\sigma(1),\dots,\sigma(h-1)\} = \{\wt\sigma_j(1),\dots,\wt\sigma_j(h-1)\}$, we have
\[
  F_j
  \defeq
  G_{\sigma(1)} \cup \dots \cup G_{\sigma(h-1)}
  =
  G_{\wt\sigma_j(1)} \cup \dots \cup G_{\wt\sigma_j(h-1)}.
\]

Now consider a $T$-branch covering of $G_1 \cup \dots \cup G_m$ that consists of
\begin{itemize}
\item
graphs $A^j_{\wt\sigma_j(1)},\dots,A^j_{\wt\sigma_j(h-1)}$, followed by
\item
an optimal enumeration of an optimal $T_j \ominus F_j$-branch covering of $G_j \ominus F_j$ (i.e.\ witnessing the value of $\Psi(T_j \ominus F_j)$), followed by
\item
graphs $B_{\sigma(h+1)},\dots,B_{\sigma(m)}$.
\end{itemize}

Cases $j = 2,5$ ($h = 3,4$) of bound (iii) are obtained as follows:
\begin{alignat*}{3}
  \Psi(T) 
  &\ge
  \Psi(T_2 \ominus (A^2_1 \cup A^2_3)) &&+ 
  \vv\Delta(A^2_1,A^2_3) +
  \vv\Delta(B_5,B_4 \mid A^2_1 \cup A^2_3 \cup G_2)
  \\
  &=
  \Psi(T_2 \ominus (G_1 \cup G_3)) &&+ 
  \vv\Delta(G_1,G_3) +
  \vv\Delta(G_5,G_4 \mid G_1 \cup G_3 \cup G_2)
  \\
  &=
  \Psi(T_2 \ominus (G_1 \cup G_3)) &&+ 
  \vv\Delta(G_1,G_3,G_2,G_5,G_4) &&- 
  \Delta(G_2 \ominus (G_1 \cup G_3)) 
  \\
  &=
  \Psi(T_2 \ominus F_2) &&+ \vv\Delta(G_{\wt\sigma_2(1)},G_{\wt\sigma_2(2)},G_{\wt\sigma_2(3)},G_{\wt\sigma_2(4)},G_{\wt\sigma_2(5)})
  &&- \Delta(G_2 \ominus F_2),
  \vphantom{|_{\Big|}}\\
  \Psi(T) 
  &\ge
  \Psi(T_5 \ominus (A^5_1 \cup A^5_2 \cup A^5_3)) &&+ 
  \vv\Delta(A^5_1,A^5_2,A^5_3) +
  \vv\Delta(B_4 \mid A^5_1 \cup A^5_2 \cup A^5_3 \cup G_5)
  \\
  &=
  \Psi(T_5 \ominus (G_1 \cup G_2 \cup G_3)) &&+ 
  \vv\Delta(G_1,G_2,G_3) +
  \vv\Delta(G_4 \mid G_1 \cup G_2 \cup G_3 \cup G_5)
  \\
  &=
  \Psi(T_5 \ominus (G_1 \cup G_2 \cup G_3)) &&+ 
  \vv\Delta(G_1,G_2,G_3,G_5,G_4) &&- 
  \Delta(G_5 \ominus (G_1 \cup G_2 \cup G_3))
  \\
  &=
  \Psi(T_5 \ominus F_5) &&+ \vv\Delta(G_{\wt\sigma_5(1)},G_{\wt\sigma_5(2)},G_{\wt\sigma_5(3)},G_{\wt\sigma_5(4)},G_{\wt\sigma_5(5)})
   &&- \Delta(G_5 \ominus F_5).
\end{alignat*}
These calculations (and those for $j=1,3,4$) are easily verified.\medskip 

For a general shift permutation $\sigma : [m] \stackrel\cong\to [m]$ and $j \in [m]$ and $h \defeq \sigma^{-1}(j)$, we have
\[
  \Psi(T)
  &\ge
  \begin{aligned}[t]
  &\vv\Delta(A^j_{\wt\sigma_j(1)},\dots,A^j_{\wt\sigma_j(h-1)})\\
  &+
  \Psi(T_j \ominus (A^j_{\wt\sigma_j(1)}\cup \dots \cup A^j_{\wt\sigma_j(h-1)}))  \\
  &+
  \vv\Delta(B_{\sigma(h+1)},\dots,B_{\sigma(m)} \mid A^j_{\wt\sigma_j(1)}\cup \dots \cup A^j_{\wt\sigma_j(h-1)} \cup G_j)
  \end{aligned}\\ 
  &=
  \begin{aligned}[t]
  &\vv\Delta(G_{\wt\sigma_j(1)},\dots,G_{\wt\sigma_j(h-1)})
  &&\ \ \text{(since $A^j_{\wt\sigma_j(1)} \cup \dots \cup A^j_{\wt\sigma_j(i)} = G_{\wt\sigma_j(1)} \cup \dots \cup G_{\wt\sigma_j(i)}$ for all $i \in [h-1]$)}\\
  &+
  \Psi(T_j \ominus F_j)
  &&\ \ \text{(since $A^j_{\wt\sigma_j(1)} \cup \dots \cup A^j_{\wt\sigma_j(h-1)} = F_j$)}\\
  &+
  \vv\Delta(B_{\sigma(h+1)},\dots,B_{\sigma(m)} \mid 
  F_j \cup G_j
  )
  \hspace{-1in}
  \end{aligned}\\ 
  &=
  \Psi(T_j \ominus F_j) - \Delta(G_j \ominus F_j) + 
  \underbrace{\vv\Delta(G_{\wt\sigma_j(1)},\dots,G_{\wt\sigma_j(h-1)})
  + \Delta(G_j \ominus F_j)
  + \vv\Delta(B_{\sigma(h+1)},\dots,B_{\sigma(m)} \mid 
  F_j \cup G_j)}_{
  \ds
  \begin{aligned}
  &= \vv\Delta(G_{\wt\sigma_j(1)},\dots,G_{\wt\sigma_j(h-1)},G_j,B_{\sigma(h+1)},\dots,B_{\sigma(m)})\\
  &= \vv\Delta(G_{\wt\sigma_j(1)},\dots,G_{\wt\sigma_j(m)}).
  \end{aligned}
  }
\]
The last equality above is justified by the observation that, for all $i \in \{h,h+1,\dots,m\}$, 
\[
  G_{\wt\sigma_j(1)} \cup \dots \cup G_{\wt\sigma_j(h-1)} \cup G_j \cup B_{\sigma(h+1)} \cup \dots \cup B_{\sigma(i)}
  &=
  G_{\sigma(1)} \cup \dots \cup  
  G_{\sigma(h)} \cup B_{\sigma(h+1)} \cup \dots \cup B_{\sigma(i)}\\
  &=
  G_{\sigma(1)} \cup \dots \cup G_{\sigma(i)}\\
  &=
  G_{\wt\sigma_j(1)} \cup \dots \cup G_{\wt\sigma_j(i)}
\]
by (\ref{eq:fact}) and the fact that
$
  \{\wt\sigma_j(1),\dots,\wt\sigma_j(i)\}
  =
  \{\sigma(1),\dots,\sigma(i)\}
$.
\end{proof}

\begin{cor}\label{cor:psi-bound}
Suppose $T = \sem{T_1,\dots,T_m}$ where $T_j$ are $G_j$-join trees with $G_j \subset \Path_\Z$. 
Then for every $j \in [m]$, we have
\[
  \Psi(T) 
  &\ge 
  \Psi(T_j) + \vv\Delta(G_1,\dots,G_m \mid G_j).
\]
\end{cor}

\begin{proof}
Applying Lemma \ref{la:psi-bound}(i) in the case $\sigma = \sigma_{[m]\setminus[j-1]} : (1,\dots,m) \mapsto (j,1,\dots,j-1,j+1,\dots,m)$, we have 
\[
  \Psi(T)
  &\ge
  \Psi(T_{\sigma(1)}) 
  + 
  \vv\Delta(G_{\sigma(2)},\dots,G_{\sigma(m)} \mid G_{\sigma(1)})
  =  
  \Psi(T_j) 
  + 
  \vv\Delta(G_1,\dots,G_{j-1},G_{j+1},\dots,G_m \mid G_j).
  \qedhere
\]
\end{proof}

\subsection{Proof of Pre-Main Lemma \ref{la:pre-pi-sigma}(II)}\label{sec:pre-pi-sigma2}

We now prove bound (II) of our Pre-Main Lemma \ref{la:pre-pi-sigma}. The proof is significantly simpler than bound (I).

\newtheorem*{la1}{Pre-Main Lemma \ref{la:pre-pi-sigma}(II)}

\begin{la1}[restated]
Suppose $G_1 \cup \dots \cup G_m = \Path_k$ and $\vv\Delta(G_1,\dots,G_m)$ $= 1$. 
Then there exists a shift permutation $\sigma : [m] \stackrel\cong\to [m]$ such that 
$
  \vv{\lambda}(G_{\sigma(1)},\dots,G_{\sigma(m)}) 
  \ge 
  k/4.
$
\end{la1}

\begin{proof}
Without loss of generality, assume that $G_1 = \Path_{s_1,t_1}$ where $0 \le s_1 < t_1 \le k$ and $s \le k/2$.
For each $j \in \{2,\dots,m\}$, let $\Path_{s_j,t_j}$ be the rightmost connected component of $G_j$ (i.e.,\ with maximal $t_j$). 

Define the set $J \subseteq [m]$ by
\[
  J \defeq \{j \in [m] : t_j > \max\{0,t_1,\dots,t_{j-1}\}\}.
\]
Note that $1 \in J$ and
\[
  \Path_{s_1,k} \subseteq
  \bigcup_{j \in J} G_j
  \quad\text{ and }\quad
  G_1 \cup \dots \cup G_j \subseteq \Path_{0,t_j}
  \text{ for all } j \in J.
\]

Now let $j_1,\dots,j_r \in J$, $1 \le j_1 < \dots < j_r \le m$,
be any \underline{minimal} increasing sequence of indices in $J$ such that
\[
  \Path_{s_1,k} \subseteq G_{j_1} \cup \dots \cup G_{j_r}.
\]
Minimality of the sequence $j_1,\dots,j_r$ implies that we have interleaving endpoints
\[
  s_{j_{1}} 
  < s_{j_{2}}
  \le t_{j_{1}}
  < s_{j_{3}}
  \le t_{j_{2}}
  < s_{j_{4}}
  \le t_{j_{3}}
  < s_{j_{5}}
  \le t_{j_{4}}
  < s_{j_{6}}
  \le
  \dots
  \le t_{j_{{r-2}}}
  < s_{j_{r}}
  \le t_{j_{{r-1}}}
  < t_{j_{r}}.
\]
In particular, 
\[
  s_{j_{1}} 
  < t_{j_{1}} 
  < s_{j_{3}} 
  < t_{j_{3}}  
  < s_{j_{5}} 
  < t_{j_{5}} 
  < \cdots
  \quad\text{ and }\quad
  s_{j_{2}} 
  < t_{j_{2}} 
  < s_{j_{4}} 
  < t_{j_{4}}  
  < s_{j_{6}} 
  < t_{j_{6}} 
  < \cdots.
\]
It follows that for all $h \in [r-2]$,
\[
  \Path_{s_{j_{h+2}},t_{j_{h+2}}}
  \subseteq
  G_{j_{h+2}} \ominus (G_1 \cup \dots \cup G_{j_h}).
\]
Also note that $s_{j_1} \le k/2$ and $t_{j_r} = k$, so that
\[
  \frac{k}{2} 
  \ \ \le\ \ 
  \|\Path_{s_1,k}\| 
  \ \ \le\ \  
  \|G_{j_1} \cup \dots \cup G_{j_r}\|
  \ \ \le\ \  
  \sum_{h \in [r]} (t_{j_h}-s_{j_h}).
\]

We next define a sequence $(i_1,\dots,i_p)$ to be either $(j_1,j_3,j_5,\dots)$ or $(j_2,j_4,j_6,\dots)$ according to two cases. 
\begin{itemize}
  \item
    If 
    $\ds\sum_{\tu{odd }h \in [r]} (t_{j_h}-s_{j_h}) \ge \frac{k}{4}$, then let $\ds p \defeq \left\lceil \frac{r}{2} \right\rceil$ and $(i_1,\dots,i_p) \defeq (j_1,j_3,j_5,\dots,j_{2p-1})$.
  \item
    Otherwise, we have $\ds\sum_{\tu{even }h \in [r]} (t_{j_h}-s_{j_h}) \ge \frac{k}{4}$
    and let $\ds p \defeq \left\lfloor \frac{r}{2} \right\rfloor$ and $(i_1,\dots,i_p) \defeq (j_2,j_4,j_6,\dots,j_{2p})$.
\end{itemize}

Let $i_0 \defeq 0$ and let 
\[
  I \defeq 
  [m] \setminus \bigcup_{h\in[p]} \{l \in [m] : i_{h-1} < l < i_h\}
  \quad
  \Big({=}\ 
  \{i_1,\dots,i_p\} \cup \{l \in [m] : i_p < l \le m\}\Big).  
\]
Recall that $\sigma \defeq \sigma_I : [m] \stackrel\cong\to [m]$ is 
a shift permutation satisfying
\[
  \sigma(i_{h-1}+1) = i_h 
  \quad\text{ and }\quad
  \{\sigma(1),\sigma(2),\dots,\sigma(i_{h-1})\} = \{1,2,\dots,i_{h-1}\}
  \quad \text{ for all } h \in [p].
\]
We conclude the proof as follows:
\[
  \vv{\lambda}(G_{\sigma(1)},\dots,G_{\sigma(m)})
  &=
  \sum_{j=1}^m
  \lambda(G_{\sigma(j)} \ominus (G_{\sigma(1)} \cup 
  \dots \cup G_{\sigma(j-1)}))\\
  &\ge
  \sum_{j\in\{i_0+1,i_1+1,\dots,i_{p-1}+1\}}
  \lambda(G_{\sigma(j)} \ominus (G_{\sigma(1)} \cup 
  \dots \cup G_{\sigma(j-1)}))
  \\
  &=
  \sum_{h=1}^p
  \lambda(G_{\sigma(i_{h-1} + 1)} \ominus 
  (G_{\sigma(1)} \cup 
  \dots \cup G_{\sigma(i_{h-1})})
  )\\
  &=
  \sum_{h=1}^p
  \lambda(
    G_{i_h}
    \ominus 
    (G_1 \cup 
    \dots \cup G_{i_{h-1}})
  )
  \\
  &\ge
  \sum_{h=1}^p
  \lambda(
  \Path_{s_{i_h},t_{i_h}}
  )
  \ \ = \ \ 
  \sum_{h=1}^p
  (t_{i_h}-s_{i_h}
  )
  \ \ \ge\ \ 
  \frac{k}{4}.\qedhere
\]
\end{proof}

Under the same hypothesis where $G_1 \cup \dots \cup G_m = \Path_k$ and $\vv\Delta(G_1,\dots,G_m) = 1$, the next lemma modifies this construction of a shift permutation $\sigma : [m] \stackrel\cong\to [m]$ to extract an additional useful property at a small cost to $\vv{\lambda}(G_{\sigma(1)},\dots,G_{\sigma(m)})$.

\begin{la}[Stronger Pre-Main Lemma (II)]\label{la:2}
Suppose $G_1 \cup \dots \cup G_m = \Path_k$ and $\vv\Delta(G_1,\dots,G_m) = 1$. 
Then there exists a shift permutation $\sigma : [m] \stackrel\cong\to [m]$ such that 
\[
  \vv{\lambda}(G_{\sigma(1)},\dots,G_{\sigma(m)}) 
  &\ge 
  \frac{k}{8} - \frac{\max\{\lambda(G_1),\dots,\lambda(G_m)\}}{2},\\
  \vv\Delta(G_{\wt\sigma_j(1)},\dots,G_{\wt\sigma_j(m)}) 
  &\ge 
  \frac{\vv\Delta(G_1,\dots,G_m)}{2} \quad\text{ for all } j \in [m].
\]
\end{la}

\begin{proof}
Construct sequences $1 \le j_1 < \dots < j_r \le m$ and $0 = i_0 < i_1 < \dots < i_p = m$ exactly as in the previous proof of Lemma \ref{la:pi-sigma}(II).
For concreteness, let us assume that $(i_1,\dots,i_p) = (j_2,j_4,j_6,\dots,j_{2p})$; the argument is similar when $(i_1,\dots,i_p) = (j_1,j_3,j_5,\dots,j_{2p-1})$.

Let
\[
  \ell \defeq \max\{\lambda(G_1),\dots,\lambda(G_m)\}
\]
Since $t_{j_{2h}}-s_{j_{2h}} \le \ell$ for all $h \in [p]$, there exists a partition $[p] = Q_1 \sqcup Q_2$ such that, for both $b \in \{1,2\}$, 
\[
  \sum_{h \in Q_b}\ (t_{j_{2h}}-s_{j_{2h}}) 
  \ge
  \frac{k}{8} - \frac{\ell}{2}.
\]

We consider two shift permutations $\sigma_{I_1},\sigma_{I_2} : [m] \stackrel\cong\to [m]$ where sets $\{m\} \in I_1,I_2 \subseteq [m]$ are defined by 
\[
  I_b
  &\defeq
  [m]
  \setminus
  \bigcup_{h \in Q_b} \{j_{2h-2}+1,\dots,j_{2h}-1\}.
\]
Note that $I_1 \cup I_2 = [m]$ and
$
  \sigma_{I_b}(j_{2h-2}+1) = j_{2h}
$
for all $h \in Q_b$.  

For both $b \in \{1,2\}$, we have
\[
  \vv{\lambda}(G_{\sigma_{I_b}(1)},\dots,G_{\sigma_{I_b}(m)})
  &=
  \sum_{j=1}^m
  \lambda(G_{\sigma_{I_b}(j)} \ominus (F \cup G_{\sigma_{I_b}(1)} \cup 
  \dots \cup G_{\sigma_{I_b}(j-1)}))\\
  &\ge
  \sum_{h \in Q_b}
  \lambda(G_{\sigma_{I_b}(j_{2h-2}+1)} \ominus (F \cup G_{\sigma_{I_b}(1)} \cup 
  \dots \cup G_{\sigma_{I_b}(j_{2h-2})}))
  \\
  &=
  \sum_{h \in Q_b}
  \lambda(G_{\sigma_{I_b}(j_{2h})} \ominus (F \cup G_{\sigma_{I_b}(1)} \cup 
  \dots \cup G_{j_{2h-2}}))\\
  &\ge
  \sum_{h \in Q_b}
  (t_{j_{2h}} - s_{j_{2h}})
  \ \ \ge \ \ 
  \frac{k}{8} - \frac{\ell}{2}.
\]
Next, since $I_1 \cup I_2 = [m]$, it follows that for at least one $b \in \{1,2\}$, we have
\[
   \sum_{l \in I_b} \Delta(G_l \ominus (G_1 \cup \dots \cup G_{l-1}))
   \ \ \ge\ \ 
   \frac12 \sum_{l=1}^m \Delta(G_l \ominus (G_1 \cup \dots \cup G_{l-1}))
   \ \ =\ \ 
   \frac{\vv\Delta(G_1,\dots,G_m)}{2}.
\]
Letting $\sigma \defeq \sigma_{I_b}$ for one such $b \in \{1,2\}$, the proof is completed noting that by Lemma \ref{la:simple-property},
\[
   \vv\Delta(G_{\wt\sigma_j(1)},\dots,G_{\wt\sigma_j(m)})
   &\ge
   \sum_{l \in I_b} \Delta(G_l \ominus (G_1 \cup \dots \cup G_{l-1})).\qedhere
\]
\end{proof}

\subsection{Proof of Main Lemma \ref{la:pi-sigma}(II)}\label{sec:pi-sigma2}

In this subsection we prove Main Lemma \ref{la:pi-sigma}(II). In fact, we prove a stronger result (Lemma \ref{la:4}), which is what we actually use in our proof of Theorem \ref{thm:tradeoff}(II). For graphs $G_1 \cup \dots \cup G_m = \Path_k$, this result generalizes Lemma \ref{la:2} in terms of a parameter, $\mr{gap}(G_1,\dots,G_m)$, which measures the density within the real interval $[0,k]$ of the set of midpoints of connected components of graphs $G_j \ominus (G_1 \cup \dots \cup G_m)$.  

\begin{df}[The parameter $\mr{gap}(G_1,\dots,G_m)$]\label{df:gap}
Suppose $G_1 \cup \dots \cup G_m = \Path_k$. Let $c \defeq \vv\Delta(G_1,\dots,G_m)$ and let $0 \le s_1 < t_1 < \dots < s_c < t_c \le k$ be the unique integers such that
\[
  \bigcup_{j=1}^m G_j \ominus (G_1 \cup \dots \cup G_m) 
  =
  \bigcup_{i=1}^c \Path_{s_i,t_i}.
\]
Note that $0 < \frac{s_1+t_1}{2} < \frac{s_2+t_2}{2} < \dots < \frac{s_c+t_c}{2} < k$.

Let $[0,k]$ denote the closed interval of real numbers between $0$ and $k$. We define $\mr{gap}(G_1,\dots,G_m) \in [0,k]$ by
\[
  \mr{gap}(G_1,\dots,G_m)
  \defeq
  \max_{y \in [0,k]}\ \min_{i \in \{1,\dots,c\}}\ 
  \left|
    \frac{s_i+t_i}{k} - y
  \right|.
\]
\end{df}

\begin{la}\label{la:gap}
Let $G_1 \cup \dots \cup G_m = \Path_k$ and let $g \defeq \mr{gap}(G_1,\dots,G_m)$.
\begin{enumerate}[\quad\normalfont(i)]
\item
$\ds \vv\Delta(G_1,\dots,G_m) \ge \frac{k}{2g}$.
\item
$\ds
  \vv\Delta(G_1,\dots,G_m \mid F) 
  \ge
  \frac{k - \lambda(F)\Delta(F)}{2g}
  -
  \Delta(F)$ 
for all $F \subseteq \Path_k$.
\item
$\ds
  \vv\Delta(G_j,G_1,\dots,G_{j-1},G_{j+1},\dots,G_m)
  \ge
  \frac{k - \lambda(G_j)\Delta(G_j)}{4g}$
for all $j \in [m]$.
\end{enumerate}
\end{la}

\begin{proof}
For $0 \le s_1 < t_1 < \dots < s_c < t_c \le k$ as in Definition \ref{df:gap}, let $p_i \defeq \frac{s_i+t_i}{2}$ for $i \in \{1,\dots,c\}$. 

Inequality (i) follows from the observation that $\max\{p_1,p_2-p_1,p_3-p_2,\dots,p_c-p_{c-1},1-p_c\} \ge \frac{k}{2c}$ with equality iff $(p_1,\dots,p_c) = (\frac{k}{2c},\frac{3k}{2c},\frac{5k}{2c},\dots,\frac{(2c-1)k}{2c})$.

For inequality (ii), let $r \defeq \lambda(F)/2$ and $b = \Delta(F)$.  Assume $b \ge 1$, since otherwise (iii) reduces to (i).  Let $0 < q_1 < \dots < q_b < k$ be the midpoints of components of $F$.
Note that
\[
  \vv\Delta(G_1,\dots,G_m \mid F) 
  &\ge
  \max\Big\{0,\:
    \#\Big\{i \in [c] : p_i \le q_1 - r\Big\} - 1
  \Big\}  
  +
  \max\Big\{0,\:
    \#\Big\{i \in [c] : p_i \ge q_b + r\Big\} - 1
  \Big\}\\
  &\quad+
  \sum_{a=1}^{b-1}
  \max\Big\{0,\:
    \#\Big\{i \in [c] : q_a + r \le p_i \le q_b - r\Big\} - 2
  \Big\}.
\]
It follows that
\[
  \vv\Delta(G_1,&\dots,G_m \mid F) \\
  &\ge
  \max\left\{0,\:
    \left\lceil\frac{q_1-r}{2g}\right\rceil - 1
  \right\}  
  +
  \max\left\{0,\:
    \left\lceil\frac{k-q_b-r}{2g}\right\rceil - 1
  \right\}
  +
  \sum_{a=1}^{b-1}
  \max\left\{0,\:
    \left\lceil\frac{q_{a+1}-q_a-2r}{2g}\right\rceil - 2
  \right\}\\
  &\ge
  -2b
  +
  \left\lceil\frac{q_1-r}{2g}\right\rceil
  +
  \left\lceil\frac{k-q_b-r}{2g}\right\rceil
  +
  \sum_{a=1}^{b-1}
  \left\lceil\frac{q_{a+1}-q_a-2r}{2g}\right\rceil\\
  &\ge
  -2b
  +
  \frac{1}{2g}
  \bigg(
  (q_1-r)
  +
  (k-q_b-r)
  +
  \sum_{a=1}^{b-1}
  (q_{a+1}-q_a-2r)
  \bigg)
  \\
  &= 
  -2b
  +
  \frac{k-2rb}{2g}
  \ \ =\ \ 
  \frac{k-\lambda(F)\Delta(F)}{2g} - 2\Delta(F).
\]

For inequality (iii), note that $\vv\Delta(G_j,G_1,\dots,G_{j-1},G_{j+1},\dots,G_m)
  \ge
  \Delta(G_j)$ and
\[
  \vphantom{\Big|}
  \vv\Delta(G_j,G_1,\dots,G_{j-1},G_{j+1},\dots,G_m)
  &=
  \Delta(G_j) + \vv\Delta(G_1,\dots,G_{j-1},G_{j+1},\dots,G_m \mid G_j)\\
  &=
  \Delta(G_j) + \vv\Delta(G_1,\dots,G_m \mid G_j)\\
  &\ge
  \frac{k-\lambda(G_j)\Delta(G_j)}{2g} - \Delta(G_j)
  \quad\text{(by (ii) with $F = G_j$).}
\]
Inequality (iii) follows by the convex combination of these bounds.
\end{proof}

The next lemma generalizes the proof of Pre-Main Lemma \ref{la:pre-pi-sigma}(II).

\begin{la}\label{la:3}
Suppose $G_1 \cup \dots \cup G_m = \Path_k$ and let $g \defeq \mr{gap}(G_1,\dots,G_m)$.
Then there exists a shift permutation $\sigma : [m] \stackrel\cong\to [m]$ such that 
\[
  \vv{\lambda}(G_{\sigma(1)},\dots,G_{\sigma(m)}) 
  &\ge 
  \frac{g - \max\{\lambda(G_1),\dots,\lambda(G_m)\}}{2}.
\]
\end{la}

\begin{proof}
Let $\ell \defeq \max\{\lambda(G_1),\dots,\lambda(G_m)\}$ and $U_j \defeq G_1 \cup \dots \cup G_j$ and $H \defeq \bigcup_{j=1}^m (G_j \ominus U_{j-1})$. 
Assume that $\ell \le g/2$, since the lemma is trivial otherwise.

For $0 \le s_1 < t_1 < \dots < s_c < t_c \le k$ as in Definition \ref{df:gap}, let $p_i \defeq \frac{s_i+t_i}{2}$ for $i \in \{1,\dots,c\}$. Then either $p_1 \ge g$ or $k-p_c \ge g$ or $p_{i+1} - p_i \ge 2g$ for some $i \in \{1,\dots,c-1\}$.

Consider first the case that $p_{i+1} - p_i \ge 2g$ for some $i \in \{1,\dots,c-1\}$. Note that $s_{i+1} - t_i \le 2g - \ell$. Let $l \in [m]$ be the minimum index such that $\Path_{s_{i+1},t_i} \subseteq U_l$. For each $j \in [l-1]$, note that
\[
  \Path_{s_{i+1},t_i} \cap U_j = \Path_{s_{i+1},a_j} \cup \Path_{b_j,t_i}
\]
for some unique $s_{i+1} \le a_j < b_j \le t_i$. Next, observe that $b_{l-1} - a_{l-1} \le \ell$ since $\Path_{a_{l-1},b_{l-1}} \subseteq G_l$ and $\lambda(G_l) \le \ell$. 
Therefore, $(a_{l-1} - s_{i+1}) + (t_i - b_{l-1}) \ge 2(g-\ell)$.
It follows that $a_{l-1} - s_{i+1} \ge g-\ell$ or $t_i-b_{l-1} \ge g-\ell$. 

Without loss of generality, assume that $a_{l-1} - s_{i+1} \ge g-\ell$. We now apply precisely the argument of Lemma \ref{la:pre-pi-sigma}(II), but focused on the interval $\Path_{a_{l-1},s_{i+1}}$ in place of $\Path_{s_1,k}$ (of length $\ge k/2$) in the proof of Lemma \ref{la:pre-pi-sigma}(II). This produces a shift permutation $\sigma$ with $\vv{\lambda}(G_{\sigma(1)},\dots,G_{\sigma(m)}) \ge (g-\ell)/2$ (instead of $k/4$).

The case where $p_1 \ge g$ or $k-p_c \ge g$ follow from applying the argument of Lemma \ref{la:pre-pi-sigma}(II) with respect to intervals $\Path_{0,t_1}$ and $\Path_{s_c,k}$, respectively. Here we get a shift permutation $\sigma$ with an even stronger bound $\vv{\lambda}(G_{\sigma(1)},\dots,G_{\sigma(m)}) \ge g/2$.
\end{proof}

As a corollary of Lemma \ref{la:3}, we get a proof of:

\newtheorem*{cor3}{Main Lemma \ref{la:pi-sigma}(II)}

\begin{cor3}[restated]
Suppose $G_1 \cup \dots \cup G_m = \Path_k$. Then there exists a shift permutation $\sigma : [m] \stackrel\cong\to [m]$ such that
$  \vv{\lambda\Delta}(G_{\sigma(1)},\dots,G_{\sigma(1)}) 
  \ge
  \sqrt{k/8}.
$
\end{cor3}

\begin{proof}
Letting
\[
  \ell \defeq \max\{\lambda(G_1),\dots,\lambda(G_m)\},
  \qquad
  g \defeq \mr{gap}(G_1,\dots,G_m),
\]
we have
\[
  \max_\sigma\ \vv{\lambda\Delta}(G_{\sigma(1)},\dots,G_{\sigma(1)})
  &\ge
  \max\left\{
    \ell,\ 
    \max_\sigma \vv\lambda(G_{\sigma(1)},\dots,G_{\sigma(1)}),\ 
    \vv\Delta(G_1,\dots,G_m)
  \right\}\\
  &\ge
  \max\left\{
    \ell,\ 
    \frac{g}{2} - \ell,\ 
    \frac{k}{2g}
  \right\}
  \ \ \ge\ \ 
  \max\left\{
    \frac{g}{4},\ 
    \frac{k}{2g}
  \right\}
  \ \ \ge\ \ 
  \sqrt{\frac{k}{8}}.\qedhere
\]
\end{proof}

Our proof of Theorem \ref{thm:tradeoff}(II) in 
the next subsection does not actually use Lemma \ref{la:pi-sigma}(II), but rather 
the following stronger version which combines of Lemmas \ref{la:2} and \ref{la:3}.

\begin{la}[Stronger Main Lemma (II)]\label{la:4}
Suppose $G_1 \cup \dots \cup G_m = \Path_k$ and let $g \defeq \mr{gap}(G_1,\dots,G_m)$.
Then there exists a shift permutation $\sigma : [m] \stackrel\cong\to [m]$ such that 
\[
  \vv{\lambda}(G_{\sigma(1)},\dots,G_{\sigma(m)}) 
  &\ge 
  \frac{g - 3\max\{\lambda(G_1),\dots,\lambda(G_m)\}}{4},\\
  \vv\Delta(G_{\wt\sigma_j(1)},\dots,G_{\wt\sigma_j(m)}) 
  &\ge 
  \frac{k}{4g} \quad\text{ for all } j \in [m].
\]
\end{la}

\begin{proof}
Again let $\ell \defeq \max\{\lambda(G_1),\dots,\lambda(G_m)\}$. 
In the case of $\vv\Delta(G_1,\dots,G_m) = 1$, given the shift permutation $\sigma_I$ of Lemma \ref{la:pre-pi-sigma}(II) which satisfies $\vv{\lambda}(G_{\sigma_I(1)},\dots,G_{\sigma_I(m)}) \ge k/4$, recall that the argument of Lemma \ref{la:2} produces two shift permutations, $\sigma_{I_1}$ and $\sigma_{I_2}$, \underline{both} of which satisfy 
\[
  \vv{\lambda}(G_{\sigma_{I_b}(1)},\dots,G_{\sigma_{I_b}(m)}) 
  \ge
  \frac{\vv{\lambda}(G_{\sigma_I(1)},\dots,G_{\sigma_I(m)}) - \ell}{2}
  \ge
  \frac{k}{8} - \frac{\ell}{2},
\]  
and at least \underline{one} of which additionally satisfies 
\[
  \vv\Delta(G_{\wt{\sigma_{I_b}}{}_{\,j}(1)},\dots,G_{\wt{\sigma_{I_b}}{}_{\,j}(m)}) \ge \frac{\vv\Delta(G_1,\dots,G_m)}{2}
  \quad\text{ for all $j \in [m]$.}
\]

Applying the argument of Lemma \ref{la:2} instead to the shift permutation $\sigma_I$ of Lemma \ref{la:3} which satisfies $\vv{\lambda}(G_{\sigma_I(1)},\dots,G_{\sigma_I(m)}) \ge (g-\ell)/2$, we get
\begin{gather*}
  \vv{\lambda}(G_{\sigma_{I_b}(1)},\dots,G_{\sigma_{I_b}(m)}) 
  \ge 
  \frac{\vv{\lambda}(G_{\sigma_I(1)},\dots,G_{\sigma_I(m)}) - \ell}{2}
  \ge 
  \frac{g-3\ell}{4},\\
  \vv\Delta(G_{\wt{\sigma_{I_b}}{}_{\,j}(1)},\dots,G_{\wt{\sigma_{I_b}}{}_{\,j}(m)})
  \ge 
  \frac{\vv\Delta(G_1,\dots,G_m)}{2}
  \ge 
  \frac{k}{4g},
\end{gather*}
where the last inequality is by Lemma \ref{la:gap}(i).
\end{proof}

\subsection{Proof of Theorem \ref{thm:tradeoff}(II)}\label{sec:proof2}

As in \S\ref{sec:proof1}, we consider a more general restatement of Theorem \ref{thm:tradeoff}(II) that is better suitable for induction on $d$.

\newtheorem*{tradeoff2}{Main Theorem \ref{thm:tradeoff}(II)}

\begin{tradeoff2}[slightly more general restatement]
Let $P$ be a finite subgraph of $\Path_\Z$ and let $T$ be a $P$-join tree of $\semempty$-depth $d$.
Then
\[
  \Psi(T)
  \ge 
  \frac{1}{\sqrt{32\Exp}} d\lambda(P)^{1/2d} - d + \Delta(P).
\]
\end{tradeoff2}

\begin{proof}
We argue by induction on $d$. Here we treat $d=0$ as the base case, where $P$ is a single edge (or empty) and the bound $\Psi(T) \ge \Delta(P)$ is trivial.

For the induction step where $d \ge 1$, we reduce to the case $\Delta(P)=1$ exactly as in our 
proof of Theorem \ref{thm:tradeoff}(I).  We therefore assume that $T = \sem{T_1,\dots,T_m}$ is a $\Path_k$-join tree of $\semempty$-depth $m$ and aim to show
\begin{equation}\label{eq:goal2}
  \Psi(T)
  \ge 
  \eps dk^{1/2d} - d + 1
  \quad\text{ where }\quad
  \eps \defeq \frac{1}{\sqrt{32\Exp}}.
\end{equation}

Let $G_1 \cup \dots \cup G_m = \Path_k$ where $G_j$ is the root graph of $T_j$, and let $g \defeq \mr{gap}(G_1,\dots,G_m)$. We prove inequality (\ref{eq:goal2}) by considering three cases, according to whether $\max_{j \in [m]} \lambda(G_j)$ is:
\begin{center}
    greater than $\dfrac{\eps^2 k}{\Delta(G_j)}$,
  \qquad
    between $\dfrac{g}{8}$ and $\dfrac{\eps^2 k}{\Delta(G_j)}$, 
  \qquad
    or less than $\dfrac{g}{8}$.
\end{center}
The most interesting is the third case, which relies on Lemma \ref{la:4} (the stronger version of Main Lemma \ref{la:pi-sigma}(II)).

\paragraph{Case 1:} 
Assume there exists $j \in [m]$ such that 
$\lambda(G_j) \ge \ds\frac{\eps^2 k}{\Delta(G_j)}$. 

By the induction hypothesis applied to $T_j$, we have
\begin{align}\label{eq:ind-hyp2}
  \Psi(T_j)
  &\ge
  \eps(d-1)\lambda(G_j)^{1/(2d-2)} 
  + \Delta(G_j)
  - d + 1\\
  \notag
  &=
  (2d-2)\left(\left(\frac{\eps}{2}\right)^{2d-2} \lambda(G^j)\right)^{1/(2d-2)}
  + \Delta(G_j)
  - d + 1.
\end{align}
Therefore,
\[
  \Psi(T) + d - 1
  {}_{\vphantom{\big|}}
  &\ge
  \Psi(T_j) + d - 1
  &&\text{(since clearly $\Psi(T) \ge \Psi(T_j)$)}\\
  &\ge
  (2d-2)\left(\left(\frac{\eps}{2}\right)^{2d-2} \lambda(G_j)\right)^{1/(2d-2)} 
  + \Delta(G_j)
  &&\text{(by induction hypothesis (\ref{eq:ind-hyp2}))}\\
  &\ge
  (2d-1)\left(\left(\frac{\eps}{2}\right)^{2d-2} \lambda(G_j) \frac{\Delta(G_j)}{2} \right)^{1/(2d-1)} 
  + \frac{\Delta(G_j)}{2}
  &&\text{(by Lemma \ref{la:numerical0})}
  \\
  &\ge
  2d\left(\left(\frac{\eps}{2}\right)^{2d-2} \lambda(G_j)\frac{\Delta(G_j)^2}{4} 
  \right)^{1/2d}   
  &&\text{(by Lemma \ref{la:numerical0} again)}
  \\
  &\ge
  2d\left(\left(\frac{\eps}{2}\right)^{2d-2} \frac{\eps^2}{4}k\right)^{1/2d}   
  &&\text{(since $\lambda(G_j) \ge \ds\frac{\eps^2 k}{\Delta(G_j)}
  \ge \ds\frac{\eps^2 k}{\Delta(G_j)^2}$)}
  \\
  &=
  \eps d k^{1/2d}.
\]

\paragraph{Case 2:} 
Assume there exists $j \in [m]$ such that
$\ds
  \frac{g}{8}
  \le
  \lambda(G_j) 
  \le
  \frac{\eps^2 k}{\Delta(G_j)}.
$ 

We have
\[
  \Psi(T) + d - 1
  &\ge
  \Psi(T_j) + \vv\Delta(G_1,\dots,G_{j-1},G_{j+1},\dots,G_m \mid G_j)
  + d - 1
  &&\text{(by 
  Corollary \ref{cor:psi-bound})}\\
  &\ge
  (2d-2)\left(\left(\frac{\eps}{2}\right)^{2d-2} \lambda(G_j)\right)^{1/(2d-2)} 
  +\vv\Delta(G_j,G_1,\dots,G_{j-1},G_{j+1},\dots,G_m)  
  \hspace{-2in}\\
  &&&\text{(by induction hypothesis (\ref{eq:ind-hyp2}))}\\
  &\ge
  (2d-2)\left(\left(\frac{\eps}{2}\right)^{2d-2} \lambda(G_j)\right)^{1/(2d-2)} 
  +
  \frac{(1-\eps^2)k}{4g}
  &&\text{(by 
  Lemma \ref{la:gap}(iii))}\\
  &\ge
  2d\left(\left(\frac{\eps}{2}\right)^{2d-2} \lambda(G_j)
  \left(\frac{(1-\eps^2)k}{8g}\right)^2
  \right)^{1/2d} 
  &&\text{(by 
  Lemma \ref{la:numerical0} twice)}
  \\
  &\ge
  2d\left(\left(\frac{\eps}{2}\right)^{2d-2} 
  \frac{(1-\eps^2)^2}{64} k
  \right)^{1/2d} 
  &&\text{(since $\lambda(G_j) \ge \frac{g}{8}$ and $k \ge g$)}
  \\
  &=
  \eps d\left( 
  \frac{(1-\eps^2)^2}{16\eps^2} k
  \right)^{1/2d} 
  \\
  &\ge
  \eps d k^{1/2d}
  &&\text{(since $
  \eps = \frac{1}{\sqrt{32\Exp}}$)}.
\]

\paragraph{Case 3:} 
Assume $\ds\lambda(G_j)
\le 
\frac{g}{8}$
for all $j \in [m]$.

By Lemma \ref{la:4} (the stronger version of Main Lemma \ref{la:pi-sigma}(II)), 
there exists a shift permutation $\sigma : [m] \stackrel\cong\to [m]$ such that
\begin{align}
\label{eq:last1}
  \vv{\lambda}(G_{\sigma(1)},\dots,G_{\sigma(m)}) 
  &\ge 
  \frac{g-3\max_{j\in[m]}\lambda(G_j)}{4} 
  >
  \frac{g}{8},\\
\label{eq:last2}
  \vv\Delta(G_{\wt\sigma_j(1)},\dots,G_{\wt\sigma_j(m)}) 
  &\ge 
  \frac{k}{4g} \quad \text{ for all } j \in [m].
\end{align}

For each $j \in [m]$, let 
\[
  &&&&&&&&
  j^\star &\defeq \sigma^{-1}(j)
  ,\vphantom{\Big|}\\
  &&&&&&&&
  F_j &\defeq G_{\sigma(1)} \cup \dots \cup G_{\sigma(j^\star-1)}
  &&({=}\ G_{\wt\sigma_j(1)} \cup \dots \cup G_{\wt\sigma_j(j^\star-1)})
  ,\\
  &&&&&&&&
  G^j &\defeq G_j \ominus F_j,
  \vphantom{\Big|}\\
  &&&&&&&&
  T^j &\defeq T_j \ominus F_j 
  &&({=}\ T_j \tu{ restricted to } G^j).
  &&&&&&&&
\]Note the inequality
\begin{align}\label{eq:no-wt}
\vphantom{\Big|}
  \vv{\Delta}(G_{\wt\sigma_j(1)},\dots,G_{\wt\sigma_j(m)})
  &\ge
  \vv{\Delta}(G_{\wt\sigma_j(j^\star)},\dots,G_{\wt\sigma_j(m)}
  \mid G_{\wt\sigma_j(1)} \cup \dots \cup G_{\wt\sigma_j(j^\star-1)})
  \\
  \notag
  &=
  \vv{\Delta}(G_{\sigma(j^\star)},\dots,G_{\sigma(m)}
  \mid G_{\sigma(1)}\cup \dots \cup G_{\sigma(j-1)})\\
  \notag
  &=
  \sum_{i=j^\star}^m \Delta(G^{\sigma(i)}).
\end{align}
Also, applying the induction hypothesis to each $T^j$, we have
\begin{align}\label{eq:ind-hyp3}
  \Psi(T^j) 
  &\ge
  \eps (d-1)\lambda(G^j)^{1/(2d-2)} + \Delta(G^j) - d + 1.
\end{align}

For each $j \in [m]$, we have
\[
  \vphantom{\Big|}
  \Psi(T) + d - 1
  &\ge
  \Psi(T^j) 
  - \Delta(G^j)
  + \vv\Delta(G_{\wt\sigma_j(1)},\dots,G_{\wt\sigma_j(m)})
  + d - 1
  &&\text{(by Lemma \ref{la:psi-bound}(iii))}
  \\
  &\ge
  (2d-2)\left(\left(\frac{\eps}{2}\right)^{2d-2} \lambda(G^j)\right)^{1/(2d-2)}
  + \vv\Delta(G_{\wt\sigma_j(1)},\dots,G_{\wt\sigma_j(m)})
  &&\text{(by (\ref{eq:ind-hyp3}))}\\
  &\ge
  (2d-2)\left(\left(\frac{\eps}{2}\right)^{2d-2} \lambda(G^j)\right)^{1/(2d-2)} 
  + \frac{k}{8g}
  + \sum_{i=j^\star}^m \frac{\Delta(G^{\sigma(i)})}{2}
  &&\text{(by $\tfrac12$(\ref{eq:last2}) $+$ $\tfrac12$(\ref{eq:no-wt}))}\\
  &\ge
  (2d-1)\left(\left(\frac{\eps}{2}\right)^{2d-2} \frac{k}{8g} \lambda(G^j)\right)^{1/(2d-1)}
  + \sum_{i=j^\star}^m \frac{\Delta(G^{\sigma(i)})}{2}
  &&\text{(by Lemma \ref{la:numerical0})}.
\]
Maximizing over $j \in [m]$, we finally obtain the desired bound \ref{eq:goal2} using Lemma \ref{la:numerical} (the key numerical inequality) as follows:
\[
  \Psi(T) + d - 1
  &\ge
  \max_{j\in[m]}\ 
  (2d-1)\left(\left(\frac{\eps}{2}\right)^{2d-2} \frac{k}{8g} \lambda(G^j)\right)^{1/(2d-1)}
  + \sum_{i=j^\star}^m \frac{\Delta(G^{\sigma(i)})}{2}
  \hspace{-3in}
  \\
  &=
  \max_{h\in[m]}\ 
  (2d-1)\left(
  \left(\frac{\eps}{2}\right)^{2d-2} 
  \frac{k}{8g} \lambda(G^{\sigma(h)})
  \right)^{1/(2d-1)}
  + \sum_{i=h}^m \frac{\Delta(G^{\sigma(h)})}{2}
  \hspace{-3in}
  \\
  &\ge
  2d\left( \frac{1}{\Exp}
  \sum_{h=1}^m 
  \left(\frac{\eps}{2}\right)^{2d-2} \frac{k}{8g} 
  \lambda(G^{\sigma(h)})\frac{\Delta(G^{\sigma(h)})}{2} \right)^{1/2d}
  &&\text{(by Lemma \ref{la:numerical})}
  \\
  &=
  \eps d\left( 
    \frac{k}{4 \Exp \eps^2 g} 
    \vv{\lambda\Delta}(G_{\sigma(1)},\dots,G_{\sigma(m)})
  \right)^{1/2d}
  \\
  &\ge
  \eps d\left( 
    \frac{k}{4 \Exp \eps^2 g} 
    \vv{\lambda}(G_{\sigma(1)},\dots,G_{\sigma(m)})
  \right)^{1/2d}
  \\
  &\ge
  \eps d\left( 
    \frac{k}{32 \Exp \eps^2 }
  \right)^{1/2d}
  &&\text{(by (\ref{eq:last1}))}
  \\
  &=
  \eps d k^{1/2d}.
  &&\qedhere
\]
\end{proof}

\section{Tradeoffs for pathsets}\label{sec:pathset-tradeoff}

In this section, we review the Pathset Framework of papers \cite{kush2023tree,rossman2015correlation,rossman2018formulas,RossmanICM} and prove new tradeoffs for pathset complexity (Theorem \ref{thm:chi2}) that follow from our tradeoffs for join trees (Theorem \ref{thm:tradeoff}).

\subsection{Relations and joins} 

Throughout this section, we fix arbitrary positive integers $k$ and $n$. We additionally fix an arbitrary parameter $\n \le n$. In our application, we will set $\n \defeq n^{(k-1)/k}$. 

\begin{df}[$G$-relations]
For a graph $G \subseteq \Path_k$, we refer to sets $\mc A \subseteq [n]^{V(G)}$ as {\em $G$-relations}. We denote the set of all $G$-relations by $\scr R_G$. (That is, $\scr R_G$ is the set of ``$V(G)$-ary'' relations on $[n]$.)
\end{df}

The join operation $\bowtie$ combines a $G$-relation and an $H$-relation into a $G \cup H$-relation.

\begin{df}[Join] 
For graphs $G,H \subseteq \Path_k$ and relations $\mc A \in \scr R_G$ and $\mc B \in \scr R_H$, the {\em join} of $\mc A$ and $\mc B$ is the relation $\mc A \bowtie \mc B \in \scr R_{G \cup H}$ defined by
\[
  \mc A \bowtie \mc B &\defeq \{\gamma \in [n]^{V(G) \cup V(H)} : \gamma_{V(G)} \in \mc A \text{ and } \gamma_{V(H)} \in \mc B\}.
\]
\end{df}

Note that the join behaves as direct product $\mc A \times \mc B$ when $V(G) \cap V(H) = \emptyset$, and as intersection $\mc A \cap \mc B$ when $V(G) = V(H)$.

\subsection{Density and pathsets}

\begin{df}[Density]
Let $G \subseteq \Path_k$. The {\em density} of a relation $\mc A \in \scr R_G$ is defined by
\[
  &&&&\mu(\mc A) &\defeq \frac{|\mc A|}{n^{|V(G)|}}
  &&= \Pr_{\alpha \in [n]^{V(G)}}[\ \alpha \in \mc A\ ].&&&&
\intertext{For a graph $F \subseteq \Path_k$ and $I = V(F) \cap V(G)$,
the {\em maximum density of $\mc A$ conditioned on $F$}
is defined by}
  &&&&\mu(\mc A\mid F) &\defeq \max_{\beta \in [n]^{V(F)}}
  \frac{|\{
    \alpha \in \mc A : \alpha_I = \beta_I\}|}{n^{|V(G) \setminus V(F)|}}
  &&=
  \max_{\beta \in [n]^{V(F)}}\
  \Pr_{\alpha \in [n]^{V(G)}}[\ \alpha \in \mc A\ |\ 
    \alpha_I = \beta_I
  \ ].
\] 
(Note that $0 \le \mu(\mc A) = \mu(\mc A \mid \emptyset) \le \mu(\mc A \mid F) \le 1$.)
\end{df}

The next lemma and corollary relate the density of a join to the maximum conditional densities of the constituent relations. The proof is another simple exercise in relational algebra.

\begin{la}[Chain rule for density of a join]\label{la:density}
For all graphs $F,G,H \subseteq \Path_k$ and relations $\mc A \in \scr R_G$ and $\mc B \in \scr R_H$, we have
\[
  \mu(\mc A \bowtie \mc B \mid F)
  &\le
  \mu(\mc A \mid F) \cdot \mu(\mc B \mid F \cup G).
\]
\end{la}

\begin{cor}[$m$-ary version of Lemma \ref{la:density}]\label{cor:density}
For all graphs $G_1,\dots,G_m \subseteq \Path_k$ and relations $\mc A_j \in \scr R_{G_j}$ \tu($j\in[m]$\tu), we have
\[
  \mu(\mc A_1 \bowtie \cdots \bowtie \mc A_m)
  &\le
  \,\prod_{j=1}^m\:
  \mu(\mc A_j \mid G_1 \cup \dots \cup G_{j-1}).
\]
Furthermore, since the density of the join does not depend on the ordering of relations $\mc A_1,\dots,\mc A_m$, we have
\[
  \mu(\mc A_1 \bowtie \cdots \bowtie \mc A_m)
  &\le 
  \prod_{j=1}^m\: 
  \mu(\mc A_{\pi(j)} \mid G_{\pi(1)} \cup \dots \cup G_{\pi(j-1)})
\]
for every permutation $\smash{\pi : [m] \stackrel\cong\to [m]}$.
\end{cor}

The lifting technique of papers \cite{rossman2015correlation,rossman2018formulas} focuses on a special class of $G$-relations that are subject to a family of density constraints in terms of $\Delta(\cdot | \cdot)$. 

\begin{df}[$G$-pathsets]
For a graph $G \subseteq \Path_k$, a {\em $G$-pathset} 
(with respect to parameters $n$ and $\n$) 
is a relation $\mc A \in \scr R_G$ such that for all graphs $F \subseteq \Path_k$,
\[
  \mu(\mc A \mid F) \le (1/\n)^{\Delta(G \,\mid\, F)}.
\]
The set of $G$-pathsets is denoted by $\scr P_G$. (Note that $\scr P_G$ is a proper subset of $\scr R_G$ when $G$ is nonempty.)
\end{df}

The next result (Corollary \ref{cor:density} applied to pathsets) bounds the density of a join of pathsets in terms of the operation $\vv\Delta(\cdot)$.

\begin{cor}\label{cor:pathset-density}
For all graphs $G_1,\dots,G_m \subseteq \Path_k$ and pathsets $\mc A_j \in \scr P_{G_j}$ \tu($j\in[m]$\tu), we have
\[
  \mu(\mc A_1 \bowtie \cdots \bowtie \mc A_m)
  &\le
  \,\min_{\pi \,:\, [m] \,\stackrel\cong\to\, [m]}\:
  (1/\n)^{\vv\Delta(G_{\pi(1)},\dots,G_{\pi(m)})}.
\]
\end{cor}

\subsection{Pathset complexity}

For a graph $G \subseteq \Path_k$, we define a family of complexity measures $\mu_T : \scr P_G \to \N$ for each $G$-join tree $T$.

\begin{df}[Pathset complexity]
Let $G \subseteq \Path_k$, let $\mc A \in \scr P_G$, and let $T$ be a $G$-join tree. The {\em $T$-pathset complexity} of $\mc A$, denoted by $\chi_T(\mc A)$, is defined inductively as follows:
\begin{itemize}
  \item
    If $T$ is a single node labeled by $G$ (in which case $\|G\| \le 1$), then
    \[
      \chi_T(\mc A) 
      &\defeq
      \begin{cases}
        0 &\text{if $\|G\| = 0$ or $\mc A$ is empty},\\
        1 &\text{if $\|G\| = 1$ and $\mc A$ is nonempty}.
      \end{cases}
    \]
  \item
    If $T = \un{T_1}{T_2}$ where $T_1,T_2$ are $G_1,G_2$-join trees, then
    \[
      \chi_T(\mc A)
      &\defeq
      \min_{\substack{
        \vphantom{\big|}\text{sequences } \{(\mc A_i,\mc B_i,\mc C_i)\}_i \,:\\
        \vphantom{|}
        \smash{(\mc A_i,\mc B_i,\mc C_i) \,\in\, \scr P_G \times \scr P_{G_1} \times \scr P_{G_2}},\\
        \vphantom{\big|}
        \mc A_i \,\subseteq\, \mc B_i \,\bowtie\, \mc C_i,\ 
        \mc A \,\subseteq\, \bigcup_i \mc A_i
      }}
      \,\sum_i\ 
      \max\{\chi_{T_1}(\mc B_i),\, \chi_{T_2}(\mc C_i)\}.
    \]
    Here $i$ ranges over an arbitrary finite index set.
\end{itemize}
\end{df}

The following lower bound on pathset complexity was first proved in \cite{rossman2018formulas}, then again in \cite{kush2023tree} with an improved big-$\Omega$ constant.

\begin{thm}[\cite{kush2023tree,rossman2018formulas}]\label{thm:chi1}
For every $\Path_k$-join tree $T$ and $\Path_k$-pathset $\mc A$, 
\[
  \chi_T(\mc A) \ge \n^{\Omega(\log k)} \cdot \mu(\mc A).
\]
\end{thm}

\subsection{Lower bounds on pathset complexity}

The main results of this section (Theorem~\ref{thm:chi2} and Corollary~\ref{cor:full-range}) give tradeoffs for pathset complexity $\chi_T(\mc A)$ with respect to the $\sqq{}$- and $\semempty$-depth of the join tree $T$. These lower bounds are based on the following lemma.

\begin{la}\label{la:chi2}
For every $\Path_k$-join tree $T$ and $\Path_k$-pathset $\mc A$, 
\[
  \chi_T(\mc A) \ge \n^{\Psi(T)} \cdot \mu(\mc A).
\]
\end{la}

\begin{proof}
Fix an enumerate $G_1,\dots,G_m$ of a $T$-branch covering of $\Path_k$ such that $\Psi(T) = \vv\Delta(G_1,\dots,G_m)$.
We first observe that pathset complexity $\chi_T(\mc A)$ is invariant under {\em rotations} of $T$, that is, exchanging the left and right children of any subtree.  Noting that the theorem statement has nothing to do with $\sqq{}$-depth (which is obviously not invariant under rotations), we may assume without loss of generality that $\{G_1,\dots,G_m\}$ is the $T$-branch covering associated with the right spine of $T$. That is, assume that $T = \sqq{T_1,\dots,T_m}$ where each $T_j$ is a $B_j$-join tree and $(G_1,\dots,G_m) = (B_{\pi(1)},\dots,B_{\pi(m)})$ for some permutation $\smash{\pi : [m] \stackrel\cong\to [m]}$.

For each $j \in [m]$, let 
\[
  A_j \defeq B_j \cup \dots \cup B_m,\qquad
  S_j \defeq \sqq{T_j,\dots,T_m}.
\]  
Note that $S_j$ is an $A_j$-join tree. In the case $j=1$, we have $S_1=T$ and $A_1=\Path_k$.

By definition of pathset complexity, for all $j \in [m-1]$ and indices $i_1,\dots,i_{j-1}$, there exist nonempty pathsets $\mc A_{i_1,\dots,i_{j-1}} \in \scr P_{A_j}$ and $\mc B_{i_1,\dots,i_j} \in \scr P_{B_j}$ such that, letting $\mc A_{()} \defeq \mc A$, for all $j \in [m-1]$, we have
\[
  \mc A_{i_1,\dots,i_{j-1}} 
  &\subseteq\: 
  \bigcup_{i_j}\ (\mc B_{i_1,\dots,i_j} \bowtie \mc A_{i_1,\dots,i_j}),\\
  \chi_{S_j}(\mc A_{i_1,\dots,i_{j-1}})
  &= 
  \sum_{i_j}\  \max\{\chi_{T_j}(\mc B_{i_1,\dots,i_j}),\, \chi_{S_{j+1}}(\mc A_{i_1,\dots,i_j})\}.
\]
It follows that
\[
  \chi_T(\mc A) =
  \chi_{S_1}(\mc A_{()})
  \ge 
  \sum_{i_1} \chi_{S_2}(\mc A_{i_1})
  \ge 
  \sum_{i_1,i_2} \chi_{S_3}(\mc A_{i_1,i_2})
  \ge
  \cdots
  \ge
  \sum_{i_1,\dots,i_{m-1}} \chi_{S_m}(\mc A_{i_1,\dots,i_{j-1}})
  =
  \sum_{i_1,\dots,i_{m-1}} 1.
\]
Similarly, letting $\mc B_{i_1,\dots,i_{m-1},1} \defeq \mc A_{i_1,\dots,i_{m-1}}$, we have
\[
  \mc A =
  \mc A_{()}
  \subseteq\: 
  \bigcup_{i_1}\ \mc B_{i_1} \bowtie \mc A_{i_1}
  \:&\subseteq\:
  \bigcup_{i_1}\ \mc B_{i_1} \bowtie\ \bigcup_{i_2}\ \mc B_{i_1,i_2} \bowtie \mc A_{i_1,i_2}
  \:\subseteq\: \cdots\\
  \cdots\:&\subseteq\:
  \bigcup_{i_1}\
  \mc B_{i_1} \bowtie\ \bigcup_{i_2}\ \mc B_{i_1,i_2} 
  \bowtie\ \bigcup_{i_3}\ \mc B_{i_1,i_2,i_3} 
  \bowtie \cdots \bowtie\ \bigcup_{\!\!\!\!i_{m-1}\!\!\!\!}\ \mc B_{i_1,\dots,i_{m-1}} \bowtie \mc A_{i_1,\dots,i_{m-1}}\\
  &=\:
  \bigcup_{\substack{i_1,\dots,i_{m-1}\\ i_m=1}}
  \mc B_{i_1} \bowtie \mc B_{i_1,i_2} 
  \bowtie \mc B_{i_1,i_2,i_3} 
  \bowtie \cdots 
  \bowtie \mc B_{i_1,\dots,i_{m-1}} 
  \bowtie \mc B_{i_1,\dots,i_m}.
\]

We may now bound $\mu(\mc A)$ as follows.
\[
  \mu(\mc A) 
  &\le 
  \sum_{\substack{i_1,\dots,i_{m-1}\\ i_m=1}}\,
  \mu(\mc B_{i_1} \bowtie \mc B_{i_1,i_2} 
  \bowtie \mc B_{i_1,i_2,i_3} 
  \bowtie \cdots \bowtie \mc B_{i_1,\dots,i_m})\\
  &\le 
  \sum_{\substack{i_1,\dots,i_{m-1}\\ i_m=1}}\,
  \prod_{j=1}^m\,
  \mu(\mc B_{i_1,i_2,\dots,i_{\pi(j)}} \mid
  B_{\pi(1)} \cup B_{\pi(2)} \cup \dots \cup B_{\pi(j-1)}
  )
  \quad\ \text{(by Corollary \ref{cor:density})}\\
  &\le 
  \sum_{\substack{i_1,\dots,i_{m-1}\\ i_m=1}}\,
  \prod_{j=1}^m\,
  (1/\n)^{\Delta(B_{\pi(j)}\,\mid\,B_{\pi(1)} \cup B_{\pi(2)} \cup \dots \cup B_{\pi(j-1)})}
  \quad\ \text{(since each $\mc B_{i_1,i_2,\dots,i_{\pi(j)}}$ is a $B_{\pi(j)}$-pathset)}\\
  &=
  (1/\n)^{\vv\Delta(B_{\pi(1)},\dots,B_{\pi(m)})}
  \cdot\sum_{\substack{i_1,\dots,i_{m-1}\\ i_m=1}} 1
  \ \ \le\ \ 
  (1/\n)^{\vv\Delta(B_{\pi(1)},\dots,B_{\pi(m)})}
  \cdot \chi_T(\mc A).
\]
Since $(B_{\pi(1)},\dots,B_{\pi(m)}) = (G_1,\dots,G_m)$, we get the desired bound
\[
  \chi_T(\mc A)
  \ \ &\ge\ \ 
  \n^{\vv\Delta(G_1,\dots,G_m)} \cdot \mu(\mc A) 
  \ \ =\ \  
  \n^{\Psi(T)} \cdot \mu(\mc A).\qedhere
\]
\end{proof}

As an immediate corollary of Lemma \ref{la:chi2} and Theorem \ref{thm:tradeoff}, we get the following tradeoff for pathset complexity.

\begin{thm}[Pathset complexity tradeoffs]\label{thm:chi2}
For every $\Path_k$-join tree $T$ and $\Path_k$-pathset $\mc A$,
\begin{enumerate}[\quad\normalfont(I)]
\item
    \parbox{2in}{$\ds\vphantom{\Big|}
    \smash{
       \chi_T(\mc A) \ge \n^{\frac{1}{30\Exp}dk^{1/d}-d
  } \cdot \mu(\mc A)}
    $}
    where $d$ is the $\sqq{}$-depth of $T$,
\item
    \parbox{2in}{$\vphantom{\big|}\ds
    \smash{
       \chi_T(\mc A) \ge \n^{\frac{1}{\sqrt{32\Exp}}dk^{1/2d}-d
  } \cdot \mu(\mc A)}
    $}
    where $d$ is the $\semempty$-depth of $T$.
\end{enumerate}
\end{thm}

We remark that Theorem \ref{thm:chi2} does not imply Theorem \ref{thm:chi1}, since the bounds of Theorem \ref{thm:chi2} are trivial for $d = \Omega(k)$ while Theorem \ref{thm:chi1} applies to join trees of arbitrary $\sqq{}$- and $\semempty$-depth.
The ``maximally overlapping'' $\Path_k$-join tree provides an example which has 
$\sqq{}$- and $\semempty$-depth $k-1$ and $\Psi$-value $1$.
Together these bounds imply:

\begin{cor}\label{cor:full-range}
Suppose that $\n = n^{\Omega(1)}$ (for example, $\n = n^{1-\frac1k}$ in our application in this paper). Then for every $\Path_k$-join tree $T$ and $\Path_k$-pathset $\mc A$,
\begin{enumerate}[\quad\normalfont(I)]
\item
    \parbox{2in}{$\ds\vphantom{\Big|}
    \smash{
       \chi_T(\mc A) \ge n^{\Omega(d(k^{1/d}-1))} \cdot \mu(\mc A)}
    $}
    where $d$ is the $\sqq{}$-depth of $T$,
\item
    \parbox{2in}{$\vphantom{\big|}\ds
    \smash{
       \chi_T(\mc A) \ge n^{\Omega(d(k^{1/2d}-1))} \cdot \mu(\mc A)}
    $}
    where $d$ is the $\semempty$-depth of $T$.
\end{enumerate}
\end{cor}

The lower bounds of Corollary \ref{cor:full-range} are asymptotically tight for every fixed $d$ and $k$, and both bounds converge to $n^{\Omega(\log k)}$ since $d(k^{1/d}-1)$ converges to $\ln k$ as $d \to \infty$.

\section{Tradeoffs for $\normalfont\ACzero$ and $\normalfont\SACzero$ formulas}\label{sec:tradeoffs-for-formulas}

In this section we prove 
Theorem \ref{thm:AC0tradeoffs}, our size-depth tradeoffs for $\SACzero$ and $\ACzero$ formulas in both the monotone and non-monotone settings.
In \S\ref{sec:conversion} we discuss two different ways of converting $\SACzero$ and $\ACzero$ formula to DeMorgan formulas with binary $\wedge$ and $\vee$ gates. 
We then prove tradeoff (I)$^+$ over the course of \S\ref{sec:minterms}--\ref{sec:I+}. 
At a high-level, the proof is a reduction that starts out with a monotone $\SACzero$ formula of $\bigwedge$-depth $d$ and size $s$ computing $\SPMM_{n,k}$. The reduction extracts a $\Path_k$-join tree $T$ of $\sqq{}$-depth $d$ and a $\Path_k$-pathset $\mc A \subseteq [n]^{\{0,\dots,k\}}$ of density $n^{-O(1)}$ such that $\chi_T(\mc A) \le n^{O(1)} \cdot s$.
On the other hand, the pathset complexity lower bound of \S\ref{sec:pathset-tradeoff} (Corollary \ref{cor:full-range}) implies $\chi_T(\mc A) = n^{\smash{\Omega(dk^{1/d})}} \cdot \mu(\mc A)$.
Combining these inequalities, we get the desired size lower bound $s = n^{\Omega(dk^{1/d}) - O(1)}$ ($= n^{\Omega(dk^{1/d})}$ since $dk^{1/d} \ge \log k$).

We next prove tradeoff (II)$^-$ for non-monotone $\ACzero$ formulas in \S\ref{sec:support-tree}--\ref{sec:II-} via a novel reduction that makes a surprising connection between $\ACzero$ formula depth to $\semempty{}$-depth of join trees.
The remaining two tradeoffs (I)$^-$ and (II)$^+$ of Theorem \ref{thm:AC0tradeoffs} (for non-monotone $\SACzero$ formulas and monotone $\ACzero$ formulas) are proved by a ``convex combination'' of arguments in the proofs of (I)$^+$ of (II)$^-$. 
We omit the redundant details.

\begin{proviso}
Throughout this section, we assume that $d \le \log k$ and either $k \le \log\log n$ or $k \le \log^\ast n$ in the monotone and non-monotone settings, respectively.
In all statements involving pathsets and pathset complexity, the density parameter $\n$ is fixed to $n^{(k-1)/k}$.
\end{proviso}

\subsection{Converting $\ACzero$ formulas to DeMorgan formulas}\label{sec:conversion}

\begin{df}[DeMorgan formulas]
A {\em DeMorgan formula} is a rooted binary tree with non-leafs (``gates'') labeled by $\wedge$ or $\vee$, and leaves (``inputs'') labeled by constants or literals.
As with $\ACzero$ formulas, we measure {\em size} by the number of inputs labeled by literals and {\em depth} (resp.\ {\em $\wedge$-depth}) by the maximum number of gates (resp.\ $\wedge$-gates) on a root-to-leaf branch.
Additionally, we define {\em left-depth} (resp.\ {\em $\wedge$-left-depth}) of a DeMorgan formula as the maximum number of left descents (resp.\ left descents at $\wedge$-gates) on a root-to-leaf branch.
\end{df}

Our tradeoffs for $\SACzero$ and $\ACzero$ formulas involve
two different ways of converting an $\cc{(S)}\ACzero$ formula $\ff F$ to an equivalent DeMorgan formula $\f$.

\begin{df}[DeMorgan conversions of $\ACzero$ formulas]
\ 
\begin{itemize}
\item
The {\bf\em right-deep DeMorgan conversion} of an $\ACzero$ formula $\ff F$ is the DeMorgan formula $\f$ obtained by replacing each $m$-ary each fan-in $m$ $\bigvee$\:/\:$\bigwedge$ gate with a right-deep tree of $m-1$ binary $\vee$\:/\:$\wedge$ gates.

For example, if $\ff F = \bigwedge_{i=1}^8 \ff F_i$, then its right-deep DeMorgan conversion is the formula
\[
  \ff f = \f_1 \wedge (\f_2 \wedge (\f_3 \wedge  
      (\f_4 \wedge (\f_5 \wedge (\f_6 \wedge  
      (\f_7 \wedge \f_8))))))
\]
where each $\f_i$ is the right-deep DeMorgan conversion of $\ff F_i$.

\item
The {\bf\em balanced DeMorgan conversion} of an $\ACzero$ formula $\ff F$ is the DeMorgan formula $\f$ obtained by replacing each $m$-ary each fan-in $m$ $\bigvee$\:/\:$\bigwedge$ gate with a balanced tree of $m-1$ binary $\vee$\:/\:$\wedge$ gates.

For example, if $\ff F = \bigwedge_{i=1}^8 \ff F_i$, then its balanced DeMorgan conversion is the formula
\[
  \ff f =  ((\f_1 \vee \f_2) \vee (\f_3 \vee \f_4)) \vee 
      ((\f_5 \vee \f_6) \vee (\f_7 \vee \f_8))
\]
where each $\f_i$ is the balanced DeMorgan conversion of $\ff F_i$.
\end{itemize}
\end{df}

Note that if $\ff F$ has depth $d$ and fan-in $m$, then its balanced DeMorgan conversion has depth at most $d\lceil\log m\rceil$.  If $\ff F$ has $\bigwedge$-fan-in $d$ and $\bigwedge$-fan-in $m$, then its right-deep DeMorgan conversion has $\wedge$-left-depth $d$ and $\wedge$-depth at most $dm$ (though possibly much greater depth).  Both conversions preserve size.
\bigskip

Our first tradeoff for $\SACzero$ formulas (\S\ref{sec:minterms}--\ref{sec:I+}) uses the right-deep DeMorgan conversion to establish a connection between $\bigwedge$-fan-in, $\wedge$-left-depth, and $\sqq{}$-depth of join trees. 
Our second tradeoff for $\ACzero$ formulas (\S\ref{sec:support-tree}--\ref{sec:II-}) relies on a randomized version of the balanced DeMorgan conversion to establish a connection between $\ACzero$ depth and $\semempty{}$-depth of join trees. 
We state the definition here, but postpone analysis to \S\ref{sec:support-tree}.

\begin{df}[\bf\em Randomized balanced DeMorgan conversion]\label{df:scrD}
Let $t \in \N$ be an arbitrary parameter.
For every $\ACzero$ formula $\ff F$, we define a random DeMorgan formula $\f$ with distribution $\scr D_t(\ff F)$ as follows:  
\begin{itemize}
\item
If $\ff F$ has depth $0$ (i.e.,\ is a constant or literal), then  $\f = \ff F$ (with probability $1$).
\item
Suppose that $\ff F$ is $\bigwedge_{i=1}^m \ff F_i$ (resp.\ $\bigvee_{i=1}^m \ff F_i$).  Sample independent random DeMorgan formulas $\f_i \sim \scr D_t(\ff F_i)$ and independent uniform random indices $\mb i_1,\dots,\mb i_{tm} \in [m]$.
Now let $\f$ be the DeMorgan formula consisting of a balanced tree of $tm-1$ binary $\wedge$-gates (resp.\ $\vee$-gates) with subformulas $\f_{\mb i_1},\dots,\f_{\mb i_{tm}}$ feeding in.
\end{itemize}
\end{df}

Note that if $\ff F$ has depth $d$ and size $s$, then every $\f$ in the support of $\scr D_t(\ff F)$ has depth at most $d\lceil\log(ts)\rceil$ and size at most $t^d s$. In particular, if $t = (\log s)^{O(1)}$ and $d = o(\frac{\log s}{\log\log s})$, then $\f$ has size $s^{1+o(1)}$.
Also note that $\f$ does not necessarily compute the same function as $\ff F$. However, for $t \ge (\log s)^2$, we at least ensure that $\f(x) = \ff F(x)$ with high probability for any given input $x$.\medskip

We conclude this subsection with a few remarks on notation.

\begin{notn}
Throughout this section, we speak of both $\ACzero$ and DeMorgan formulas on both $kn^2$ variables and $k$ variables.  To keep these objects straight, we consistently write
\begin{itemize}
  \item
    $\ff F$ for an 
    $\ACzero$ formula on $kn^2$ variables,
  \item
    $\f$ for a DeMorgan formula on $kn^2$ variables,
  \item
    $\ff G$ for an 
    $\ACzero$ formula on $k$ variables,
  \item
    $\g$ for a DeMorgan formula on $k$ variables.
\end{itemize}
We use typewriter font $\ff F,\f$ and $\ff G,\g$ to distinguish these formulas from the Boolean functions of $kn^2$ and $k$ variables that they compute, denoted by italic $f$ and $g$.

In the context of tradeoffs (I)$^\pm$, we refer to $\ff F$ as an ``$\SACzero$ formula'' whenever we assume that it has $\bigwedge$-fan-in at most $n^{1/k}$.
\end{notn}

\begin{notn}
Note the distinction between
syntactic {\em equality} of formulas (notated with $=$) and
semantic {\em equivalence} of formulas (notated with $\equiv$).
For example, $\f_1 = \f_2$ expresses that $\f_1$ and $\f_2$ are isomorphic as labeled binary trees, whereas  
$\g_1 \equiv \g_2$ 
expresses that 
$\g_1$ and $\g_2$ 
compute the same Boolean function $\{0,1\}^k \to \{0,1\}$.
We sometimes write $\equiv$ between an $\ACzero$ formula and an equivalent DeMorgan formula (e.g.,\ $\ff F \equiv \f$ or $\ff G \equiv \g$).
\end{notn}

\subsection{Monotone functions on $kn^2$ variables}\label{sec:minterms}

In the context of problems $\BMM_{n,k}$, $\SPMM_{n,k}$, $\PMM_{n,k}$, we identify elements of $\{0,1\}^{kn^2}$ with $k$-tuples $\vec M = (M^{(1)},\dots,M^{(k)})$ of Boolean matrices $M^{(i)} \in \{0,1\}^{n \times n}$.
However, for purposes of our lower bound method, it is convenient to instead identify $\{0,1\}^{kn^2}$ with the set of subgraphs of the $n$-blow-up of $\Path_k$, defined below.

\begin{df}[The $n$-blow-up of a graph]
The {\em $n$-blow-up} of a graph $G$, denoted $G^{\uparrow n}$, is the graph with vertex and edge sets
\[
  V(G^{\uparrow n}) &\defeq \{v^{(a)} : v \in V(G),\ a \in [n]\} \quad \text{(a way of writing $ V(G) \times [n]$)},\\
  E(G^{\uparrow n}) &\defeq \{\{v^{(a)},w^{(b)}\} : \{v,w\} \in E(G),\ a,b \in [n]\}.
\]
That is, for each vertex $v$ of $G$, the blow-up contains $n$ distinct vertices named $v^{(1)},\dots,v^{(n)}$; and for each edge $\{v,w\}$ of $G$, the blow-up contains the complete bipartite graph between sets $\{v^{(1)},\dots,v^{(n)}\}$ and $\{w^{(1)},\dots,w^{(n)}\}$.

For $\alpha \in [n]^{V(G)}$, 
we write $G^{(\alpha)}$ for the isomorphic copy of $G$ with
\[
  V(G^{(\alpha)}) &\defeq \{v^{(\alpha_v)} : v \in V(G)\},\\
  E(G^{(\alpha)}) &\defeq \{\{v^{(\alpha_v)},w^{(\alpha_w)}\} : \{v,w\} \in E(G)\}.
\]
Graphs $G^{(\alpha)}$ 
are called {\em sections} of the blow-up $G^{\uparrow n}$.

To simplify notation: given a section $G^{(\alpha)}$ of $G^{\uparrow n}$ and a subgraph $H \subseteq G$, we write $H^{(\alpha)}$ instead of $H^{(\alpha_{V(H)})}$ for the corresponding section of $H^{\uparrow n}$.
\end{df}

\begin{df}[Identifying $\{0,1\}^{kn^2}$ with the set of subgraphs of $\Path_k^{\uparrow n}$]
\ 
\begin{itemize}
\item
The blow-up $\Path_k^{\uparrow n}$ has $kn^2$ edges and $2^{kn^2}$ subgraphs. (Recall that {\em graphs} have no isolated vertices according to our definition.)
\item
For $G \subseteq \Path_k$ and $\alpha \in [n]^{V(G)}$, we refer to graphs $G^{(\alpha)}$ as {\em $G$-sections} of $\Path_k^{\uparrow n}$.
\item
We identify elements of $\{0,1\}^{kn^2}$ with subgraphs of $\Path_k^{\uparrow n}$, as well as $k$-tuples of $n$-by-$n$ Boolean matrices.
\end{itemize}
\end{df}

\begin{df}[$G$-minterms]
For a monotone function $f : \{0,1\}^{kn^2} \to \{0,1\}$, we write $\mc M_G(f)$ for the $G$-relation
\[
  \mc M_G(f)
  &\defeq
  \{
    \alpha \in [n]^{V(G)} :
    G^{(\alpha)} \text{ is a minterm of } f
  \}.
\]
That is, $\mc M_G(f)$ is the set of $\alpha \in [n]^{V(G)}$ such that $f(G^{(\alpha)}) = 1$ and $f(H^{(\alpha)}) = 0$ for all $H \subsetneqq G$.
\end{df}

For example, we have $\mc M_{\Path_k}(\BMM_{n,k}) = \{\alpha \in [n]^{\{0,\dots,k\}} : \alpha_0=\alpha_k=1\}$.
More generally:

\begin{obs}\label{obs:Mpath}
A monotone function $f : \{0,1\}^{kn^2} \to \{0,1\}$ computes $\SPMM_{n,k}$ (i.e., agrees with $\BMM_{n,k}$ whenever the input is a $k$-tuple of sub-permutation matrices) if, and only if, $\mc M_{\Path_k}(f) = \{\alpha \in [n]^{\{0,\dots,k\}} : \alpha_0=\alpha_k=1\}$.  In this case, we have $\mu(\mc M_{\Path_k}(f)) = n^{-2}$.
\end{obs}

The next lemma relates the $G$-minterm relations of $f_1 \vee f_2$ and $f_1 \wedge f_2$ to those of $f_1$ and $f_2$.

\begin{la}\label{la:join0}
For every graph $G \subseteq \Path_k$ and monotone functions $f_1,f_2 : \{0,1\}^{kn^2} \to \{0,1\}$, we have
\begin{itemize}
\item
$\ds\vphantom{\Big|}
  \mc M_G(f_1 \vee f_2)
  \ \subseteq\ 
  \mc M_G(f_1) \cup \mc M_G(f_2)$,
\item
\mbox{$\ds\mc M_G(f_1 \wedge f_2)
  \ \subseteq\  
  \!
  \bigcup_{\substack{
  \!\!G_1,G_2 \,\subseteq\, G \,:\!\!\\ \vphantom{\ts|}
  \!\!\!\!  
  G_1 \,\cup\, G_2 \,=\, G 
  \!\!\!\!
  }}
  \!\!\!
  \mc M_{G_1}(f_1) \bowtie \mc M_{G_2}(f_2) 
  \ \subseteq\ 
  \mc M_G(f_1) \cup \mc M_G(f_2) \cup
  \!\!
  \bigcup_{\substack{
  \!\!
  G_1,G_2 \,\subsetneqq\, G\,:
  \!\!
  \\ \vphantom{\ts|}
  \!\!\!\! 
  G_1 \,\cup\, G_2 \,=\, G
  \!\!\!\!
  }}
  \!\!\!
  \mc M_{G_1}(f_1) \bowtie \mc M_{G_2}(f_2).
$}
\end{itemize}
\end{la}

\begin{proof}
These containments follow from the elementary observation: for any monotone Boolean functions $h_1$ and $h_2$, every minterm of $h_1 \vee h_2$ is a minterm of $h_1$ or a minterm of $h_2$, and every minterm of $h_1 \wedge h_2$ is the union of a minterm of $h_1$ and a minterm of $h_2$. 
\end{proof}

The following corollary extends Lemma \ref{la:join0} to $m$-ary disjunctions and conjunctions. We do not directly use Corollary \ref{cor:join0} in what follows, but similar reasoning shows up in the proof of Lemma \ref{la:halcali}.

\begin{cor}[$m$-ary version of Lemma \ref{la:join0}]\label{cor:join0}
For every nonempty graph $\emptyset \ne G \subseteq \Path_k$, monotone functions $f_1,\dots,f_m : \{0,1\}^{kn^2} \to \{0,1\}$, we have
\begin{itemize}
\item
$\ds
  \mc M_G(\bigvee_{j=1}^m f_j)
  \ \subseteq\ 
  \bigcup_{j=1}^m\, \mc M_G(f_j),
  $
\item 
\mbox{$\ds\mc M_G(\bigwedge_{j=1}^m f_j)
  \ \subseteq\ 
  \!\!\!\!
  \bigcup_{\substack{\vphantom{\ts|}
  \!\!G_1,\dots,G_m \,\subseteq\, G \,:\!\!\\
  \vphantom{\ts|}
  \!\!G_1 \,\cup\, \cdots \,\cup\, G_m \,=\, G\!\!}}\  
  \bigbowtie_{j=1}^m\ \mc M_{G_j}(f_j)
  \ \subseteq\  
  \bigcup_{t=1}^{\|G\|}
  \bigcup_{\substack{\vphantom{\ts|}
  G_1,\dots,G_t \,\subseteq\, G\,:\\
  \vphantom{\ts|}
  G_1 \,\cup\, \cdots \,\cup\, G_t \,=\, G,\\
  \vphantom{\ts|}
  \hspace{-10pt}
    \forall s \,\in\, [t],\: 
    G_s \,\nsubseteq\, G_1 \,\cup\, \cdots \,\cup\, G_{s-1} 
  \hspace{-10pt}
  }}
  \bigcup_{\vphantom{\ts|}
  1 \,\le\, j_1 \,<\, \cdots \,<\, j_t \,\le\, m}\ 
  \,\bigbowtie_{s=1}^t\ \mc M_{G_s}(f_{j_s}).
$
}
\end{itemize}
\end{cor}

In the bottom inclusions $\mc M_G(\bigwedge_{j=1}^m f_j) \subseteq \bigcup \bigbowtie \ldots \subseteq \bigcup\bigcup\bigcup \bigbowtie \ldots$, the lefthand union is indexed over $(2^{\|G\|}-1)^m$ ($\le 2^{km}$) possibilities for $G_1,\dots,G_m$, while the righthand unions are indexed over at most $2^{\smash{\|G\|^2}} m^{\|G\|}$ ($\le 2^{\smash{k^2}} m^k$) possibilities for $t,j_1,\dots,j_t,G_1,\dots,G_t$. 
In particular, when $k \le \log\log n$ and $m \le n^{1/k}$, there are only $2^{k^2} m^k = n^{O(1)}$ indices in the righthand unions. 
This is essentially the reason why our $\SACzero$ formula lower bounds extend to $\bigwedge$-fan-in $n^{1/k}$.

\subsection{Strict join trees}

\begin{df}[Strict join trees]
For $G \subseteq \Path_k$, a $G$-join tree $T$ is {\em strict} if either $\|G\|\le 1$ and $T$ is a single node labeled $G$, or $\|G\| \ge 2$ and $T = \un{T_1}{T_2}$ where $T_1,T_2$ are $G_1,G_2$-join trees with $G_1,G_2 \subsetneqq G$. We denote the set of strict $G$-join trees by $\scr T_G$.
\end{df}

A simple counting argument shows that $|\scr T_G| \le 2^{2^{\|G\|}}$. For bounded $\sqq{}$- or $\semempty{}$-depth, we get an even better bound:

\begin{la}\label{la:counting}
For all $G \subseteq \Path_k$, there are at most $2^{\|G\|^{d+1}}$ distinct $T \in \scr T_G$ with $\sqq{}$-depth \tu(or $\semempty$-depth\tu) at most $d$.
\end{la}

\begin{proof}
We argue by induction on $d$. Let $\ell \defeq \|G\|$ and assume $\ell \ge 1$, since otherwise $\scr T_G$ is empty. In the case $d = 0$, we have $|\scr T_G| \le 1$. In the case $d=1$, we have $|\scr T_G| = \ell! < 2^{\ell^2}$. 

Suppose $d \ge 2$ and consider any $T = \sqq{T_1,\dots,T_m} \in \scr T_G$ where $m \ge 2$ and each $T_j$ is a $G_j$-join tree of $\sqq{}$-depth at most $d-1$. Note that $G_1 \cup \dots \cup G_m = G$ and $G_j \nsubseteq G_1 \cup \dots \cup G_{j-1}$ for all $j \in [m]$. 
It follows that $m \le \ell$, and there are at most $2^{\ell^2}$ possible sequences $(G_1,\dots,G_m)$.

For each $j \in [m]$, since $\|G_j\| \le \ell - 1$, there are at most $2^{(\ell-1)^d}$ possibilities for $T_j$ by the induction hypothesis. 
Therefore, the number of possible $T$ is at most
\[
  2^{\ell^2} \cdot (2^{(\ell-1)^d})^{\ell}
  =
  2^{\ell^2 + \ell(\ell-1)^d}
  \le
  2^{\ell^{d+1}},
\]
since $\ell^2 + \ell(\ell-1)^d \le \ell^{d+1}$ for all $\ell \ge 1$ and $d \ge 2$.

The exact same argument applies to $\semempty$-depth, since $\sem{T_1,\dots,T_m} \in \scr T_G$ also implies that $G_j \nsubseteq G_1 \cup \dots \cup G_{j-1}$ for all $j \in [m]$. 
\end{proof}

\begin{df}\label{df:MGT}
Let $\f$ be a monotone DeMorgan formula on $kn^2$ variables. For each graph $G \subseteq \Path_k$ and strict join tree $T \in \scr T_G$, we define a subset $\mc M_G^{(T)}(\f) \subseteq \mc M_G(\f)$ inductively as follows.
If $\|G\| \le 1$, then $\mc M^{(T)}_G(\f) \defeq \mc M_G(\f)$.
If $\|G\| \ge 2$ and $T = \un{T_1}{T_2}$, then
\[
  \mc M^{(T)}_G(\f) \defeq 
  \begin{cases}
    \emptyset \vphantom{\Big|}
    &\text{if $\f$ has depth $0$},\\
    \mc M_G(\f) \cap \Big(\mc M_G^{(T)}(\f_1) \cup \mc M_G^{(T)}(\f_2)\Big)
    &\text{if $\f = \f_1 \vee \f_2$},
    \vphantom{\bigg|}\\
    \mc M_G(\f) \cap 
      \Big(\mc M^{(T)}_G(\f_1) \cup \mc M^{(T)}_G(\f_2)
      \cup
      \Big(\mc M^{(T_1)}_{G_1}(\f_1) \bowtie \mc M^{(T_2)}_{G_2}(\f_2)\Big)
      \Big)
    &\text{if $\f = \f_1 \wedge \f_2$}.
  \end{cases}
\]
\end{df}

The next lemma follows directly from Lemma \ref{la:join0} and the definition of subsets $\mc M_G^{(T)}(\f)$.

\begin{la}\label{la:covering}
For every monotone DeMorgan formula $\f$ on $kn^2$ variables and graph $G \subseteq \Path_k$, we have
\[
  \bigcup_{T \in \scr T_G} \mc M^{(T)}_G(\f) = \mc M_G(\f).
\]
Moreover, if $\f$ has $\wedge$-left-depth $d$, then $\mc M^{(T)}_G(\f)$ is nonempty only for $T \in \scr T_G$ with $\sqq{}$-depth at most $d$.
It follows that there exists a strict join tree $T \in \scr T_G$ with $\sqq{}$-depth at most $d$ such that 
\[
  \mu(\mc M^{(T)}_G(\f)) 
  \ \ge\ 
  \frac{\mu(\mc M_G(\f))}{
  |\{T \in \scr T_G : T \text{ has $\sqq{}$-depth at most $d$}\}|
  }
  \ \ge\ 
  \frac{\mu(\mc M_G(\f))}{2^{\|G\|^{d+1}}}.
\] 
\end{la}

In the event that $\mc M_G(\f_0)$ is a $G$-pathset for every graph $G \subseteq \Path_k$ and subformula $\f_0$ of $\f$, we are able to bound on the size of $\f$ in terms of the pathset complexity of subsets $\mc M^{(T)}_G(\f) \subseteq \mc M_G(\f_0)$.

\begin{la}[Lemma 6.6 of \cite{rossman2018formulas}]
\label{la:halcali}
Suppose that $\f$ is a monotone DeMorgan formula on $kn^2$ variables with $\wedge$-depth $D$ such that $\mc M_G(\f_0)$ 
is a {$G$-pathset} for every graph $G \subseteq \Path_k$ and subformula $\f_0$ of $\f$.
Then for every $G \subseteq \Path_k$ and $T \in \scr T_G$, we have
\[
  \chi_T(\mc M^{(T)}_G(\f))
  \ \le\  
  \binom{D + \|G\| - 1}{\|G\| - 1} \cdot \sz(\f)
  \ \le\ 
  (D+1)^{\|G\|} \cdot \sz(\f).
\]
\end{la}

\begin{proof}
We prove the inequality by induction on $\f$.  In the base case where $\f$ is a constant, or $\f$ is a positive literal and $\|G\| \ge 2$, then $\mc M^{(T)}_G(\f)$ is empty, hence $\chi_T(\mc M^{(T)}_G(\f)) = 0$ and $\sz(\f) = 0$. In the remaining base case that $\f$ is a positive literal and $\|G\| = 1$, we $\chi_T( \cdot ) \le 1$ and $\sz(\f) = 1$.

The induction step is straightforward when $\f$ is a disjunction $\f_1 \vee \f_2$.
We are left with the case that $\f$ is a conjunction $\f_1 \wedge \f_2$. Here the induction step is trivial when $\|G\| = 1$. So we assume that $\|G\| \ge 2$ and $T = \un{T_1}{T_2}$.

For $i \in \{1,2\}$, note that $\|G_i\| \le \|G\| - 1$ and $\depth_{\wedge}(\f_i) \le D-1$ and $\sz(\f) = \sz(\f_1) + \sz(\f_2)$.
By definition of pathset complexity $\chi_T(\cdot)$ and the induction hypothesis applied to $\f_1$ and $\f_2$, 
\[
  \chi_T(\mc M^{(T)}_G(\f))
  &=
  \chi_T\Big(
    \mc M_G(\f) \cap 
    \Big(\mc M^{(T)}_G(\f_1) \cup \mc M^{(T)}_G(\f_2)
      \cup
      \Big(\mc M^{(T_1)}_{G_1}(\f_1) \bowtie \mc M^{(T_2)}_{G_2}(\f_2)\Big)
    \Big)
  \Big)
  \vphantom{\bigg|}\\
  &\le
  \chi_T(\mc M^{(T)}_G(\f_1))
  + \chi_T(\mc M^{(T)}_G(\f_2))
  + \max\big\{\chi_{T_1}(\mc M^{(T_1)}_{G_1}(\f_1)),\: 
    \chi_{T_2}(\mc M^{(T_2)}_{G_2}(\f_2))\big\}
  \vphantom{\bigg|}\\
  &\le 
  \sum_{i \in \{1,2\}}
  \Big(
  \binom{\depth_{\wedge}(\f_i) + \|G\| - 1}{\|G\| - 1} 
  +
  \binom{\depth_{\wedge}(\f_i) + \|G_i\| - 1}{\|G_i\| - 1}
  \Big) \cdot \sz(\f_i)\\
  &\le 
  \sum_{i \in \{1,2\}}
  \Big(
  \binom{D
  +  
  \|G\|
  - 
  2}{\|G\| - 1} 
  +
  \binom{D
  + 
  \|G\| 
  -  
  2}{
  \|G\| - 2}
  \Big) \cdot \sz(\f_i)\\
  &= 
  \binom{D  + \|G\| - 1}{\|G\| - 1} 
  \cdot \sz(\f).\qedhere
  \vphantom{{}_{\Big|}}
\]
\end{proof}

\begin{cor}\label{cor:halcali}
Suppose that $\f$ is a monotone DeMorgan formula on $kn^2$ variables with $\wedge$-left-depth $d$ and $\wedge$-depth $D$ such that $\mc M_G(\f_0)$ is a $G$-pathset for every graph $G \subseteq \Path_k$ and subformula $\f_0$ of $\f$.
Then
\[
  \sz(\f) 
  \ge 
  \frac{n^{\Omega(d(k^{1/d}-1))}}{2^{k^{d+1}} \cdot (D+1)^k}  \cdot \mu(\mc M_{\Path_k}(\f)).
\]
\end{cor}

\begin{proof}
By Lemma \ref{la:covering}, there exists a strict join tree $T \in \scr T_{\Path_k}$ such that
\[
  \mu(\mc M^{(T)}_{\Path_k}(\f)) \ge \frac{1}{2^{k^{d+1}}} \cdot \mu(\mc M_{\Path_k}(\f)).
\]
We now obtain the desired bound as follows:
\[
  &&&&&&
  \sz(\f)
  &\ge
  \frac{1}{(D+1)^k} \cdot \chi_T(\mc M^{(T)}_{\Path_k}(\f))
  &&\text{(by Lemma \ref{la:halcali})}
  &&&&&&\\
  &&&&&&
  &\ge
  \frac{n^{\Omega(d(k^{1/d}-1))}}{(D+1)^k} \cdot \mu(\mc M^{(T)}_{\Path_k}(\f))
  &&\text{(by Corollary \ref{cor:full-range}(I))}\\
  &&&&&&
  &\ge
  \frac{n^{\Omega(d(k^{1/d}-1))}}{2^{k^{d+1}}\cdot (D+1)^k}  \cdot \mu(\mc M_{\Path_k}(\f)).
  &&&&&&&&\qedhere
\]
\end{proof}

The hypothesis of Corollary \ref{cor:halcali} requires that $G$-minterm relations $\mc M_G(\f_0)$ are pathsets for all subformulas $\f_0$ of $\f$.  We point out that this condition may fail to hold even for relatively simple DeMorgan formulas $\f$.
For example, 
suppose $\f$ is the balanced DeMorgan conversion of the read-once CNF formula $\bigwedge_{i\in[k]}\ \bigvee_{(a,b) \in [n]^2}\ M^{(i)}_{a,b}$.
In this case, $\mc M_{\Path_k}(\f)$ is the complete $\Path_k$-relation $[n]^{\{0,\dots,k\}}$, which is \underline{not} a $\Path_k$-pathset since it has density $1$ ($>\n=\n^{\Delta(\Path_k)}$).

\subsection{Random restriction lemma}\label{sec:restriction}

In order to make use of Corollary \ref{cor:halcali}, we require a result (Lemma \ref{la:Xi}) showing that whenever a monotone function $f : \{0,1\}^{kn} \to \{0,1\}$ is computable by a reasonably small $\ACzero$ formula of reasonably bounded depth, the $G$-minterm relation of a certain {\em random restriction} of $f$ is a $G$-pathset with very high probability.

\begin{df}[Random matrices $\mb\xi \le \mb\zeta \in \{0,1\}^{n\times n}$ and graphs $\XI \subseteq \Path_k^{\uparrow n}$] 
\
\begin{itemize}
  \item
    Let $\mb\zeta \in \{0,1\}^{n \times n}$ be a random Boolean matrix with i.i.d.\ $\Bernoulli(
    n^{-1-\frac{1}{2k}}
    )$ entries.
  \item
    Let $\mb\xi \in \{0,1\}^{n \times n}$ be the sub-permutation matrix induced by $\mb\zeta$ as follows:
    \[
      \mb\xi_{a,b} &\defeq 
      \begin{cases}
        1 &\text{if } \mb\zeta_{a,b} = 1 \text{ and } \mb\zeta_{a',b} = \mb\zeta_{a,b'} = 0 \text{ for all } a' \ne a \text{ and } b' \ne b,\\
        0 &\text{otherwise.}
      \end{cases}
      \qquad
    \]    
  \item
    Let $\mb\zeta^{(1)},\dots,\mb\zeta^{(k)} \in \{0,1\}^{n \times n}$ be independent random matrices with distribution $\mb\zeta$, let $\mb\xi^{(1)},\dots,\mb\zeta^{(k)} \in \{0,1\}^{n \times n}$ be the induced sub-permutation matrices, and let 
    $\XI \defeq (\mb\xi^{(1)},\dots,\mb\zeta^{(k)})$ which we regard as an element of $\{0,1\}^{kn^2}$ as well as a subgraph $\XI \subseteq \Path_k^{\uparrow n}$.
\end{itemize}
\end{df}

\begin{df}[The restricted function $f^{\cup\XI}$]
\ 
\begin{itemize}
  \item
    For a Boolean function $f : \{0,1\}^{kn^2} \to \{0,1\}$, let $f^{\cup \XI} : \{0,1\}^{kn^2} \to \{0,1\}$ denote the random function defined by $f^{\cup \XI}(X) \defeq f(X \cup \XI)$ for all $X \subseteq \Path_k^{\uparrow n}$.  (Note that if $f$ is monotone, then so is $f^{\cup\XI}$.)
  \item
    For a $kn^2$-variable DeMorgan formula $\f$, let $\f^{\cup \XI}$ denote the formula obtained from $\f$ by substituting the constant $1$ (resp.\ $0$) at each positive (resp.\ negative) literal corresponding to an edge of $\XI$.  (Note that if $\f$ computes $f$, then $\f^{\cup\XI}$ computes $f^{\cup\XI}$.)
\end{itemize}
\end{df}

The next lemma extends Observation \ref{obs:Mpath} concerning the $\Path_k$-minterms of monotone functions that compute $\SPMM_{n,k}$.

\begin{la}\label{la:Mpath2}
If $f : \{0,1\}^{kn^2} \to \{0,1\}$ is a monotone function which agrees with $\SPMM_{n,k}$ on every $k$-tuple of sub-permutation matrices, then
\[
  \Pr\Big[\ \mu(\mc M_{\Path_k}(f^{\cup \XI})) \ge \frac{1}{2n^2}\ \Big] = 1-o(1).
\]
\end{la}

\begin{proof}
Let $\mb\zeta^{(1)},\dots,\mb\zeta^{(k)} \in \{0,1\}^{n \times n}$ be the random matrices which generate the $\mb\xi^{(1)},\dots,\mb\xi^{(k)} \in \{0,1\}^{n \times n}$ comprising $\XI$.
As noted in Observation \ref{obs:Mpath}, we have
$\mc M_{\Path_k}(f) = \{\alpha \in [n]^{\{0,\dots,k\}} : \alpha_0=\alpha_k=1\}$.
For each $\alpha$ in this set, a sufficient condition for $\alpha \in \mu(\mc M_{\Path_k}(f^{\cup \XI}))$ is that 
$ 
  \mb\zeta^{(i)}_{\smash{a,\alpha_i}} =
  \mb\zeta^{(i)}_{\smash{\alpha_{i-1},b}} = 0
$
for all $a,b \in [n]$. (In this case, $\XI \cup G^{(\alpha)}$ is a $k$-tuple of sub-permutation matrices for every $G \subseteq \Path_k$, from which it follows that $\alpha \in \mu(\mc M_{\Path_k}(f^{\cup \XI}))$.)
This sufficient condition occurs with probability $(1-n^{-1-\frac{1}{2k}})^{k(2n-1)} = 1 - O(kn^{-\frac{1}{2k}}) = 1 - o(1)$.
Therefore, by linearity of expectation,
\[
  \Ex\big[\ \mu(\mc M_{\Path_k}(f^{\cup \XI}))\ \big] \ge (1-o(1)) \cdot \mu(\mc M_{\Path_k}(f)) = (1-o(1)) \cdot\frac{1}{n^2}.
\]
A straightforward application of Janson's inequality \cite{janson1990poisson} shows that $\mu(\mc M_{\Path_k}(f^{\cup \XI}))$ is at least half its expectation with very high probability.
\end{proof}

The next definition gives a version of the relation $\mc M_G(f)$ for non-monotone functions $f$.  

\begin{df}[The relation $\mc N_G(f)$]\ 
\begin{itemize}
\item
For a (not necessarily monotone) function $f : \{0,1\}^{kn^2} \to \{0,1\}$ and a graph $G \subseteq \Path_k$ and $\alpha \in [n]^{V(G)}$, let $f{\uhr}G^{(\alpha)} : \{0,1\}^{\|G\|} \to \{0,1\}$ denote restricted subfunction $f{\uhr}G^{(\alpha)}(H) \defeq H^{(\alpha)}$ where we identify each input $H \in \{0,1\}^{\|G\|}$ with the corresponding subgraph of $G$.
\item
The relation $\mc N_G(f) \subseteq [n]^{V(G)}$ by
\[
  \mc N_G(f) &\defeq \{\alpha \in [n]^{V(G)} : f{\uhr}G^{(\alpha)} \text{ depends on all $\|G\|$ coordinates}\}.
\]
\end{itemize}
\end{df}

For monotone functions $f$, note that $\mc M_G(f) \subseteq \mc N_G(f)$ (hence if $\mc N_G(f)$ is a $G$-pathset, then so is $\mc M_G(f)$).

Since it will be useful later, we state here a variant of Lemma \ref{la:halcali} with $\mc N_G(\cdot)$ in place of $\mc M_G(\cdot)$. (The proof is identical.)

\begin{la}\label{la:halcali2}
Suppose that $\f$ is a DeMorgan formula on $kn^2$ variables with depth $D$ such that $\mc N_G(\f_0)$ 
is a {$G$-pathset} for every graph $G \subseteq \Path_k$ and subformula $\f_0$ of $\f$.
Then for every $G \subseteq \Path_k$ and $T \in \scr T_G$, we have
\[
  \chi_T(\mc N^{(T)}_G(\f))
  &\le 
  (D+1)^{\|G\|}
  \cdot \sz(\f).
\]
\end{la}

We conclude this subsection by stating a key lemma from \cite{rossman2018formulas}, which implies that relations $\mc M_G(f^{\cup \XI})$ are likely to be $G$-pathsets whenever $f$ is computable by monotone $\ACzero$ formulas of reasonable size and depth.   

\begin{la}[Lemma 6.5 of \cite{rossman2018formulas}]\label{la:Xi}
Suppose that $f : \{0,1\}^{kn^2} \to \{0,1\}$ is computable by a (not necessarily monotone) $\ACzero$ formula of size at most $n^k$ and depth at most $\frac{\log n}{(\log\log n)^6}$. Then for every $G \subseteq \Path_k$, we have
\[
  \Pr\big[\ 
    \mc N_G(f^{\cup \XI}) 
    \text{ is \underline{not} a $G$-pathset}
  \ \big]
  &\le
  O(n^{-2k}).
\]
\end{la}

\subsection{Tradeoff (I)$^+$ for monotone $\SACzero$ formulas}\label{sec:I+}

We are ready to prove our size-depth tradeoff for monotone $\SACzero$ formulas.

\begin{thm}[Tradeoff (I)$^+$ of Theorem \ref{thm:AC0tradeoffs}]\label{thm:I+}
Suppose that $\ff F$ is a monotone $\SACzero$ formula of $\bigwedge$-depth $d$ and $\bigwedge$-fan-in $n^{1/k}$ which computes $\SPMM_{n,k}$ where $k \le \log\log n$
and $d \le \log k$.  Then $\ff F$ has size $n^{\Omega(dk^{1/d})}$.
\end{thm}

\begin{proof}
Without loss of generality, assume that $\ff F$ has size at most $n^k$, since this is greater than the lower bound we wish to show. Let $\f$ be the right-deep DeMorgan conversion of $\ff F$.  Note that even though $\f$ has very large depth, 
it has $\wedge$-left-depth $d$ and $\wedge$-depth at most $dn^{1/k}$.
Moreover, each subformula $\f_0$ of $\f$ is equivalent to a monotone $\ACzero$ formula of $\bigwedge$-depth at most $d$ (and depth at most $2d+1 = O(\log\log n)$) and size at most $n^k$.  In particular, $\f_0$ satisfies the hypothesis
of Lemma \ref{la:Xi}.
    
We now observe that each of the following statements hold with probability $1-o(1)$: 
\begin{itemize}
  \item
    $\mc M_G(\f_0^{\cup\XI})$ is a $G$-pathset for every subformula $\f_0$ of $\f$ and $G \subseteq \Path_k$ 
    
    \mbox{}\hfill
    (by Lemma \ref{la:Xi} and a union bound over at most $O(n)^k$ choices of $\f_0$ and $G$),
  \item
    $\ds \sz(\f^{\cup \XI}) \ge \frac{n^{\Omega(d(k^{1/d}-1))}}{2^{k^{d+1}} \cdot (dn^{1/k}+1)^k} \cdot \mu(\mc M_{\Path_k}(\f^{\cup \XI}))$
    \hfill
    (by Corollary \ref{cor:halcali}),
  \item
    $\ds \sz(\f^{\cup \XI}) \ge \frac{n^{\Omega(d(k^{1/d}-1))}}{2^{k^{d+1}} \cdot (dn^{1/k}+1)^k \cdot 2n^2}$
    \hfill
    (by Lemma \ref{la:Mpath2}).
\end{itemize}
Since $\sz(\ff F) = \sz(\f) \ge \sz(\f^{\cup\XI})$, 
it follows that
\[
  \sz(\ff F) 
  &\ge 
  \frac{n^{\Omega(d(k^{1/d}-1))}}{2^{k^{d+1}} \cdot (dn^{1/k}+1)^k \cdot 2n^2}.
\]
Since $k \le \log\log n$ and $d \le \log k$, we have $2^{k^{d+1}} \le n^{o(1)}$ and $(dn^{1/k}+1)^k = n^{1+o(1)}$, so the {denominator} above is $n^{O(1)}$, whereas the {numerator} 
is at least $n^{\Omega(\log k)}$.  
We conclude that $\ff F$ has size $n^{\Omega(d(k^{1/d}-1))}$ as required.
\end{proof}

Our proof of tradeoff (I)$^+$ shows that we can slightly extend the range of $k$ from $\log\log n$ to $2^{\sqrt{\log\log n}}$ and $\bigwedge$-fan-in from $n^{1/k}$ to $n^{(\log k)/k}$.
While the range of $k$ can probably be pushed further, the proof breaks down for larger $\bigwedge$-fan-in $D = n^{\omega((\log k)/k)}$; this is ultimately due to the $D^k$ factor in the bound of Lemma \ref{la:halcali}. 
Also notice that our $n^{\Omega(dk^{1/d})}$ lower bound for monotone $\SACzero$ formulas of $\bigwedge$-depth $d$ cannot possibly extend to unbounded fan-in $\ACzero$ formulas of $\bigwedge$-depth $d$, in light of the $n^{O(dk^{1/(2d+1)})}$ upper bound given by the monotone $\Sigma_{2d+1}$ formulas of Proposition \ref{prop:upper}(II).

\subsection{Strict support tree of a $k$-variable DeMorgan formula}\label{sec:support-tree}

We next define notions of {\em strictness} and {\em $\semempty{}$-depth} for DeMorgan formulas, which are analogous to the corresponding notions for join trees.  Even though we only apply these notions to $k$-variable DeMorgan formulas, we state the definitions more generally without reference to $k$.

\begin{df}[Strict DeMorgan formulas and operations $\strict{\g}$ and $\strictjoin{T}$]
\ 
\begin{itemize}
\item
A DeMorgan formula $\g$ is {\em strict} if either
\begin{itemize}
  \item
    it has depth $0$ (i.e.,\ $\g$ is a constant or literal), or
  \item
    it is $\g_1 \wedge \g_2$ or $\g_1 \wedge \g_2$ where $\g_1,\g_2$ are strict and  
    compute distinct functions from $\g$ (i.e.,\ $\g_1 \not\equiv \g$ and $\g_2 \not\equiv \g$).
\end{itemize}
  \item
For every DeMorgan formula $\g$, we define an equivalent strict DeMorgan formula $\strict{\g}$ inductively as follows:
\begin{itemize}
  \item
    if $\g$ has depth $0$, then $\strict{\g} \defeq \g$,
  \item
    if $\g$ is $\g_1 \wedge \g_2$ (resp.\ $\g_1 \vee \g_2$), then
    \[
      \strict{\g}
      \defeq
      \begin{cases}
        \strict{\g_1} &\text{if } \g_1 \equiv \g,\\
        \strict{\g_2} &\text{if } \g_1 \not\equiv \g \text{ and } \g_2 \equiv \g,\\
        \strict{\g_1} \wedge \strict{\g_2} \text{ (resp.\ $\strict{\g_1} \vee \strict{\g_2}$)} &\text{otherwise.}
      \end{cases}
    \]
\end{itemize}
(Note that $\g$ is strict if, and only if, $\g = \strict{\g}$. On the other hand, $\g \equiv \strict{\g}$ for all $\g$.)
\item
For every 
$G$-join tree $T$, we define 
a strict $G$-join tree 
$\strictjoin{T}$  
in a similar fashion: 
\begin{itemize}
  \item
    if $T$ has $\joinop$-depth $0$, then $\strictjoin{T} \defeq T$,
  \item
    if $T = \un{T_1}{T_2}$ where $T_1,T_2$ are $G_1,G_2$-join trees, then
    \[
      \strictjoin{T} 
      \defeq 
      \begin{cases}
        \strictjoin{T_1} &\text{if } G_1 = G,\\
        \strictjoin{T_2} &\text{if } G_1 \ne G \text{ and } G_2 = G,\\
        \un{\strictjoin{T_1}}{\strictjoin{T_2}} &\text{otherwise.}
      \end{cases}
    \]
\end{itemize}
\end{itemize}
\end{df}

To help the reader keep track of the types of objects, we use typewriter font for the $\ff{strict}$ operator on DeMorgan formulas and italics for the $\mathit{strict}$ operator on join trees.  (Later on in Definition \ref{df:support-tree}, we use sans-serif $\Supp(\g)$ for a subgraph of $\Path_k$ and italic $S(\g)$ for a join tree associated with $\g$.)

\begin{df}[$\semempty$-depth of DeMorgan formulas]
\ 
\begin{itemize}
\item
For any DeMorgan formulas $\g_1,\dots,\g_m$, we define DeMorgan formulas $\sem{\g_1,\dots,\g_m}_{\wedge}$ and $\sem{\g_1,\dots,\g_m}_{\vee}$ inductively by
\[
  &&&&
  \sem{\g_1}_{\wedge} 
  =
  \sem{\g_1}_{\vee} 
  &\defeq
  \g_1,
  &&\text{for }m=1,\vphantom{\big|}&&&&\\
  &&&&
  \sem{\g_1,\dots,\g_m}_{\wedge} 
  &\defeq 
  \sem{\g_1,\dots,\g_{m-1}}_{\wedge} \wedge \sem{\g_1,\dots,\g_{m-2},\g_m}_{\wedge}
  &&\text{for }m\ge 2,\vphantom{\big|}
  \\
  &&&&\sem{\g_1,\dots,\g_m}_{\vee} 
  &\defeq 
  \sem{\g_1,\dots,\g_{m-1}}_{\vee} \vee \sem{\g_1,\dots,\g_{m-2},\g_m}_{\vee}
  &&\text{for }m\ge 2.\vphantom{\big|}
\]
\item
The {\em $\semempty{}$-depth} of a DeMorgan formula $\g$ is defined inductively as follows:
\begin{itemize}
  \item
    if $\g$ has depth $0$, then $\depth_{\semempty{}}(\g) \defeq 0$,
  \item
    if $\g$ has depth $\ge 1$, then 
    \[
      \depth_{\semempty{}}(\g) 
      \defeq
      1 +
      \min_{\substack{m \,\ge\, 2,\  \g_1,\dots,\g_m \,:\\
      \g \,=\, \sem{\g_1,\dots,\g_m}_{\wedge} 
      \text{ or } \sem{\g_1,\dots,\g_m}_{\vee}}}
      \max_{i\in[m]}\ \depth_{\semempty{}}(\g_i).
    \]
\end{itemize}
\end{itemize}
\end{df}

We now restrict attention to $k$-variable DeMorgan formulas, whose variables we identify with edges of $\Path_k$ and whose inputs in $\{0,1\}^k$ we identify with subgraphs of $\Path_k$.

\begin{la}\label{la:strict-count}
There are at most 
$2^{2^{(d+1)(k+1)}}$
distinct strict $k$-variable DeMorgan formulas of $\semempty{}$-depth at most $d$.
\end{la}

\begin{proof}
Consider any strict $k$-variable DeMorgan formula $\g$ of $\semempty{}$-depth $d$. Without loss of generality, 
\[
  \g = \sem{\g_1,\dots,\g_m}_{\wedge}
\]
for some $\g_1,\dots,\g_m$ of $\semempty{}$-depth at most $d-1$. Note that $\g_1,\dots,\g_m$ are necessarily also strict.

Let $g_1,\dots,g_m : \{0,1\}^k \to \{0,1\}$ be the functions computes by $\g_1,\dots,\g_m$. Strictness of $\g$ implies that 
\[
  g_1,\quad g_1 \wedge g_2,\quad g_1 \wedge g_2 \wedge g_3,\quad \dots,\quad g_1 \wedge \dots \wedge g_m
\]
is a strictly decreasing sequence of Boolean functions. Moreover, $g_1 \wedge \dots \wedge g_m \not\equiv 0$ assuming $m \ge 1$.
It follows that $m \le 2^k$.  (Whereas there are $2^{2^k}$ distinct Boolean functions $\{0,1\}^k \to \{0,1\}$, the longest decreasing sequence of non-zero functions has length $2^k$.)

We conclude that
\begin{multline*}
  \big|\big\{\tu{strict $k$-variable DeMorgan formulas of $\semempty{}$-depth at most $d$}\big\}\big|\\
  \le
  2 + 2k +
  2 \big|\big\{\tu{strict $k$-variable DeMorgan formulas of $\semempty{}$-depth at most $d-1$}\big\}\big|\vphantom{\big|}^{2^k}.
\end{multline*}
(Here the additive $2+2k$ counts the depth-$0$ formulas, and the other factor $2$ arises from the choice of top gate type.)  The bound of the lemma follows since $2 + 2k + 2^{2^{d(k+1)}} \le 2^{2^{(d+1)(k+1)}}$ for all $d,k \ge 0$.
\end{proof}

The next definition associates a
join tree $S(\g)$ with each $k$-variable DeMorgan formula $\g$.

\begin{df}[Support, restriction, and support tree] 
\label{df:support-tree}
\ 
\begin{itemize}
\item
The {\em support} of a Boolean function $g : \{0,1\}^k \to \{0,1\}$ is the set $\Supp(g) \subseteq [k]$ of coordinates on which $g$ depends, that is,
\[
  \Supp(g) 
  &\defeq 
  \{i \in [k] : \exists y \in \{0,1\}^k,\  
  g(y) \ne g(y_1,\dots,y_{i-1},1-y_i,y_{i+1},\dots,y_k)\}.
\]
(We identify $\Supp(g)$ with the corresponding subgraph of $\Path_k$.)

The {\em support} of a $k$-variable DeMorgan formula $\g$, denoted $\Supp(\g)$, is the support of the Boolean function it computes.
\item
For a $k$-variable DeMorgan formula $\g$ and a subgraph $H \subseteq \Path_k$, let $\g{\uhr}H$ denote the DeMorgan formula obtained from $\g$ by syntactically relabeling to $0$ (resp.\ $1$) each leaf labeled by a positive (resp.\ negative) literal whose coordinate lies in $E(\Path_k) \setminus E(H)$.

(Observe that $\Supp(\g{\uhr}H) \subseteq H$ and $\g \equiv \g{\uhr}\Supp(\g)$, even though $\g$ and $\g{\uhr}\Supp(\g)$ are not necessarily equal.) 
\item
For every $k$-variable DeMorgan formula $\g$, we define a join tree $\SuppTree(\g)$, called the {\em support tree} of $\g$, inductively as follows:
\begin{itemize}
  \item
    if $\g$ is a constant (i.e.,\ $0$ or $1$), then $\SuppTree(\g)$ is the empty join tree;
  \item
    if $\g$ is the $i$th positive or negative literal, then $\SuppTree(\g)$ is the single-node join tree labeled by $E_i$ (i.e.,\ the $i$th single-edge subgraph of $\Path_k$);
  \item
    if $\g$ is $\g_1 \wedge \g_2$ or $\g_1 \vee \g_2$, then
    $
      \SuppTree(\g) 
      \defeq
      \SuppTree(\g_1{\uhr}\Supp(\g)) \cup \SuppTree(\g_2{\uhr}\Supp(\g)).
    $
\end{itemize}
Note that $\SuppTree(\g)$ is a $\Supp(\g)$-join tree. (Had we instead defined $\SuppTree(\g) 
      \defeq
      \SuppTree(\g_1) \cup \SuppTree(\g_2)$, then $\SuppTree(\g)$ would be a $G$-join tree where $G$ is the supergraph of $\Supp(\g)$ that contain the edges of $\Path_k$ corresponding to all literals that appear in $\g$.)
\item
We refer to $\strictjoin{\SuppTree(\g)}$ as the {\em strict support tree of $\g$}. Note that $\strictjoin{\SuppTree(\g)}$ is a strict $\Supp(\g)$-join tree (i.e.,\ an element of $\scr T_{\Supp(\g)}$).
\end{itemize}
\end{df}

The next lemma follows directly from definitions.

\begin{la}\label{la:semempty-depth}
For every $k$-variable DeMorgan formula $\g$, the $\semempty{}$-depth of the join tree $\SuppTree(\g)$ (and hence also $\strictjoin{\SuppTree(\g)}$) is at most the $\semempty{}$-depth of $\g$.\qed
\end{la}

\subsection{The distribution of $\strict{\scr D_t(\ff G)}$ for a $k$-variable $\ACzero$ formula $\ff G$}

In this subsection, we prove a key result (Lemma \ref{la:epsd}) on the distribution of $\strict{\g}$ for the randomized balanced DeMorgan conversion $\g \sim \scr D_t(\ff G)$ (Def.\ \ref{df:scrD}) 
of a $k$-variable $\ACzero$ formula $\ff G$.
We begin with a preliminary lemma on the distribution of $\strictjoin{\mb T_t}$ for a certain family of random join trees $\mb T_t$.

\begin{la}\label{la:eps1}
Fix any real numbers $p_1 \ge \dots \ge p_k > 0$ with $p_1 + \dots + p_k = 1$.
For each $t \in \N$, let $\mb T_t$ be a random join tree of $\joinop$-depth $t$ with $2^t$ leaves independently labeled by single-edge graphs $E_1,\dots,E_k$ with probability $p_1,\dots,p_k$ respectively.
Then for all $t \ge \log(\frac{1}{p_k}) + k$, we have
\[
  \Pr\big[\ \strictjoin{\mb T_t} = \sem{E_1,\dots,E_k}\ \big]
  \ge
  \exp(-O(2^k)).
\]
\end{la}

A key point in Lemma \ref{la:eps1} is that the bound $\exp(-O(2^k))$ is independent of $t$ and the distribution $p_1,\dots,p_k$).

\begin{proof}
First, consider the case of the uniform distribution
$
  p_1 = \dots = p_k = \frac1k$.
Here we have
\[
  \Pr\big[\ \strictjoin{\mb T_k} = \sem{E_1,\dots,E_k}\ \big]
  =
  \Pr\big[\ \mb T_k = \sem{E_1,\dots,E_k}\ \big]
  =
  k^{-k}.
\]
Note that for any $t \ge k$, we have 
\[
  \Pr\big[\ \strictjoin{\mb T_t} = \sem{E_1,\dots,E_k}\ \big]
  \ge
  \Pr\big[\ \text{the leftmost depth-$k$ sub-join tree of $\mb T_t$ equals } \sem{E_1,\dots,E_k}\ \big]
  =
  k^{-k}.
\]

To illustrate the general argument, 
consider the following (essentially geometric) distribution with $k=4$:
\[
  p_1 \approx 0.99,
  \qquad p_2 = 2^{-10},\qquad p_3 = 2^{-100},\qquad p_4 = 2^{-1000}.
\]
First, note that
\[
  \Pr\big[\ \mb T_{1000} \text{ contains a single } E_4 
  \ \big]
  =
  \Pr\big[\ \Binomial(2^{1000},2^{-1000}) = 1\ \big]
  \approx 
  \frac{1}{\Exp{}}.
\]
Letting $\mb T_{999}^{(\Le)}$ and $\mb T_{999}^{(\Ri)}$ be the left and right children of the root of $\mb T_{1000}$, we next observe that
\[
  \Pr\big[\ 
  \underbrace{\mb T_{999}^{(\Le)} \text{ contains no } E_4
  \text{ and } \mb T_{999}^{(\Ri)} \text{ contains a single } E_4}_{\ts\tu{call this event }\mc E}
  \ \big]
  \approx 
  \frac{1}{\Exp{}^2}.
\]
Conditioned on event $\mc E$, let $\mb T_{100}^{(\Le)}$ be the leftmost depth-$k$ sub-join tree of $\mb T_{999}^{(\Le)}$, and let $\mb T_{100}^{(\Ri)}$ be the unique depth-$k$ sub-join tree of $\mb T_{999}^{(\Ri)}$ that contains $E_4$.
Now observe that a sufficient condition for $\strictjoin{\mb T_{1000}} = \sem{E_1,E_2,E_3,E_4}$ is that
\begin{enumerate}[\quad(a)]
  \item
    event $\mc E$ holds,
  \item
    $\strictjoin{\mb T_{100}^{(\Le)}} = \sem{E_1,E_2,E_3}$, and
  \item
    $\strictjoin{\mb T_{100}^{(\Ri)}} = \sem{E_1,E_2,E_4}$.
\end{enumerate}
(Given $\mc E$, condition (b) implies that $\strictjoin{\mb T_{999}^{(\Le)}} = \sem{E_1,E_2,E_3}$ and (c) implies that $\strictjoin{\mb T_{999}^{(\Ri)}} = \sem{E_1,E_2,E_4}$.  Conditions (a),(b),(c) therefore together imply that $\strictjoin{\mb T_{1000}} = \sem{E_1,E_2,E_3,E_4}$.)
We now have
\[
  \Pr\big[\ \strictjoin{\mb T_{1000}} = {}&\sem{E_1,E_2,E_3,E_4}\ \big]\\
  &\ge
  \Pr\Big[\ 
    \mc E 
    \text{ and }
    \strictjoin{\mb T_{100}^{(\Le)}} = \sem{E_1,E_2,E_3}
    \text{ and }
    \strictjoin{\mb T_{100}^{(\Ri)}} = \sem{E_1,E_2,E_4}
  \ \Big]\\
  &\approx
  \frac{1}{\Exp{}^2}
  \cdot
  \Pr\Big[\ 
    \strictjoin{\mb T_{100}^{(\Le)}} = \sem{E_1,E_2,E_3}
  \ \Big|\ 
    \mc E
  \ \Big]
  \cdot
  \Pr\Big[\ 
    \strictjoin{\mb T_{100}^{(\Ri)}} = \sem{E_1,E_2,E_4}
  \ \Big|\ 
    \mc E
  \ \Big].
\]

Splitting $\mb T_{100}^{(\Le)}$ (resp.\ $\mb T_{100}^{(\Ri)}$) into left and right subtrees $\mb T_{99}^{(\Le\Le)}$ and $\mb T_{99}^{(\Le\Ri)}$ (resp.\ $\mb T_{99}^{(\Ri\Le)}$ and $\mb T_{99}^{(\Ri\Ri)}$), we continue this analysis:
\[
  \Pr\Big[\ 
  \underbrace{\mb T_{99}^{(\Le\Le)} \text{ contains no } E_3
  \text{ and } \mb T_{99}^{(\Le\Ri)} \text{ contains a single } E_3}_{\ts\tu{call this event }\mc E^{(\Le)}} 
  \ \Big|\ 
  \mc E
  \ \Big]
  &\approx 
  \frac{1}{\Exp{}^2},\\
  \Pr\Big[\ 
  \underbrace{\mb T_{99}^{(\Ri\Le)} \text{ contains no } E_4
  \text{ and } \mb T_{99}^{(\Ri\Ri)} \text{ contains a single } E_4}_{\ts\tu{call this event }\mc E^{(\Ri)}}
  \ \Big|\ 
  \mc E
  \ \Big]
  &\approx 
  \frac{1}{\Exp{}^2}.
\]
Next, for appropriate definitions of sub-join trees $\mb T_{10}^{(\Le\Le)},\mb T_{10}^{(\Le\Ri)}$ of $\mb T_{99}^{(\Le)}$ and $\mb T_{10}^{(\Ri\Le)},\mb T_{10}^{(\Ri\Ri)}$ of $\mb T_{99}^{(\Ri)}$, we have
\[
  \Pr\Big[\ 
    \strictjoin{\mb T_{100}^{(\Le)}} = {}&\sem{E_1,E_2,E_3}
  \ \Big|\ 
    \mc E
  \ \Big]\\
  &\approx
  \frac{1}{\Exp{}^2}
  \cdot
  \Pr\Big[\ 
    \strictjoin{\mb T_{10}^{(\Le\Le)}} = \sem{E_1,E_2}
  \ \Big|\ 
    \mc E \wedge \mc E^{(\Le)}
  \ \Big]
  \cdot
  \Pr\Big[\ 
    \strictjoin{\mb T_{10}^{(\Le\Ri)}} = \sem{E_1,E_3}
  \ \Big|\ 
    \mc E \wedge \mc E^{(\Le)}
  \ \Big],\\
  \Pr\Big[\ 
    \strictjoin{\mb T_{100}^{(\Ri)}} = {}&\sem{E_1,E_2,E_4}
  \ \Big|\ 
    \mc E
  \ \Big]\\
  &\approx
  \frac{1}{\Exp{}^2}
  \cdot
  \Pr\Big[\ 
    \strictjoin{\mb T_{10}^{(\Ri\Le)}} = \sem{E_1,E_2}
  \ \Big|\ 
    \mc E \wedge \mc E^{(\Ri)}
  \ \Big]
  \cdot
  \Pr\Big[\ 
    \strictjoin{\mb T_{10}^{(\Ri\Ri)}} = \sem{E_1,E_4}
  \ \Big|\ 
    \mc E \wedge \mc E^{(\Ri)}
  \ \Big].
\]

This analysis continues down one more layer. In the end we get a bound
\[
  \Pr\big[\ \strictjoin{\mb T_{1000}} = \sem{E_1,E_2,E_3,E_4}\ \big]
  &\gtrsim
  \frac{1}{\Exp{}^2}
  \left(
    \frac{1}{\Exp{}^2}
    \left(
      \frac{1}{\Exp{}^2}
    \right)
    \left(
      \frac{1}{\Exp{}^2}
    \right)
  \right)
  \left(
    \frac{1}{\Exp{}^2}
    \left(
      \frac{1}{\Exp{}^2}
    \right)
    \left(
      \frac{1}{\Exp{}^2}
    \right)
  \right)
  =
  \frac{1}{\Exp{}^{14}}.
\]
Generalizing this example to larger $k$, we get a lower bound
\[
  \Pr\big[\ \strictjoin{\mb T_{10^k}} = \sem{E_1,\dots,E_k}\ \big]
  &\gtrsim
  \exp(-2(2^k-1)),
\]
which is essentially the worst case for the lemma.

To generalize this argument for an arbitrary distribution $p_1 \ge \dots \ge p_k > 0$, in place of depths $1,10,100,1000$ in the example above, we consider the sequence $1 = t_1 < t_2 < \dots < t_k$ defined by
\[
  t_i \defeq \Big\lfloor\log\frac{1}{p_i+\dots+p_k}\Big\rfloor+i.
\]
The same analysis yield the desired $\exp(-O(2^k))$ lower bound.
\end{proof}

Simply rephrasing Lemma \ref{la:eps1} in terms of depth-$1$ $\ACzero$ formulas, we get the following:

\begin{cor}\label{cor:eps1}
For every $k$-variable $\Pi_1$ or $\Sigma_1$ formula $\ff G$ (i.e.,\ an $m$-ary $\bigwedge$ or $\bigvee$ of literals, where $m$ may be arbitrarily large relative to $k$),
there exists an {equivalent strict} DeMorgan formula $\ff G^\ast$ of $\semempty{}$-depth at most $1$ such that for all $t \ge 2^k$, we have
\[
  \Pr_{\g \sim \scr D_t(\ff G)}
  \big[\ 
  \strict{\g} 
  =
  \ff G^\ast\ \big]
  \ge
  \exp(-O(2^k)).
\]
\end{cor}

\begin{proof}
Without loss of generality, it suffices to consider the case that $\ff G$ is a monotone $\Pi_1$ formula of size $m = m_1+\dots+m_k$ with $m_i$ leaves labeled by the positive literal $X_i$, where $m_1 \ge \dots \ge m_k \ge 1$. Further assume w.l.o.g.\ that $m$ is a power of $2$. Letting $p_i \defeq \frac{m_i}{m}$ for $i \in [k]$, note that $\g \sim \scr D_t(\ff G)$ is identical (after renaming $\wedge$ to $\cup$ and $X_i$ to $E_i$) to the random join tree $\mb T_{\log(tm)}$ of Lemma \ref{la:eps1}.
Since $\log(tm) \ge k + \log m \ge k + \log\frac{1}{p_k}$, 
Lemma \ref{la:eps1} implies
\[
  \Pr_{\g \sim \scr D_t(\ff G)}
  \big[\ \strict{\g} = \sem{X_1,\dots,X_k}_{\wedge}\ \big]
  &\ge
  \exp(-O(2^k)).
\]
The desired statement is proved by letting $\ff G^\ast \defeq \sem{X_1,\dots,X_k}_{\wedge}$, which is clearly strict, equivalent to $\ff G$, and has $\semempty{}$-depth $1$.
\end{proof}

The next lemma extends Corollary \ref{cor:eps1} to $\ACzero$ formulas $\ff G$ beyond depth $1$.
For the analysis to work out, we consider the distribution $\scr D_t(\ff G)$ for larger $t$, which allows us to condition on the event that $\g_0 \sim \scr D_t(\ff G_0)$ is equivalent to $\ff G_0$ for every sub-formula $\ff G_0$ of $\ff G$.

\begin{la}\label{la:epsd}
For every $k$-variable $\ACzero$ formula $\ff G$ of depth $d$, there exists an {equivalent strict} DeMorgan formula $\ff G^\ast$ of $\semempty{}$-depth at most $d$ such that for all $t \ge 2^{2^{2^k}} \cdot (\log\sz(\ff G))^3$, we have
\[
  \Pr_{\g \sim \scr D_t(\ff G)}
  \big[\ \strict{\g} = \ff G^\ast\ \big]
  \ge
  2^{-\exp(O(d2^{2^k}))}.
\]
\end{la}

\begin{proof}
The case $d=0$ is trivial and Corollary \ref{cor:eps1} proves the case $d=1$.

For the induction step when $d \ge 2$, assume that $\ff G = \bigwedge_{i=1}^m \ff G_i$.
(The argument for $\ff G = \bigvee_{i=1}^m \ff G_i$ is identical after swapping $\vee$ for $\wedge$.) 
By the induction hypothesis, there exist $k$-variable DeMorgan formulas $\ff G_1^\ast,\dots,\ff G_m^\ast$ which satisfy the lemma with respect to subformulas $\ff G_1,\dots,\ff G_m$.

Let $h_1,\dots,h_s : \{0,1\}^k \to \{0,1\}$ enumerate the distinct functions computed by sub-formulas $\ff G_1,\dots,\ff G_m$, 
ordered such that $m_1 \ge \dots \ge m_s$ where
\[
  m_j \defeq \big|\big\{i \in [m] : \ff G_i \text{ computes } h_j\big\}\big|.
\] 
Next, for each $j \in [s]$, consider the set 
\[
  \big\{\ff G_i^\ast : i \in [m]\text{ such that $\ff G_i$ computes }h_j\big\}.
\]
Each element of this set is a strict $k$-variable DeMorgan formula of $\semempty{}$-depth at most $d-1$, hence the size of this set is at most $2^{2^{d(k+1)}}$ by Lemma \ref{la:strict-count}. 
It follows that there exist strict $k$-variable DeMorgan formulas $\ff h_1,\dots,\ff h_s$ of $\semempty{}$-depth at most $d-1$ such that for all $j \in [s]$,
\[
  \big|\big\{i \in [m] : \ff G_i^\ast = \ff h_j\big\}\big| 
  \ge 
  \frac{m_j}{2^{2^{d(k+1)}}}.
\]

Let $J = \{j_1 < \dots < j_r\} \subseteq [s]$ be the subset of indices defined by
\[
  J 
  \defeq 
  \big\{j \in [s] : h_1 \wedge \dots \wedge h_{j-1} \not\equiv h_1 \wedge \dots \wedge h_j\big\}.
\]
Note that whereas $s \le 2^{2^k}$, we have $r \le 2^k$ since this is the maximum length of a strictly decreasing sequence of nonzero functions $\{0,1\}^k \to \{0,1\}$ (assuming w.l.o.g.\ that $\ff G \not\equiv 0$).

At last, we choose $\ff G^\ast$ to be the formula $\sem{\ff h_{j_1},\dots,\ff h_{j_r}}_{\wedge}$. Observe that $\ff G^\ast$ is equivalent to $\ff G$ and has $\semempty{}$-depth at most $d$. Moreover, it is strict and satisfies
$\strict{\sem{\ff h_1,\dots,\ff h_s}_{\wedge}} = \ff G^\ast$.

Onto the probabilistic analysis.
We generate $\g \sim \scr D_t(\ff G)$ a balanced tree of binary $\wedge$-gates with subformulas $\g_{\mb\iota_1},\dots,\g_{\mb\iota_{tm}}$ feeding in for independent uniform random $\mb\iota_1,\dots,\mb\iota_{tm} \in [m]$ and $\g_i \sim \scr D_t(\ff G_i)$ ($i \in [m]$). Note that our choice of $t$ 
ensures that $\g_i \equiv \ff G_i$ for all $i \in [m]$ with probability $1-o(1)$. Let us henceforth automatically condition on this event in all instances of $\Pr[\,\cdot\,]$ below.

By the induction hypothesis and our choice of formulas $\ff h_1,\dots,\ff h_s$, for all indices $j \in [s]$ and $q \in [tm]$, we have
\[
  \Pr
  \big[\ 
    \strict{\g_{\mb\iota_q}} = \ff h_j
  \ \big|\ 
  \ff G_{\mb\iota_q} \tu{ computes } h_j
  \ \big]
  &\ge
  2^{-2^{d(k+1)}}
  \cdot
  \Pr
  \big[\ 
    \strict{\g_{\mb\iota_q}} = G_{\mb\iota_q}^\ast
  \ \big]\\
  &\ge
  2^{-2^{d(k+1)}} \cdot 2^{-\exp(O((d-1)2^{2^k}))}.
\]

Now consider the top $\log(tm)$ layers of $\g$, which we view as a random tree of binary $\wedge$ gates with inputs labeled by the functions computed at subformulas $\g_{\mb\iota_1},\dots,\g_{\mb\iota_{tm}}$ (elements of the set $\{h_1,\dots,h_s\}$). Call this tree $T(\g)$. 
Viewing $T(\g)$ as a join tree over single-edge graphs $E_1,\dots,E_s$ (or, alternatively, viewing $T(\g)$ as a DeMorgan formula over literals corresponding to functions $h_1,\dots,h_s$), Lemma \ref{la:eps1} (or from a different perspective Corollary \ref{cor:eps1}) tells us 
\[
  \Pr\big[\ 
  \underbrace{\strictjoin{T(\g)} 
  =
  \sem{h_1,\dots,h_s}_{\wedge}}_{\ts\tu{call this event }\mc E}\ \big]
  \ge
  2^{-c 2^s}
\]
for some absolute constant $c$ determined by Lemma \ref{la:eps1}.

Conditioned on event $\mc E$, there exist $2^s$ indices --- name them $1 \le q_1 < \dots < q_{2^s} < tm$ --- which embed the leaves of $\sem{h_1,\dots,h_s}_{\wedge}$ among the leaves of $T(\g)$. By definition of this sequence, formulas
$\g_{\mb\iota_{q_1}},\dots,\g_{\mb\iota_{q_{2^s}}}$ compute functions $h_1,h_2,h_1,h_3,\dots,h_1,h_2,h_1,h_s$.
Now observe that
\[
  \Pr\big[\ &\strict{\g} = \ff G^\ast\ \big|\ \mc E\ \big]\\
  \vphantom{\Big|}&=
  \Pr\Big[\ \strict{\g} = \strict{\sem{\ff h_1,\dots,\ff h_s}_{\wedge}}\ \Big|\ \mc E\ \Big]\\
  \vphantom{\Big|}&\ge
  \Pr\Big[\ (\strict{\g_{\mb\iota_{q_1}}},\dots,\strict{\g_{\mb\iota_{q_{2^s}}}}) = 
  (\ff h_1,\ff h_2,\ff h_1,\ff h_3,\dots,\ff h_1,\ff h_2,\ff h_1,\ff h_s)
  \ \Big|\ \mc E\ \Big]\vphantom{\Big|}\\
  \vphantom{\Big|}&=
  \Pr\Big[\ \strict{\g_{\mb\iota_{q_1}}} = \ff h_1 \ \Big|\ 
  \ff G_{\mb\iota_{q_1}} \tu{ computes } h_1\ \Big]
  \cdot
  \Pr\Big[\ \strict{\g_{\mb\iota_{q_2}}} = \ff h_2 \ \Big|\ 
  \ff G_{\mb\iota_{q_2}} \tu{ computes } h_2\ \Big]
  \cdot \, \ldots \mbox{}\\
  &\quad\ 
  \mbox{}\ldots \, \cdot
  \Pr\Big[\ \strict{\g_{\mb\iota_{q_{2^s-1}}}} = \ff h_1 \ \Big|\ 
  \ff G_{\mb\iota_{q_{2^s-1}}} \tu{ computes } h_1\ \Big] \cdot
  \Pr\Big[\ \strict{\g_{\mb\iota_{q_{2^s}}}} = \ff h_s \ \Big|\ 
  \ff G_{\mb\iota_{q_{2^s}}} \tu{ computes } h_s\ \Big]\\
  \vphantom{\Big|}&\ge
  \Big(2^{-2^{d(k+1)}} \cdot 2^{-\exp(O((d-1)2^{2^k}))}\Big){}^{2^s}.
\]

Finally, we obtain the desired bound
\[
  \vphantom{\Big|}
  \Pr\big[\ \strict{\g} = \ff G^\ast\ \big]
  \ge
  \Pr\big[\ \mc E\ \big] \cdot
  \Pr\big[\ \strict{\g} = \ff G^\ast\ \big|\ \mc E\ \big]
  &\ge
  2^{-c 2^s}
  \cdot 
  \Big(2^{-2^{d(k+1)}} \cdot 2^{-\exp(O((d-1)2^{2^k}))}\Big){}^{2^s}\\
  &\ge 
  \Big(2^{-c} \cdot 2^{-2^{d(k+1)}} \cdot 2^{-\exp(O((d-1)2^{2^k}))}\Big)\smash{{}^{2^{2^{2^k}}}}\\
  &\ge
  2^{-\exp(O(d2^{2^k}))}
\]
for an appropriate big-$O$ constant depending on the constant $c$.
\end{proof}

\subsection{Tradeoff (II)$^-$ for non-monotone $\ACzero$ formulas}\label{sec:II-}

We are ready to prove size-depth tradeoff (II)$^-$ for non-monotone $\ACzero$ formulas.  The proof has an overall similar structure to the proof of tradeoff (I)$^+$ in \S\ref{sec:I+}. 

\begin{thm}[Tradeoff (II)$^-$ of Theorem \ref{thm:AC0tradeoffs}]\label{thm:II-}
Suppose that $\ff F$ is a depth-$d$ $\ACzero$ formula 
which computes $\SPMM_{n,k}$ where $k \le \log^\ast n$
and $d \le \log k$.  Then $\ff F$ has size $n^{\Omega(dk^{1/2d})}$.
\end{thm}

By a simple additional argument, the same $n^{\Omega(dk^{1/2d})}$ lower bound may actually be obtained for $\ACzero$ formulas of depth $d+1$ (precisely matching the statement of tradeoff (II)$^-$ in Theorem \ref{thm:AC0tradeoffs}).  We present the lower bound for depth-$d$ $\ACzero$ formulas for simplicity sake, since the factor $2$ in the second exponent is anyway probably not optimal.  

\begin{proof}
Without loss of generality, assume that $\ff F$ has size at most $n^k$,\vspace{-1.5pt} since this is greater than the lower bound we wish to show. 
Let 
\[
  t \defeq 2^{2^{2^k}} \cdot (\log n)^4
\] 
and consider the randomized balanced DeMorgan conversion $\f \sim \scr D_t(\ff F)$.
Lemma \ref{la:Xi} implies that, with probability $1-o(1)$, the $G$-relations $\mc N_G(\f_0^{\cup\XI})$ are pathsets for all subformulas $\f_0$ of $\f$ and subgraphs $G \subseteq \Path_k$. We proceed with our analysis conditioned on this almost sure event.

For each $\f_0$ and $G$, we define a covering $\bigcup_{T \in \scr T_G} \mc N^{(T)}_G(\f_0^{\cup\XI}) = \mc N_G(\f_0^{\cup\XI})$ by essentially Definition \ref{df:MGT}, except that in the induction step where $T = \un{T_1}{T_2}$, we are now forced to treat the case $\f_0 = \f_1 \vee \f_2$ the same as $\f_0 = f_1 \wedge \f_2$, in both cases defining 
\[
  \mc N_G^{(T)}(\f_0^{\cup\XI})
  \defeq 
      \mc N_G(\f_0^{\cup\XI}) \cap 
      \Big(\mc N^{(T)}_G(\f_1^{\cup\XI}) \cup \mc N^{(T)}_G(\f_2^{\cup\XI})
      \cup
      \Big(\mc N^{(T_1)}_{G_1}(\f_1^{\cup\XI}) \bowtie \mc N^{(T_2)}_{G_2}(\f_2^{\cup\XI})\Big)
      \Big).
\]

By Lemma \ref{la:Mpath2}, $\mc N_{\Path_k}(\ff F^{\cup \XI})$
almost surely has density at least $\frac{1}{2n^2}$.
For each $\alpha \in \mc N_{\Path_k}(\ff F^{\cup \XI})$, we consider the restricted $k$-variable depth-$d$ $\ACzero$ formula 
\[
  \ff G_\alpha \defeq \ff F^{\cup\XI}{\uhr}\Path_k^{(\alpha)}
\]
obtained by $\ff F$ by fixing variables corresponding to $\XI$ to $1$, and all other variables except the $k$ edges of $\Path_k^{(\alpha)}$ to $0$.

We now invoke the key Lemma \ref{la:epsd} to obtain, for each $\alpha \in \mc N_{\Path_k}(\ff F^{\cup \XI})$, a strict DeMorgan formula $\ff G_\alpha^\ast$ of $\semempty{}$-depth at most $d$, which is equivalent to $\ff G_\alpha$ (hence depends on all $k$ variables) and satisfies
\[
  \Pr_{\g_\alpha \sim \scr D_t(\ff G_\alpha)}
  \big[\ \strict{\g_\alpha} = \ff G_\alpha^\ast\ \big]
  \ge
  2^{-\exp(O(d2^{2^k}))}.
\]
Let $T_\alpha^\ast$ 
denote the strict support tree of $\ff G_\alpha^\ast$:  
\[
  T_\alpha^\ast \defeq \strictjoin{\SuppTree(\ff G_\alpha^\ast)}.
\]
Note that $T_\alpha^\ast$ is a strict $\Path_k$-join tree since $\ff G_\alpha$ ($=\ff F^{\cup\XI}{\uhr}\Path_k^{(\alpha)}$) depends on all variables corresponding to edges of $\Path_k^{(\alpha)}$. Moreover, $T_\alpha^\ast$ has $\semempty{}$-depth at most $d$ (i.e.,\ at most the $\semempty{}$-depth of $\ff G_\alpha^\ast$) by Lemma \ref{la:semempty-depth}.

Next, we observe that the distribution of $\g_\alpha \sim \scr D_t(\ff G_\alpha)$ --- that is, $\g_\alpha \sim \scr D_t(\ff F^{\cup\XI}{\uhr}\Path_k^{(\alpha)})$ --- is {\em identical} to distribution of $\ff f{\uhr}\Path_k^{(\alpha)}$ where $\ff f \sim \scr D_t(\ff F^{\cup\XI})$.
Generating $\g_\alpha$ as $\ff f{\uhr}\Path_k^{(\alpha)}$, we have the following implication: 
\[
  \strict{\g_\alpha} = \ff G_\alpha^\ast \ \Longrightarrow\ 
\alpha \in \mc N^{(T_\alpha^\ast)}_{\Path_k}(\ff f).
\]

Since there are at most $2^{k^{d+1}}$ strict $\Path_k$-join trees of $\semempty{}$-depth at most $d$ (by Lemma \ref{la:counting}), we may fix one such join tree $T^\ast$ such that
\[
  \mu(\{\alpha \in \mc N_{\Path_k}(\ff F^{\cup \XI})
  : T_\alpha^\ast = T^\ast
  \})
  \ge
  \frac{1}{
  2^{k^{d+1}}} \cdot \mu(\mc N_{\Path_k}(\ff F^{\cup \XI}))
  \ge
  \frac{1}{2^{k^{d+1}} \cdot 2n^2}.
\]
It follows that
\[
  \mu(\mc N^{(T^\ast)}_{\Path_k}(\ff F^{\cup\XI}))
  \ge
  \frac{1}{2^{\exp(O(d2^{2^k}))} \cdot 2^{k^{d+1}} \cdot 2n^2}
  =
  \frac{1}{n^{2+o(1)}}.
\]

Recall that the random DeMorgan formula $\f \sim \scr D_t(\ff F)$ has depth at most $d\cdot\lceil\log(t\cdot\sz(\ff F))\rceil$ and size at most $t^d \cdot\sz(\ff F)$. Moreover, with probability $1 - o(1)$, functions $\ff F^{\cup\XI}$ and $\ff f^{\cup\XI}$ agree on all inputs of Hamming weight $\le k$ (by a union bound), hence
$
  \mc N^{(T^\ast)}_{\Path_k}(\ff F^{\cup\XI})
  =
  \mc N^{(T^\ast)}_{\Path_k}(\ff f^{\cup\XI}).
$

Since $\mc N_G(\f_0^{\cup\XI})$ are pathsets for all $\f_0$ of $\f$ and $G \subseteq \Path_k$,
Lemma \ref{la:halcali2}
yields an upper bound on pathset complexity:
\[
  \chi_{T^\ast}(\mc N^{(T^\ast)}_{\Path_k}(\ff F^{\cup\XI}))
  \vphantom{\Big|}&=
  \chi_{T^\ast}(\mc N^{(T^\ast)}_{\Path_k}(\ff f^{\cup\XI}))\\
  \vphantom{\Big|}&\le 
  \big(\depth(\f^{\cup\XI})+1\big){}^k
  \cdot \sz(\f^{\cup\XI})
  \\
  \vphantom{\Big|}&=  
  O\big(d\cdot\log(t\cdot\sz(\ff F))\big){}^k
  \cdot t^d \cdot\sz(\ff F)
  \\
  \vphantom{\Big|}&\le  
  O\big(d\cdot\log(2^{2^{2^k}} \cdot (\log n)^4 \cdot n^k)\big){}^k
  \cdot 2^{d2^{2^k}} \cdot (\log n)^{4d} \cdot \sz(\ff F)\\
  \vphantom{\Big|}&=
  n^{o(1)} \cdot \sz(\ff F).
\]
On the other hand, since $T^\ast$ has $\semempty{}$-depth at most $d$, Corollary \ref{cor:full-range}(II) provides the lower bound
\[
  \chi_{T^\ast}(\mc N^{(T^\ast)}_{\Path_k}(\ff F^{\cup\XI}))
  &\ge
  n^{\Omega(dk^{1/2d})} \cdot \mu(\mc N^{(T^\ast)}_{\Path_k}(\ff F^{\cup\XI}))
  \ge
  \frac{n^{\Omega(dk^{1/2d})}}{n^{2+o(1)}}.
\]
Combining these inequalities, we conclude that $\ff F$ has size $n^{\Omega(dk^{1/2d})}$ as required.
\end{proof}

\begin{rmk}\label{rmk:range}
This proof shows that the\vspace{1pt} range $k \le \log^\ast n$ can be extended to $k \le 0.99\log\log\log\log n$, small enough to ensure that $2^{\exp(O(d2^{\smash{2^k}}))} = n^{o(1)}$. We have made little effort to optimize the range of $k$ and expect it can be extended further, potentially to $\log\log n$.
\end{rmk}

\section*{Acknowledgements}

This paper was partially written during visits to the Tokyo Institute of Technology and the National Institute of Informatics. 
The author is grateful to the anonymous referees of STOC 2024 for valuable feedback on an earlier draft of this paper.

\bibliographystyle{abbrv}

\bibliography{bens-allrefs.bib}

\end{document}